\documentclass{report}
\usepackage[utf8]{inputenc}

\usepackage{amsmath}
\usepackage{amsfonts}
\usepackage{amssymb}
\usepackage{amsthm}
\usepackage{hyperref}
\usepackage[all]{hypcap}
\usepackage{mathtools}
\usepackage{enumerate}
\usepackage{enumitem}
\usepackage{booktabs}
\usepackage{comment}

\DeclareMathOperator*{\argmax}{arg\,max}
\DeclareMathOperator*{\argmin}{arg\,min}

\usepackage{tikz}
\usetikzlibrary {positioning}
\tikzset{small dot blue/.style={fill=blue!20,circle,scale=1.5}}
\tikzset{other node/.style={circle,fill=blue!20,draw,minimum size=1.5cm,inner sep=0pt},}
\tikzset{other node 1/.style={circle,fill=blue!15,draw,minimum size=1.5cm,inner sep=0pt},}
\tikzset{other node 2/.style={circle,fill=blue!30,draw,minimum size=1.5cm,inner sep=0pt},}
\tikzset{other node 3/.style={circle,fill=blue!50,draw,minimum size=1.5cm,inner sep=0pt},}
\tikzset{other node 4/.style={circle,fill=blue!70,draw,minimum size=1.5cm,inner sep=0pt},}
\tikzset{last node/.style={circle,fill=blue!20,draw,minimum size=1cm,inner sep=0pt},}
\tikzset{main node/.style={circle,fill=blue!20,draw,minimum size=2cm,inner sep=0pt},}

\tikzset{small dot/.style={fill=black,circle,scale=1}}

\tikzset{small dot new/.style={fill=blue!20,circle,scale=1.9}}

\tikzset{small dot 3/.style={fill=blue!20,circle,scale=1.25}}

\newcommand*{\prob}{\mathbb{P}}
\newcommand{\vect}[1]{\boldsymbol{#1}}

\newtheorem{prop}{Proposition}
\numberwithin{prop}{chapter}

\newtheorem{coro}{Corollary}
\numberwithin{coro}{chapter}
\newtheorem{rem}{Remark}
\numberwithin{rem}{chapter}
\newtheorem{defi}{Definition}
\numberwithin{defi}{chapter}
\newtheorem{lemma}{Lemma}
\numberwithin{lemma}{chapter}
\newtheorem{theo}{Theorem}
\numberwithin{theo}{chapter}

\DeclareMathOperator\supp{supp}

\usetikzlibrary{arrows}

\newcommand\twoheaduparrow{\mathrel{\rotatebox[origin=c]{90}{$\twoheadrightarrow$}}}

\newcommand\twoheaddownarrow{\mathrel{\rotatebox[origin=c]{270}{$\twoheadrightarrow$}}}

\usepackage{xcolor}
\usepackage{colortbl}

    \usepackage{fancyhdr}
    
\usepackage{pgfplots}
\pgfplotsset{compat=1.17}

\usepackage[toc,page]{appendix}

\usepackage{pdfpages}

\usepackage[absolute,overlay]{textpos}

\begin{document}

\begin{titlepage}



\begin{center}
\includegraphics[scale=0.3]{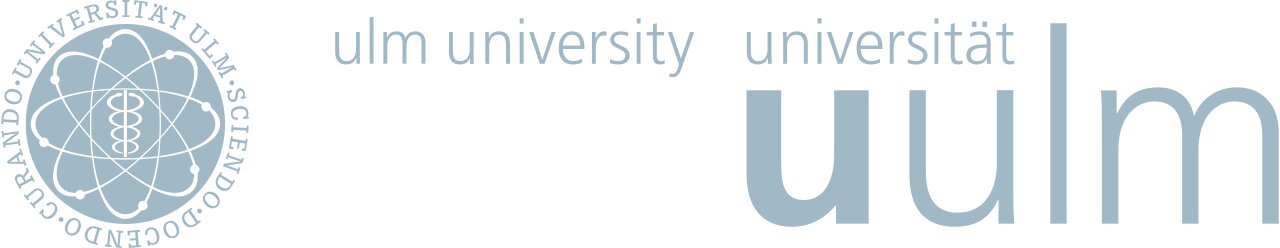}
\end{center}

\begin{textblock*}{12cm}(4cm,8cm) 
{\setlength{\parindent}{0cm}
{\fontsize{13}{13}\selectfont
\textbf{Universität Ulm} $\mid$ 89069 $\mid$ Germany
}
}
\end{textblock*}

\begin{textblock*}{8cm}(12cm,8cm) 
{\setlength{\parindent}{0cm}
{\fontsize{13}{13}\selectfont
\textbf{Fakultät für}

\textbf{Ingenieurwissenschaften,}

\textbf{Informatik}

\textbf{und Psychologie}

Institut für Neuroinformatik

Direktor: Prof. Dr. Dr. Daniel Braun
}
}
\end{textblock*}

\begin{textblock*}{12cm}(4cm,12.5cm)
\par\noindent\rule{\textwidth}{2pt}
\newline

{\setlength{\parindent}{0cm}
{\fontsize{18}{18}\selectfont \textbf{Order-theoretic models for decision-making: Learning, optimization, complexity and computation
}
}
}
\end{textblock*}

\begin{textblock*}{13cm}(4cm,16cm) 
{\setlength{\parindent}{0cm}
{\fontsize{13}{13}\selectfont
Kumulative Dissertation zur Erlangung des Doktorgrades

Doktor der Naturwissenschaften (Dr. rer. nat.)

der Fakultät für Ingenieurwissenschaften, Informatik und Psychologie

der Universität Ulm
}
}
\end{textblock*}

\begin{textblock*}{12cm}(4cm,21cm) 
{\setlength{\parindent}{0cm}
{\fontsize{13}{13}\selectfont
vorgelegt von

Pedro Hack

aus Buenos Aires (Argentinien)
\newline

Ulm 2023
}
}
\end{textblock*}

\newpage
\thispagestyle{empty}
\mbox{}
\newpage

\begin{textblock*}{12cm}(4cm,18cm) 
{\setlength{\parindent}{0cm}
{\fontsize{13}{13}\selectfont
   Amtierende Dekanin der Fakultät für Ingenieurwissenschaften, Informatik und
Psychologie: Prof. Dr. Anke Huckauf
\newline

Gutachter: Prof. Dr. Dr. Daniel A. Braun

Gutachter: Prof. Dr. Gianni Bosi

Gutachter: PD Dr. Friedhelm Schwenker
\newline




Tag der Promotion: 09.02.2024
}
}
\end{textblock*}

\end{titlepage}

\newpage\null\thispagestyle{empty}\newpage

\newpage
\thispagestyle{empty}
\mbox{}
\newpage

\chapter*{Abstract}
\addcontentsline{toc}{chapter}{Abstract}

The modern agent-based study of intelligent systems explains the behaviour of agents in terms of economic rationality, that is, it intends to justify the transitions that systems undergo as an attempt to optimally achieve a goal. This results in an optimization principle, typically involving a real-valued function or utility, which states that, in the long run, the system will evolve until the configuration of maximum utility is achieved. Recently, this theory has evolved into bounded rationality, where constraints are incorporated into the picture. In particular, instead of simply predicting utility maximization in the limit, the optimum is achieved when the utility is maximized while respecting some information-processing constraints. Abstractly, these constraints can be formulated as uncertainty bounds. Hence, alternatively, the principle can be reformulated as an attempt from the system to maximize uncertainty while achieving a certain utility. This fundamental duality is reminiscent of thermodynamic systems, where the long-run behaviour is explained by the maximization of uncertainty subject to some constraint on the energy function.
As such, the study of intelligent systems has not only benefited from the tools of equilibrium thermodynamics, but, even more recently, also from achievements in non-equilibrium thermodynamics,
where not only the infinite time limit is considered, but also regularities in the dynamics are exploited to obtain predictions that involve short-run realizations of the system, the so-called fluctuation theorems. The first aim of this thesis is to clarify the applicability of these results in the study of intelligent systems,
in particular in the context of human sensorimotor adaptation

To understand the structural similarity between thermodynamic and intelligent system models on a more abstract level, we can think of the local transition steps for both kinds of system as being driven by uncertainty. In fact, it turns out that the transitions in both systems can be described in terms of a preorder known as majorization.
Hence, real-valued uncertainty measures like Shannon entropy, which constitute the basis of their respective optimization principles, are simply a proxy for their more involved behaviour. More in general, real-valued functions are fundamental in the study of optimization and complexity in the order-theoretic approach to several topics,
including economics, thermodynamics, and quantum mechanics,
and are
key to understanding the structural similarity between
them.
The second aim of this thesis is to improve on the classification of preorders in terms of real-valued functions that can be applied i.a. to the majorization preorder.

The basic structural similarity between thermodynamic and intelligent system models exposed before is based on the notion of uncertainty expressed by a special preorder. From a computer science perspective, we can think of the transitions in the steps of a computational process as a decision-making procedure where uncertainty is reduced, like a computation that converges towards the decimal representation of a real number by producing a sequence of (longer) finite substrings of it. In fact, by adding some requirements on the considered order structures, we can build an abstract model of uncertainty reduction that allows to incorporate computability, that is, to distinguish the objects that can be constructed by following a finite set of instructions from those that cannot. The third aim of this thesis is to clarify the requirements on the order structure that allow such a framework.

 In summary, in this work, we take an abstract approach to decision-making, which we equate with the reduction of uncertainty and we model using order structures. We start by considering a decision-theoretic approach to learning for intelligent systems. The optimization tools common in this field lead us to the study of the relation between real-valued functions and order structures, which include not only optimization but also different notions of complexity for decision-making systems. We finish by considering the notions of computation that can be derived from these order-theoretic approaches to decision-making.

\newpage
\thispagestyle{empty}
\mbox{}
\newpage
\thispagestyle{empty}
\mbox{}
\newpage

\chapter*{Acknowledgements}
\addcontentsline{toc}{chapter}{Acknowledgements}

I am definitely indebted to several people in Ulm for their help during the process that led to this piece of work. 
I would like to start by thanking my supervisor Daniel Braun for giving me this opportunity and the necessary support for conducting this research. In this regard, I would also like to extend my gratitude to the European Research Council for making this research possible.
Sebastian Gottwald also helped me greatly throughout the course of my stay at Ulm. I am thankful for that and also for the fun discussions on the different projects we have worked together in.
I would like to thank Cecilia Lindig-León for helping me with the data recording and, more in general, for her cheerful attitude and for giving me the chance to talk Spanish constantly during my stay at Ulm, which definitely made me feel more at home.
I would also like to acknowledge Peter Bellmann's relaxed and fun presence. I have definitely enjoyed talking to him, specially during the times when Covid was everywhere. I felt lucky every time I would find him is his office when many people worked remotely. 

I would like to continue by thanking the supervisors I had before starting my postgraduate studies.
I am definitely indebted to Andreas Winter. I learned a lot by listening to him and trying to approach research the way he does.
I would also like to acknowledge the help and fantastic supervision by Helmut Schmidt during my undergraduate internship at the \emph{Centre de recerca matemàtica}, where Alex Roxin kindly hosted me. 

It would be unfair to leave out the people who inspired me at the \emph{Universitat Autònoma de Barcelona}, where I completed my undergraduate studies.
The first was Julià Cufí, who introduced us to calculus. The class was really disorganized and one would often find him stopping to prove some lemma he forgot before. Many people complained about it, but I enjoyed watching him all over the place. By far, this was the most inspiring class I have ever taken. The next great teacher I encountered was Juan Camacho, who taught us thermodynamics during my second undergraduate year. The organization of the curriculum changed that year, so in my class there were only the few students in the double degree program in mathematics and physics. Hence, I believe Juan taught us the course in quite an unusual way,
introducing the subject via equivalence classes, for example. Such a simple and fundamental approach really captivated me. In fact, my interest in thermodynamics has definitely lead me through several of the topics that constitute this piece of work.
It was also during my second year when I attended
Francesc Perera's lecture on algebraic structures. The class was great. Francesc introduced us to order theory, which has definitely occupied me the most along this work. I do remember a fantastic class where
Francesc told us about the axiom of choice. There is definitely something wonderful about considering foundational questions.
It was also during my second year that I met Susana Serna in a numerical analysis class. I would like to thank Susana for her encouragement and help during my last years as an undergraduate. I would also like to extend my gratitude to Artur Nicolau, from whom I learned not only harmonic analysis but several other things, like how one ought to prepare for teaching a class. Artur is truly a wonderful person and someone to look up to. 

I would like to thank \emph{Colegio 
Montesión} in Palma for helping me improve as a person. It was also there where both Pere Mateu and Rafel got me interested in physics and science in general.
I also do not dare to forget the other teachers I encountered during my life, specially those from whom I learned endurance, which was key in the completion of this piece of work. I do treasure beautiful moments with both Oscar and Joaquín, who taught me, respectively, chess and basketball when I was a kid.

Lastly, and most importantly, I am thankful to my family for their constant support. I would like to thank my mother for her tenderness and for being always there when I needed her, and
my father for all his effort and dedication, and for being an inspiration to me. I am indebted to both my sisters for being so warm and encouraging, they always bring me joy. I do have a lot to learn from them. I would also like to thank Rebecca for the light she irradiates and for the beautiful way in which she looks at the world, it is truly a joy to see her do the most common of things. 
I would like to finish by acknowledging the people who have gifted me with their honest friendship and, lastly, my family back in Argentina.

\newpage
\thispagestyle{empty}
\mbox{}
\newpage
\thispagestyle{empty}
\mbox{}
\newpage

\chapter*{Notation}
\addcontentsline{toc}{chapter}{Notation}

\begin{tabular}{ c | l}
$p^*$ & Boltzmann distribution\\
$\perp$ & bottom element in partial order\\
$\langle \cdot, \cdot \rangle$ & Cantor pairing function\\
$\preceq_C$ & Cantor partial order\\
$\times$ & Cartesian product \\
$K(P)$ & compact elements of dcpo $P$\\
$\preceq_d$ & $d$-majorization\\
$\text{dom}(\cdot)$ & domain of a map\\
$\Pi_d$ & $d$-stochastic matrix \\
$[\cdot]$ & equivalence class \\
$E$ & energy function\\
$\mathbb E$ & expected value\\
$\Sigma$ & finite alphabet\\
$\alpha$ & finite map\\
$\Omega$ & finite set\\
$F$ & free energy\\
$\mathcal F$ & generalized entropies majorization multi-utility\\
$D_{KL}$ & Kullback-Leibler divergence\\
$\preceq_L$ & lexicographic partial order\\
$\lim_{n \to \infty} \preceq_n$  & limit partial order sequence $(\preceq_n)_{n \geq 0}$ \\
$\preceq_M$ & majorization preorder\\
$\mu$ & map inducing the Scott topology everywhere\\
$\Lambda_d$ & map that defines $d$-majorization\\
$min(P)$ & minimal elements of partial order $P$\\
$\mathcal U$ & multi-utility that typically defines majorization\\
$\mathbb N$ & natural numbers (including zero)\\
$[0,\infty)^{op}$ & non-negative reals equipped with reversed ordering \\
$\mathcal N_B$ & non-trivial elements of basis $B$\\
$\sim$ & order equivalence \\
$\bowtie$ & order indifference\\
$\otimes$ & outer product \\
$\sqcup$ & partial order supremum\\
$Z$ & partition function\\
$\preceq$ & preorder\\
\end{tabular}

\newpage

\begin{tabular}{ c | l}
$\rho^F$ ($\rho^B$) & probability density forward (backward) process\\
$\mathbb P_\Omega$ & probability distributions over $\Omega$ \\
$\preceq_{\mathcal I}$ & reversed inclusion partial order\\
$\sigma(P)$ & Scott topology over partial order $P$\\
$\mathcal I$ & set of compact intervals of the real line\\
$H$ & Shannon entropy\\
$\Sigma^*$ & set of finite words over $\Sigma$\\
$\Sigma^\omega$ & set of sequences over $\Sigma$\\
$\Sigma^\infty$ & $\Sigma^* \cup \Sigma^\omega$\\
$\leq$ & standard order on the real line\\
$\prec$ & strict order preference \\
$\supp(\cdot)$ & support of a function \\
$\preceq_T$ & trumping preorder\\
$\preceq_U$ & uncertainty preorder\\
$\ll$ & way-below relation\\
$W$ & work\\
\end{tabular}

\newpage
\thispagestyle{empty}
\mbox{}
\newpage
\thispagestyle{empty}
\mbox{}
\newpage

\chapter*{Disclaimer:\\ How to read this thesis}
\addcontentsline{toc}{chapter}{Disclaimer: How to read this thesis}

The main part of this thesis is a self-contained document that can be read independently. The reader interested in the details can consult the appendices, where we include a copy of the papers that complement the main part. If a result is simply stated in the main part, then its proof can be found in the appendix papers or elsewhere in the literature. (In the latter case, a citation is included in the statement.) The statements in the main part for which a proof is provided supplement the appendix papers.

\newpage
\thispagestyle{empty}
\mbox{}
\newpage
\tableofcontents
\newpage
\thispagestyle{empty}
\mbox{}
\newpage

\pagestyle{fancy}
\renewcommand{\chaptermark}[1]{\markboth{#1}{#1}}
\fancyhead[R]{}
\fancyhead[L]{\chaptername\ \thechapter\ --\ \leftmark}

\chapter{Introduction}
\label{intro}

The fundamental notion that ties this work together is that of uncertainty. More precisely, we consider uncertainty reduction, that is, decision-making, both abstractly and in some concrete applications. As a result, we deal with several aspects of decision-making like learning, optimization, complexity and computation.  Our interest in this topic can be divided into three clear parts. In the first part, we consider the influence of uncertainty in the study of intelligent systems. A key observation we make is that, although the fundamental notion of uncertainty that underlies the study of intelligent systems takes the form of an order relation, in practice the most predominant tool to measure uncertainty is Shannon entropy. This replacement of the order relation by a function is, however, imperfect, in the sense that a function is not able to represent the uncertainty relation faithfully in general. This leads us to the second part of this work, where we examine the properties of such a transition from an order relation to a function in general and that of the uncertainty relation to Shannon entropy in particular. More precisely, we examine how well can order relations be captured using a single or several real-valued functions. That is, we approach both the existence of optimization principles and the quantification of complexity for order relations. Lastly, in the third part of this work, we take an uncertainty-based approach to computation. In particular, given that one can consider any computation as a process of uncertainty reduction, we take an approach to computation that fundamentally relies on order structures that represent relative uncertainty relations between computational elements.  

In this chapter, we provide a brief introduction to the three parts that compose this work, emphasizing the questions we address throughout this report and their relation to the existing literature.

\section{Uncertainty and learning systems}

The key aspect where uncertainty plays a role in the study of intelligent agents is \emph{learning}. Learning is an attribute of an intelligent agent, the \emph{system}, that can take several configurations and is embedded in an environment. Learning comes into play since the environment may change and, with it, the most suitable configuration the system should adopt. Hence, learning refers precisely to the ability the system may have to recognize the changes in the environment and, moreover, adapt to them by changing its configuration. 

Historically, the study of learning systems can be thought to have followed a path similar to that of classical mechanics. In fact, in both these fields, a single real-valued function has been considered as the driving force behind the system's evolution. This has resulted in well-established optimization principles for both of them, namely (among others) potential minimization for classical mechanics \cite{goldstein1950classical} and utility maximization for learning systems \cite{fishburn1970utility,morgenstern1953theory}. The progress in classical mechanics resulted in thermodynamics, where the impossibility of keeping track of the large number of individual entities that form a system (for example, a gas) led to the inclusion of a counteracting force, namely, uncertainty. In consequence, the optimal behaviour of a thermodynamic system is probabilistic. In a similar manner, uncertainty entered the study of learning systems to model constraints (in particular, \emph{information-processing} constraints) in the adaptation of such systems to an environment. Hence, despite having a different origin than its thermodynamic counterpart, the inclusion of uncertainty in the study of learning systems resulted in their optimal behaviour also being probabilistic.
More specifically, both these fields are based on the same optimization principle, namely, that the distribution over the possible states of the system in the long run results from the maximization of Shannon entropy (which is a measure of uncertainty) over a region of distributions that fulfill some constraint on a real-valued function (the \emph{utility} or \emph{energy} function, respectively).

\begin{rem}[Bounded rationality]
    It should be noted that we  refer to the uncertainty-based approach to learning systems as \emph{bounded} rationality (following \cite{wolpert2006information} or \cite{ortega2013thermodynamics}, for example), as opposed to the approach that only considers a utility function, known as (perfect) \emph{rationality} \cite{fishburn1970utility,savage1972foundations}. In fact, what we call throughout this work \emph{bounded rationality} is only a specific branch where the considered bounds have an \emph{information-theoretic} origin. Bounded rationality refers, in general, to the abandonment of perfect rationality, which can take several forms \cite{simon1990bounded,simon1997models}.
\end{rem}

While several works have pointed out the use of optimization principles from thermodynamics in bounded rationality \cite{ortega2013thermodynamics,genewein2015bounded}, it was only recently that the \emph{fluctuation theorems} \cite{seifert2012stochastic}, a series of results in (non-equilibrium) thermodynamics where the long-run (i.e. optimal) behaviour is related with short-run (i.e. non-optimal) realizations, were put forward in the context of learning systems \cite{grau2018non}. The main distinction between \cite{grau2018non} and previous work in bounded rationality is that the latter only considered the fully adapted or optimal behaviour while \cite{grau2018non} assumed a model for the adaptation of the system. In particular, the assumption that the adaptation follows a stochastic process without memory, that is, a Markov chain \cite{levin2017markov}. Following this approach, the first purpose of this thesis is to examine the validity of the fluctuation theorems for learning systems. In particular, we consider the following research questions:  

\begin{enumerate}[label=\textbf{(Q\arabic*)}]
\item \textbf{Under what conditions do the fluctuation theorems hold for learning system?}
\item \textbf{Can these results be observed in the adaptation of biological learning systems? In particular, in human sensorimotor adaptation?}
\end{enumerate}



It should be noted that the novelty of \textbf{(Q1)} lies in its detachment of the fluctuation theorems from physics, which leads us to a reexamination of the framework and necessary assumptions.
Regarding \textbf{(Q2)}, we perform the first experimental test of the fluctuation theorems in the context of human learning. In fact, this constitutes the first time they are challenged outside the realm of physics, where they have been successfully tested (see, for example, \cite{collin2005verification} or \cite{douarche2005experimental}). This concludes the introduction to the first part of this work. We proceed now, in the following section, to an overview of the second part.

\section{Real-valued representations of preorders}
\label{intro: real repre}

In both thermodynamics and the study of learning systems, the key idea is the trade-off between uncertainty and utility (or energy). Moreover, they both represent uncertainty through a specific measure, namely, Shannon entropy. Nonetheless, the fundamental idea of uncertainty is more complex, as we expose in the next paragraphs.

Assume we have a certain system that behaves as a random variable over a finite set $\Omega$ and, thus, whose configuration can be described through a distribution $p$ in $\mathbb P_\Omega$, the set of probability distributions over $\Omega$. Shannon entropy $H(p)\coloneqq -\mathbb E_p[\log p]$ \cite{shannon1948mathematical,cover1999elements}, where $\mathbb E_p$ denotes the expected value over $p$, is used in the optimization principles of both thermodynamics and bounded rationality since it ranks the distributions in terms of how uncertain they are, that is, if $p$ is less biased than $q$, then $H(p)<H(q)$. But, what does it mean do be uncertain? We can think of it in terms of gambling. In particular, in terms of the preferences of a casino owner among games.

Consider a casino owner that needs to choose a new game for the casino. All games under consideration follow the same pattern: bets are placed before a realization of a random variable and, if the outcome is properly predicted, the gambler wins. Hence, the casino owner prefers games that are more uncertain. For example, rolling a fair die is preferred to rolling a biased one. Since each gambler has the option to bet on any proper subset of the options, when comparing games, the casino owner needs to consider the best edges a gambler could have when following any betting strategy in each of them. In particular, what is the best edge a gambler could achieve when betting on a single outcome, on a pair of outcomes, and so on?

In general, when betting on $i$ outcomes ($i < |\Omega|$), the best chance of winning a gambler could face when choosing the outcomes appropriately is the sum of the probabilities of the 
$i$ most likely outcomes. Hence, the casino owner prefers a game with outcome distribution $q$ over one with outcome distribution $p$, which we denote by $p \preceq_U q$ for all $p,q \in \mathbb P_\Omega$, if, for all $i< |\Omega|$, the sum of the probabilities of the 
$i$ most likely outcomes according to $q$ is lower or equal to the sum for those of $p$. That is, the following holds:
\begin{equation}
\label{uncert rela}
    p \preceq_U q \ \iff \ u_i(p) \leq u_i(q) \ \  \forall i\in \{1,..,|\Omega|-1\} \, ,
\end{equation}
where $u_i(p) \coloneqq -\sum_{n=1}^{i} p_n^{\downarrow}$ and $p^{\downarrow}$ denotes the non-increasing rearrangement of $p$ (same components as $p$ but ordered in a non-increasing fashion). Thus, \eqref{uncert rela} ranks the distributions precisely how the casino owner would. It should be noted that the same intuitive picture involving gambling and the uncertainty preorder has been recently introduced in \cite{brandsen2022entropy}.

Another intuitive way of thinking about the uncertainty preorder has its origin in the study of wealth distribution and is based on the so-called \emph{Pigou-Dalton} transfers \cite{dalton1920measurement,pigou2017economics,gottwald2019bounded}. This point of view offers another (equivalent) definition of $\preceq_U$ where, basically, we say that $p \preceq_U q$ for any pair $p,q \in \mathbb P_\Omega$ if and only if $q$ is the result of applying to $p$ a finite number of $(i)$ permutations (where we simply interchange the probability masses of two elements in $\Omega$) and $(ii)$ transfers of probability mass from a more likely to a less likely element in $\Omega$. While permutations are innocuous form the point of view of uncertainty and bias (they simply relabel the elements that constitute $\Omega$), it is the latter mechanism what allows the probabilities of the different elements in $\Omega$ to become closer to each other and, hence, the probability distribution to approach a uniform distribution $u \coloneqq (1/|\Omega|,\dots,1/|\Omega|) \in \mathbb P_\Omega$. It is in this sense that the bias is reduced or, equivalently, the uncertainty increases. A simple case where the uncertainty increase is exemplified can be found in Figure \ref{fig: example decision-making}.

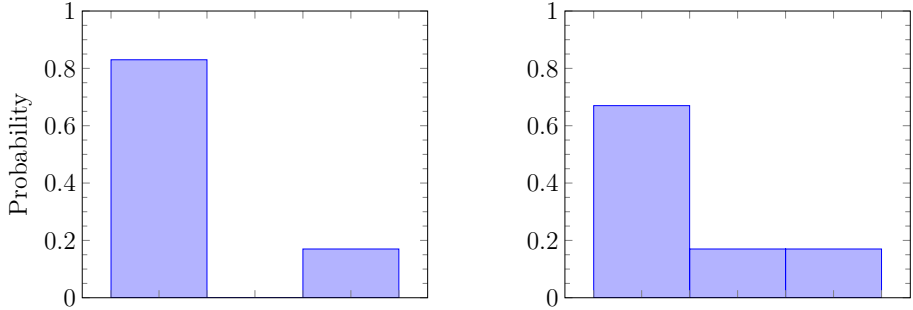
\begin{figure}
    \centering
    \begin{tikzpicture}[scale=0.67,font=\fontsize{15}{15}\selectfont]
\begin{axis}[
    ymin=0, ymax=1,
    minor y tick num = 3,
    area style,
    xticklabels={,,},
    ylabel={Probability}
    ]
\addplot+[ybar interval,mark=no] plot coordinates { (0, 0.83) (1, 0) (2, 0.17) (3,0.17)};
\end{axis}
\end{tikzpicture}
\hspace{1cm}
    \begin{tikzpicture}[scale=0.67,font=\fontsize{15}{15}\selectfont]
\begin{axis}[
    ymin=0, ymax=1,
    minor y tick num = 3,
    area style,
    xticklabels={,,}
    ]
\addplot+[ybar interval,mark=no] plot coordinates { (0, 0.67) (1, 0.17) (2, 0.17) (3,0.17)};
\end{axis}
\end{tikzpicture}
\caption{Simple example where the left distribution $p$ is more biased than the right distribution $q$, $p \preceq_U q$, and $q$ is obtained from $p$ by transferring probability mass from the first to the second element in $\Omega$.}
    \label{fig: example decision-making}
\end{figure}

Returning to our gambling picture, it should be noted that a game without bias, i.e. one whose associated distribution is the uniform distribution, is preferred over all other games. Moreover, note that, although we have referred to it as the \emph{uncertainty preorder} in \cite{hack2022representing}, $\preceq_U$ (or, rather, its dual $\preceq_M$, that is, the preorder $\preceq_M$ such that, for all $p,q \in \mathbb P_\Omega$,  $p \preceq_M q$ if and only if $q \preceq_U p$) is known in mathematics \cite{marshall1979inequalities}, economics \cite{arnold2018majorization}, and quantum physics \cite{nielsen1999conditions} as \emph{majorization}. In case it is clear from the context, we may use both terms interchangeably. Notice, importantly, that $\preceq_U$ will be a key example throughout this work.

\begin{rem}[Shannon entropy and the uncertainty preorder]
\label{shannon entropy and uncertainty}
As the preferences of the casino owner exemplify, the uncertainty preorder encompasses an intuitive picture of uncertainty. In order to see how Shannon entropy fits in this picture, we can consider Shannon's original derivation \cite{shannon1948mathematical}, which consists (roughly) of the following assumptions: (i) continuity, (ii) the increase of entropy for uniform distributions with a larger number of non-zero components, and (iii) a certain additivity property. Although (i) and (iii) are certainly desirable properties, we believe they have no intuitive justification in terms of uncertainty. The connection to uncertainty is in (ii), given that, if $u,u' \in \mathbb P_\Omega$ are uniform distributions with $u'$ having more non-zero components than $u$, then $u \preceq_U u'$ and Shannon entropy reflects that increase in uncertainty $H(u) \leq H(u')$. In fact, Shannon entropy fulfills a stronger property, namely, it is monotonic in $\preceq_U$, that is, $p \preceq_U q$ implies $H(p) \leq H(q)$ for all $p,q \in \mathbb P_\Omega$. However, this intuitive property is shared by several other functions (see \cite{gottwald2019bounded} and references therein). Hence, regarding the representation of $\preceq_U$, we believe there is no argument that favors Shannon entropy over other functions that are monotonic in $\preceq_U$. Moreover, while 
one can compare any pair of distributions in terms of Shannon entropy, this is not the case with $\preceq_U$. As a result, neither Shannon entropy nor any other function alone can represent the preferences of the casino owner faithfully. In conclusion, although useful, the optimization principles at the core of both thermodynamics and bounded rationality are simply a proxy for their more involved behaviour. The interested reader can find information concerning the proposals alternative to entropy inside of physics in, for example, \cite{tsallis1988possible,tsallis2009introduction}.
\end{rem}

The idea of substituting the preferences or transitions of a system by real-valued functions is, of course, not a special feature of majorization. On the contrary, it is present is several branches of knowledge, including
\emph{thermodynamics} \cite{lieb1999physics,giles2016mathematical}, \emph{general relativity} \cite{bombelli1987space,minguzzi2010time}, \emph{quantum physics} \cite{nielsen1999conditions,brandao2015second} and \emph{economics} \cite{debreu1954representation,ok2002utility}.
In fact, several key questions are shared among some of these areas, as was realized, for example, in \cite{candeal2001utility,campion2018survey}. Hence, it is in the interest of a large number of fields to clarify the extent to which real-valued functions can be considered instead of the order structures to describe the dynamics of the systems they are interested in. In fact, the study of real-valued representations of preorders, i.e. order structures that possess some properties that all the aforementioned fields share, has a large tradition that is closely related to mathematical economics \cite{debreu1954representation,bridges2013representations}.

In order to determine what a preorder is, let us consider two properties of majorization: 
$(i)$ any game is preferred to itself $p \preceq_U p$ and $(ii)$, if a game is preferred over another $p \preceq_U p'$ and a third one is preferred over the first one $p' \preceq_U p''$, then the third one is preferred over the second one $p \preceq_U p''$. Hence,
the first assumption allows the system to remain in the state it is in (to follow a trivial transition) and the second to concatenate transitions. These properties are known as reflexivity and transitivity, respectively, and are exactly the properties that define the binary relation named preorder \cite{bridges2013representations}.\footnote{However, these properties are not shared among all transitions systems of interest. An example of this can be found in the study of chemical processes \cite{campion2016entropy}. In fact, preorders have been also challenged in the literature that is concerned with intelligent systems \cite{armstrong1948uncertainty,fishburn1970intransitive}.}

In the second part of this work, thus, the fundamental object we consider is a preordered space $(X,\preceq)$, that is, a set $X$ of possible states a system can be in (called the \emph{ground} set \cite{harzheim2006ordered}) and a preorder $\preceq$ that models the transitions among the possible states the system can perform. Given the variety of fields that share the same concerns and the fact that we are interested in considering order-theoretic approaches to uncertainty that go beyond learning systems and majorization (as we will detail in Section \ref{sec: uncert and comp}), we address the following general research question in the second part:
\newline

\begin{enumerate}[label=\textbf{(Q\arabic*)}]\addtocounter{enumi}{2}
\item \textbf{Given the set of possible transitions of a system, how well can we encode it via measurement devices (i.e. real-valued functions)?}
\end{enumerate}



The main novelty of \textbf{(Q3)} rests in the fact that we introduce new sorts of real-valued functions related to preorders and, moreover, we clarify the relation among the existence of several descriptions of preorders via real-valued functions. In this regard, we expose the richness of both complexity and optimization, and that of their relation.
Lastly, we apply our general framework to gain insight into learning systems. In particular, we clarify the representation of the  uncertainty preorder by real-valued functions, which is key in the uncertainty-based bounded rationality approaches to learning systems. As other applications of our general considerations, we also improve on some physical models. This concludes the introduction to the second part of this work. We proceed now, in the following section, to an overview of the third part.

\section{Uncertainty and computation}
\label{sec: uncert and comp}

The fundamental notion connecting learning systems in the view of bounded rationality and thermodynamics is uncertainty, which is modelled by the majorization preorder. Computation can also be modelled using uncertainty. In fact, we can think of the individual steps that take place in some computational process as uncertainty reduction. As we will see, this notion of uncertainty reduction can be modeled through a preorder as well and, moreover, it allows the introduction of a notion of computability on uncountable spaces, that is, a distinction between elements with a \emph{finite description} and those without one. This point of view where computation is considered as a decision-making process also justifies the more abstract approach to order structures that we developed in Section \ref{intro: real repre} from majorization, given that we also wish our work there to apply to computation. 

\begin{rem}[Computability and computation]
    It should be noted that computation and computability refer to different concepts. Roughly speaking, \emph{computation} refers to the process of calculating something while \emph{computability} points to the more fundamental aspect of whether there is some computation leading to that something. However, for simplicity, we will use the terms interchangeably, with the hope that what is meant is clear from the context.   
\end{rem}

In order to clarify the first paragraph in this section, let us consider an example.
Assume, for instance, we aim to calculate some real number $x \in \mathbb{R}$ using an algorithm. In this context, roughly, to \emph{calculate} refers to the possibility of finding a finite set of instructions that, when carried out repeatedly, allow us to determine, up to any desired precision, the element inside of certain representation model, e.g., having a program that outputs (in finite time) each of the digits that conform its binary representation. Consider, thus, that the representation we use consists of the compact intervals on the real line $\mathcal I \coloneqq \{[a,b] \, | \, a,b \in \mathbb R \text{ and } a \leq b\}$. These intervals carry a natural preorder for our purpose, namely, the one given by reversed inclusion $[a,b] \preceq_{\mathcal I} [c,d]$ if $a \leq c$ and $d \leq b$. This preordered space can be used to model any computational processes that leads to $x$ by starting with some interval that $x$ belongs to and sequentially reducing its length, allowing us to reach arbitrarily small intervals that contain $x$.

An example of a computational process that can be modelled through $(\mathcal I,\preceq_{\mathcal I})$ is the bisection method, which is a simple algorithm used to calculate zeros of functions. Given, say, a polynomial with rational coefficients $p: \mathbb R \to \mathbb R$ and a pair of rational numbers $q,q' \in \mathbb{Q}$ where the polynomial changes sign $p(q)p(q')<0$, the algorithm operates in the following way for any $\varepsilon >0$:
\begin{enumerate}[label=(\roman*)]
\item Evaluate $p$ at $\frac{q+q'}{2}$.
\item If $p(\frac{q+q'}{2})=0$, return $\frac{q+q'}{2}$.
\item If $p(q)p(\frac{q+q'}{2})<0$, then $q' \coloneqq \frac{q+q'}{2}$. Otherwise, $q \coloneqq \frac{q+q'}{2}$.
\item Repeat until $|q-q'|< \varepsilon$.
\end{enumerate}
(We include a simple graphical representation of the bisection algorithm in Figure \ref{bisection method fig}.)

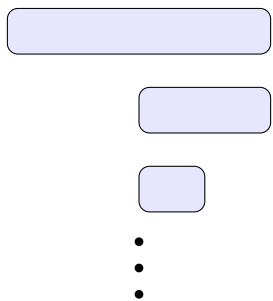
\begin{figure}[!tb]
\centering
\begin{tikzpicture}[scale=0.35, every node/.style={transform shape}]
\node[rounded corners, draw,fill=blue!10, text height = 1.5cm, minimum width = 10cm] {};
\node[rounded corners, draw,fill=blue!10, text height = 1.5cm, minimum width = 5cm, yshift=-3cm,xshift=2.5cm] {};
\node[rounded corners, draw,fill=blue!10, text height = 1.5cm, minimum width = 2.5cm, yshift=-6cm,xshift=1.25cm] {};
\node[small dot] at (0,-8) {};
\node[small dot] at (0,-9) {};
\node[small dot] at (0,-10) {};
\end{tikzpicture}
\caption{Potential intervals that can be obtained when applying the bisection method on some polynomial with rational coefficients $p: \mathbb R \to \mathbb R$ and using $[q,q']$, consisting of rational numbers $q,q' \in \mathbb{Q}$ such that $p(q)p(q')<0$, as initial interval.}
\label{bisection method fig}
\end{figure}

We can consider, thus, the action of the bisection method on arbitrary functions as a series of transitions in the space of rational compact intervals that connect two different pairs of intervals: $(a)$ those intervals that share an endpoint and where the second endpoint of the shorter one is the midpoint of the endpoints of the longer one and $(b)$ each interval $I$ to $[m_I,m_I]$, where $m_I$ is the midpoint of $I$. More in general, we can think of the transitions in several sorts of computations that aim to find some real number as going from some interval with rational points to some other interval with rational points contained in the first one. These transitions belong, hence, to the ones allowed by $\preceq_{\mathcal I}$ on $\mathcal I$. 



Under our definition of computation, any $x$ that is a root of a rational polynomial can be computed, since (removing, if necessary, the precision constraint $\varepsilon$) the bisection method is a procedure that allows us to localize $x$ to arbitrary precision. Moreover, any algorithm that follows a similar \emph{localization} strategy can be thought of as a finite set of instructions that output a set of rational intervals such that $x$ is contained in all of them. Hence, $(\mathcal I,\preceq_{\mathcal I})$ not only captures a notion of uncertainty reduction, but it is a framework in which we can distinguish between the elements that can be computed following the same uncertainty reduction strategy and those that cannot. In particular, as we did with the bisection method and $x$, we can introduce computability on the real numbers through $(\mathcal I,\preceq_{\mathcal I})$.      

Since, as we discussed above, order structures like $\preceq_{\mathcal I}$ can be used to abstractly model the uncertainty reduction that constitutes the basis of several algorithms and, moreover, to introduce computability on uncountable spaces, the third part of this thesis is concerned with the study of the order-theoretic properties that are necessary in these models. In particular, we address the following research question:  

\begin{enumerate}[label=\textbf{(Q\arabic*)}]\addtocounter{enumi}{3}
\item \textbf{Can we clarify what order-theoretic structure is required in order to introduce computability on uncountable sets?}
\end{enumerate}



The main novelty of \textbf{(Q4)} lies in its attempt to derive the order-theoretic framework from an intuitive picture, achieving, thus, a more general structure that allows the introduction of computability on uncountable sets. As part of this process, we find specific reasons supporting the stronger assumptions that are usually required in similar approaches in the literature.

Following our interest in the requirements underlying an uncertainty-based approach to computation, we encounter some countability restrictions. These restrictions are key in order to derive computability notions on uncountable sets from the well-established ones in countable sets, namely, those based on \emph{Turing machines}. This leads us to our last research question:

\begin{enumerate}[label=\textbf{(Q\arabic*)}]\addtocounter{enumi}{4}
\item \textbf{What sort of countability restrictions on the order-theoretic structures are required in order to inherit computability from Turing machines?}
\end{enumerate}

The importance of \textbf{(Q5)} lies in the fact it can establish connections between the order-theoretic approaches to computability and the characterizations of preorders by real-valued functions we discussed in Section \ref{intro: real repre}. More precisely, it allows us to relate (i) the countability restrictions on the order structure that are intimately related to the existence of such functions with (ii) the constraints that are imposed in the order-theoretic approach to computability in order for the notions of computability to be derived from
Turing machines. 

Regarding learning systems, it should be noted that they can also be considered as \emph{computation}, in the sense they are driven by uncertainty reduction \cite{gottwald2019bounded} and that they are described by a closely related mathematical structure. Hence, we can think of this third section as an attempt to improve on how computability can be introduced in the study of learning systems.

\newpage
\thispagestyle{empty}
\mbox{}
\newpage
\chapter{Uncertainty and learning systems}
\label{chapter 2}

As we briefly noted in Chapter \ref{intro}, the notion of uncertainty borrowed from thermodynamics is relevant in the contemporary study of learning systems. This is depicted, for instance, by bounded rationality. In this chapter, we aim to develop this uncertainty-based approach to learning systems further. In particular, (i) we present the basic uncertainty framework that deals with learning systems in Section \ref{thermo learning},  and (ii) we develop this framework further in Section \ref{sec:fluct thm learning}, where we derive some results in thermodynamics (the fluctuation theorems) in Section \ref{subsec: theo fluct} and apply them experimentally to a learning system in Section \ref{sec: exp}. 

The reader should note that, in this chapter, we summarize and supplement our work in \cite{hack2022jarzyski} and \cite{hack2022thermodynamic}.

\section{Thermodynamic optimization principles for learning systems}
\label{thermo learning}

The study of learning and decision-making systems has a long tradition in the economics' literature \cite{morgenstern1953theory,savage1972foundations}. In this field, a system is typically presented with a finite set of options $\Omega$ among which it ought to make a decision. The classical studies assume an unbounded decision power, which manifests itself in an optimization principle. This principle assumes that, in the long run, the system will pick the best possible option. In fact, the system's preferences are usually assumed to be represented by a real-valued function $U: \Omega \to \mathbb R$, the so-called \emph{utility} function, that assigns a larger value $U(x)<U(y)$ to an option $y \in \Omega$ that is preferred over another one $x \in \Omega$. Hence, the system's behaviour is usually expected to tend towards picking the optimal option $x^* \in \Omega$, where 
\begin{equation}
\label{no const}
    x^* \coloneqq \argmax_{ x \in \Omega} U(x)
\end{equation}
and $\argmax$ denotes the argument of the maxima (which we assume in this section, for simplicity, to consist of a single option).
Thus, this theory equates the optimal choices with those that the system will make is the long run, that is, when given enough time to make the choice. 

This framework, however, does not consider any constraints the system may have when making the decision. In this scenario, when taking into account information-processing constraints, bounded rationality enters the picture to include restrictions into the optimization principle. In particular, bounded rationality assumes certain bounds on the system's capacity to process information and, thus, to distinguish the optimal choice from the others \cite{tishby2011information,genewein2015bounded}. This results in another optimization principle which is usually stated in terms of Shannon entropy or one of its variations, like the Kullback-Leibler divergence \cite{gottwald2019bounded,cover1999elements}. In particular, the system's optimal behaviour is assumed to follow the probability distribution $p^* \in \mathbb P_\Omega$, where
\begin{equation}
\label{bounded rat prin}
    p^* \coloneqq \argmax_{ p \in \mathbb P_\Omega}  \Big\{\mathbb E_p[U] \Big| H(p) \geq H_0\Big\}
\end{equation}
and $H_0 \geq 0$. Alternatively \cite{gottwald2019bounded}, this optimization principle can be formulated in an equivalent way by interchanging the role of the utility and the uncertainty measure. In particular, we have
\begin{equation}
\label{bounded rat prin II}
    p^* = \argmax_{ p \in \mathbb P_\Omega} \Big\{H(p)\Big| \mathbb E_p[U] \geq U_0\Big\},
\end{equation}
where $U_0 \in \mathbb R$. In fact, these optimization principles have a direct analytical solution. In particular, we have that $p^*$ is a \emph{Boltzmann} distribution, that is, 
\begin{equation}
\label{boltzmann}
    p^*(x) = \frac{1}{Z} e^{\beta U(x)}
\end{equation}
for all $x \in \Omega$ \cite{jaynes1957information}, where $\beta \geq 0$ is a Lagrange multiplier for \eqref{bounded rat prin} and \eqref{bounded rat prin II}, and $Z$ is a normalization constant. Note that, in case there is no information processing constraint, we recover \eqref{no const}, that is, $p^*(x) = \delta(x-x^*)$ for all $x \in \Omega$, where $\delta$ is Dirac's delta.
It should be noted that, in its usual form \cite{genewein2015bounded,ortega2013thermodynamics}, the distribution \emph{prior} to the decision process is included in both \eqref{bounded rat prin} and \eqref{bounded rat prin II}. Nevertheless, for simplicity, we omit it here.

The Boltzmann distribution \eqref{boltzmann} is well-known since it describes the long-term probability of finding a thermodynamic system, for which the average energy is fixed $ \mathbb E_p [E]=E_0 \in \mathbb R$, in each state $x$ of a finite set of possible states $\Omega$. Note that, throughout this work, $E: \Omega \to \mathbb R$ is the energy function, which works as a negative utility function (in the sense that it assigns smaller values to the states preferred by the system).

The coincidence between the long-run behaviour of both systems comes from the fact they are obtained following the same procedure. In particular, they are instances of the \emph{maximum entropy principle} \cite{jaynes1957information,jaynes2003probability}, which, for learning systems, takes the form of \eqref{bounded rat prin} and \eqref{bounded rat prin II}, and, in thermodynamics, is stated as  
\begin{equation}
\label{max entropy}
    p^* = \argmax_{ p \in \mathbb P_\Omega} \Big\{H(p)\Big| \mathbb E_p [E] = E_0\Big\}
\end{equation}
or, alternatively, as 
\begin{equation}
    p^* = \argmin_{ p \in \mathbb P_\Omega} \Big\{\mathbb E_p [E]\Big| H(p) = H_0\Big\},
\end{equation}
where $H_0 \geq 0$ and $\argmin$ denotes the argument of the minima.

Aside from its leading role in thermodynamics and statistical mechanics \cite{jaynes1957information}, the maximum entropy principle is used in both machine learning and decision theory as a justification for several regularization techniques \cite{williams1991function,fox2016taming,maccheroni2006ambiguity,still2009information,tishby2011information,ortega2013thermodynamics}. However, as we will discuss in Chapter \ref{chap: real repre}, the maximum entropy principle is simply a proxy for a more fundamental phenomenon that underlies these different fields, namely, uncertainty. Before doing so, we continue developing the parallelism between thermodynamics and learning systems by addressing their short-term behaviour.

\section{Fluctuation theorems for learning systems}
\label{sec:fluct thm learning}

In its classical formulation \cite{callen1998thermodynamics,giles2016mathematical}, thermodynamics is restricted to the study of systems in equilibrium, that is, with those whose behaviour follows a Boltzmann distribution \eqref{boltzmann}. It was only recently \cite{crooks1998nonequilibrium,jarzynski1997equilibrium} that the regularities in the dynamics that govern the evolution of thermodynamic systems began to be exploited and gave rise to the field of \emph{non-equilibrium} thermodynamics. The main results in this area are a series of predictions, the so-called \emph{fluctuation theorems} \cite{seifert2012stochastic,jarzynski2000hamiltonian,crooks1999entropy}, that involve the statistics of the realization of non-equilibrium processes, that is, thermodynamic processes where the involved distributions are not necessarily Boltzmann-like.

Usually, the setup where these predictions are derived involves stochastic dynamics and, more specifically, Markov chains \cite{crooks1998nonequilibrium,crooks2000path}. The latter assumption is shared by the study of learning systems. In particular, Markov decision processes \cite{tishby2011information} have been widely studied. Hence, there have been attempts to deepen the parallelism between thermodynamics and the study of learning systems with the aim of profiting from the results in non-equilibrium thermodynamics \cite{grau2018non}. However, given the entanglement between mathematical and physical assumptions in the original derivations \cite{crooks1998nonequilibrium,jarzynski1997nonequilibrium}, these attempts have not clarified the derivation of the fluctuation theorems in the context of general Markov chains. Moreover, contrary to the situation in physics \cite{douarche2005experimental}, there has been no experimental evidence of fluctuation theorems in the context of learning systems. Hence, the aim of this section is to clarify the derivation for general Markov chains of the two most prominent fluctuation theorems, i.e. Jarzynski's \cite{jarzynski1997equilibrium} and Crooks' \cite{crooks1998nonequilibrium}, and to test them in the context of the adaptation of a learning system, namely, in a human sensorimotor task.

\subsection{Fluctuation theorems for general Markov chains}
\label{subsec: theo fluct}

In its simplest form, given a finite set $\Omega$, a Markov chain over $\Omega$ consists of a finite number of random variables $\vect{X}=(X_n)_{n=0}^N$, $X_n: \Omega \to \mathbb R$, such that
\begin{equation*}
\prob(X_{n}=x_n |X_0 = x_0 ,\dots,X_{n-1} =x_{n-1}) \coloneqq \prob(X_{n}=x_n | X_{n-1} =x_{n-1})
\end{equation*}
for all $(x_0,x_1,\dots,x_n) \in \Omega^{n+1}$ and $0<n \leq N$, where $\mathbb P (\cdot)$ denotes the probability of what is inside the parentheses. Hence, $\vect{X}$ is characterized by its initial distribution $p_0$, $p_0(x)=\prob(X_0=x)$ for all $x \in \Omega$, and $N$ transition matrices $(M_n)_{n=1}^N$ fulfilling 
\begin{equation*}
    (M_n)_{xy}\coloneqq \prob(X_n=x|X_{n-1}=y)
\end{equation*}
for all $x,y \in \Omega$ and $1 \leq n \leq N$. (Note that any transition matrix $M$ is \emph{stochastic}, i.e., $\sum_{x \in \Omega} (M)_{xy}=1$ for all $y \in \Omega$.)

Before we can state Jarzynski's equality, we need some definitions regarding Markov chains. A transition matrices $M$ is \emph{irreducible} if, for all $x,y \in \Omega$, there exists some $m \geq 1$ such that 
$(M^m)_{xy}>0$. Moreover, if $p$ is a distribution over $\Omega$, we say it is \emph{stationary} for $M$ if $Mp=p$, where $Mp$ stand for the usual matrix product. Importantly, any irreducible matrix $M$ has a unique stationary distribution $p$ whose components are strictly positive. (See \cite[Corollary 1.17 and Proposition 1.19]{levin2017markov}.)

Provided the transition matrices are irreducible, the initial distribution $p_0$ has non-zero components and $p_n$ is the stationary distribution of $M_n$ for $1 \leq n \leq N$, we call a family of functions $\vect{E}=(E_n)_{n=0}^N$, $E_n:\Omega\to\mathbb{R}$, a \emph{family of energies} of $\vect{X}$ if 
\begin{equation*}
    p_n(x) = \frac{1}{Z_n} e^{-\beta E_n(x)}
\end{equation*}
for $0 \leq n \leq N$, where $\beta \geq 0$ and $Z_n \coloneqq \sum_{x \in \Omega} e^{-\beta E_n(x)}$ is a normalization constant. Lastly, we have
\begin{equation}
\label{work and free en}
\begin{split}
    W_{\vect{X},\vect{E}}(\vect{x}) &\coloneqq \sum_{n=0}^{N-1} E_{n+1}(x_n)-E_{n}(x_n) \text{ and } \\
    \Delta F_{\vect{X},\vect{E}} &\coloneqq \frac{1}{\beta} \big(\log (Z_0)- \log (Z_N)\big),
\end{split}
\end{equation}
for all $\vect{x} \in \Omega^{N+1}$. (Note that, although both $W_{\vect{X},\vect{E}}$ and $\Delta F_{\vect{X},\vect{E}}$ depend on $\vect{E}$, their difference does not.)

\subsubsection{Jarzynski's equality}

Now that we have introduced all the required definitions, we can state Jarzynski's equality for Markov chains.

\begin{theo}[Jarzynski's equality for Markov chains]
\label{Jarzynski's equality}
If $\vect{X}=(X_n)_{n=0}^N$ is a Markov chain on a finite state space $\Omega$ whose initial distribution $p_0$ has non-zero entries and whose transition matrices $(M_n)_{n=1}^N$ are irreducible, then we have, for any family of energies  $\vect{E} =(E_n)_{n=0}^N$ of $\vect{X}$, 
\begin{equation}
\label{eq: jarz}
\mathbb E \big[ e^{-\beta (W(\vect{X})-\Delta F)} \big] = 1 \, ,
\end{equation}
where $W=W_{\vect{X},\vect{E}}$ and $\Delta F= \Delta F_{\vect{X},\vect{E}}$. 
\end{theo}

A proof of Theorem \ref{Jarzynski's equality} can be found in \cite[Theorem 1]{hack2022jarzyski}, where we also discuss its relation to previous approaches and potential applications in decision-making.

\subsubsection{Crooks' fluctuation theorem}

If we strengthen the requirements on the Markov chain $\vect{X}$, we can show Crooks' fluctuation theorem as well. Before we can state it, we need another definition. We say a transition matrix $M_n$ with a (unique) stationary distribution $p$ satisfies \emph{detailed balance} if
\begin{equation}
\label{det balance}
(M_n)_{yx} p(x) = (M_n)_{xy}  p(y) 
\end{equation}
for all $x,y \in \Omega$.

\begin{theo}[Crooks' fluctuation theorem for Markov chains]
\label{physics crooks}
Take $\vect{X} = (X_n)_{n=0}^N$ a Markov chain on a finite state space $\Omega$
whose initial distribution $p_0$ has non-zero entries, whose transition matrices $(M_n)_{n=1}^N$ are irreducible,
and where $p_0$ is the stationary distribution of $M_1$ (i.e. $p_1=p_0$).
If $\vect{E}$ is a family of energies of $\vect{X}$ and all transition matrices of $\vect{X}$ satisfy detailed balance, then
\begin{equation} \label{eq: crooks physics}
\frac{\prob^F(W=w)}{ \prob^B(W=-w)} = e^{\beta(w-\Delta F)} 
\qquad \forall w \in \supp (\prob^F(W)),
\end{equation}
where $\prob^F(W=w)$ is the probability that the work of a realization of $\vect{X}$ takes the value $w$, $\prob^B(W=-w)$ is the probability that the work of a realization of $\vect{Y}$ takes the value $-w$, $\supp (\prob^F(W))$ denotes the support of the probability distribution of the work value of $\vect{X}$, that is, the values that can be taken by the work with non-zero probability, and $\vect{Y}$ is the Markov chain that has
$p_N$ as initial distribution and $(M_{N-(n-1)})_{n=1}^N$ as transition matrices.
\end{theo}

A proof of Theorem \ref{physics crooks} can be found in \cite[Corollary 2]{hack2022jarzyski}, where we also discuss its relation to previous approaches. The basic difference between the general Markov chain approach here and the typical one in thermodynamics comes from the definition of $\vect{Y}$, the \emph{reverse} process of $\vect{X}$. (See \cite[Section 5]{hack2022jarzyski} for a discussion.)

Now that we have precised the applicability of Jarzynski's and Crooks' fluctuation theorems in the general framework of Markov chains, we are ready to test them in the context of learning systems. In particular, we test them in a human sensorimotor task in the following section.

\begin{rem}[Fluctuation theorems with continuous state space]
In our experiment, instead of a Markov chain where both time and the state space are discrete, we have a Markov chain where time is discrete and the state space is continuous. The derivation of the results we presented in this section hold almost identically in this scenario, although the definitions ought to be made more carefully. The principal difference is that, in case the state space is continuous (see, for example, \cite{chib1995understanding}), the stochastic matrices are replaced by a \emph{transition kernel} $\prob (A,x)$, which indicates the probability of transitioning from certain point in the state space $x \in \mathbb R$ to some subset of the state space $A \subseteq \mathbb R$ (given that the probability of transitioning to a specific point ought to be zero). However, in order to have a picture similar to the one with a discrete state space we presented above, \emph{transition kernel densities} $\rho(y,x)$ are introduced. These are simply functions that are meant to capture the \emph{probability} of a transition between two points of the state space $x,y \in \mathbb R$ since, when integrated over a region of the state space $A \subseteq \Omega$, they yield the corresponding probability as depicted by the transition kernel $\prob (A,x) = \int_{y \in A} \rho(y,x) dy$. The properties of transition matrices can be analogously defined for transition kernel densities and, hence, Theorems \ref{Jarzynski's equality} and \ref{physics crooks} hold similarly in the context of continuous state spaces. For simplicity, thus, we do not include their details.
\end{rem}



\subsection{Fluctuation theorems in human sensorimotor adaptation}
\label{sec: exp}

As we showed in Section \ref{subsec: theo fluct}, both Jarzynski's and Crooks' predictions hold under mild restrictions for Markov chains. Hence, in principle, they can be observed in several contexts. More specifically, they can be used to test models of human adaptive behaviour and, in particular, sensorimotor decision-making.
In this section, we test these predictions in this context.

As a first step towards our goal, we detail our experimental design, we interpret the thermodynamic concepts involved in Jarzynski's and Crooks' fluctuation theorems in the context of human adaptation and we introduce the notation we will use in this scenario. 


\subsubsection{Experimental design}

The adaptation task we tested involves a cursor on a screen that human participants controlled by moving a lever. The task consisted in moving the cursor towards a single fixed target on the screen. In each trial $0 \leq n \leq N$, the position of the cursor was rotated an angle $\theta_n$ relative to the actual hand position and participants had to adapt when moving the cursor towards the target. In fact, the rotation angle $\theta_n$ changed from one trial to the next according to a protocol that can be found in \cite[Section A.3.3]{hack2022thermodynamic}.

For each trial, the data we considered in our analysis is $x_n$, the angle they achieved after crossing a certain distance from the starting position. Thus, we represent the participants' response by a trajectory  $\vect{x}=(x_0,x_1,..,x_N)$. The deviation between $x_n$ and $\theta_n$ incurs a loss, which we quantify \cite{kording2004loss} as an exponential quadratic error 
\begin{equation}
\label{utility}
 E_n(x_n)= 1 - e^{- (x_n-(\theta_n+b))^2},
\end{equation}
where $b$ is a bias parameter (which is participant-specific). We think of the loss \eqref{utility} as the negative utility of $x_n$. Moreover, following \eqref{bounded rat prin} and \eqref{bounded rat prin II}, the behavior after a suitably long adaption can be described by a Boltzmann distribution \eqref{boltzmann} with $U=-E_n$, $p_n^{eq}$.

The quantity driving the adaptation was the accumulation of the externally induced error changes $\Delta E_{ext}(\vect{x})$ (analogous to the physical concept of work $W_{\vect{X},\vect{E}}(\vect{x})$ in \eqref{work and free en}). In fact, if $\Delta E_{ext}(\vect{x})$ is large, the system is more surprised and has to adapt more. We call $\Delta E_{ext}(\vect{x})$ the \emph{driving error} \cite{hack2022thermodynamic}.

In this scenario (and with the new notation), Jarzynski's and Crooks' fluctuation theorem take, respectively, the following form:

\begin{equation}
\label{prediction II}
\mathbb E \big[ e^{-\beta (\Delta E_{ext}(\vect{X})-\Delta F)} \big] = 1 \, ,
\end{equation}
where the expectation is taken over the Markov chain $\vect{X} = (X_n)_{n=0}^N$ such that (i) its realizations are the  observed data points and (ii) $(p^{eq}_n)_{n=0}^N$ are its stationary distributions, and
\begin{equation}
\label{prediction}
    \Delta E_{ext}(\vect{x}) - \Delta F = \frac{1}{\beta}\log \left( \frac{\rho^F(\Delta E_{ext}(\vect{x}))}{\rho^B(-\Delta E_{ext}(\vect{x}))} \right),
\end{equation}
where $\rho^F$ and $\rho^B$ are the learner's probability densities over driving errors after exposing them to $N+1$ environments in a certain order and its reverse, respectively.
    
In the following, we test the relationships \eqref{prediction II} and \eqref{prediction} experimentally. (We have $\Delta F = 0$ since our protocol starts and ends in the same environment.) We divide the task into 20 cycles of 66 trials each (see \cite{hack2022thermodynamic}) which we call the \emph{forward process} (trials 1 to 25 in each cycle) and the \emph{backward process} (trials 34 to 58). (In the backward process, the set of optimal angles coincides with the one for the forward process, although they are presented in reversed order.) Hence, for both the forward and the backward processes, we record 20 values of $\Delta E_{ext}(\vect{x})$ for each participant and estimate $\rho^F$ and $\rho^B$ from them using kernel density estimation. As the amount of data is limited, we use bootstrapping techniques to test \eqref{prediction II} and \eqref{prediction}.

In the following, before addressing the experimental results, we simulate the fluctuation theorems using the Metropolis-Hastings algorithm.

\begin{figure}[!tb]
\centering
    \includegraphics[scale=.2]{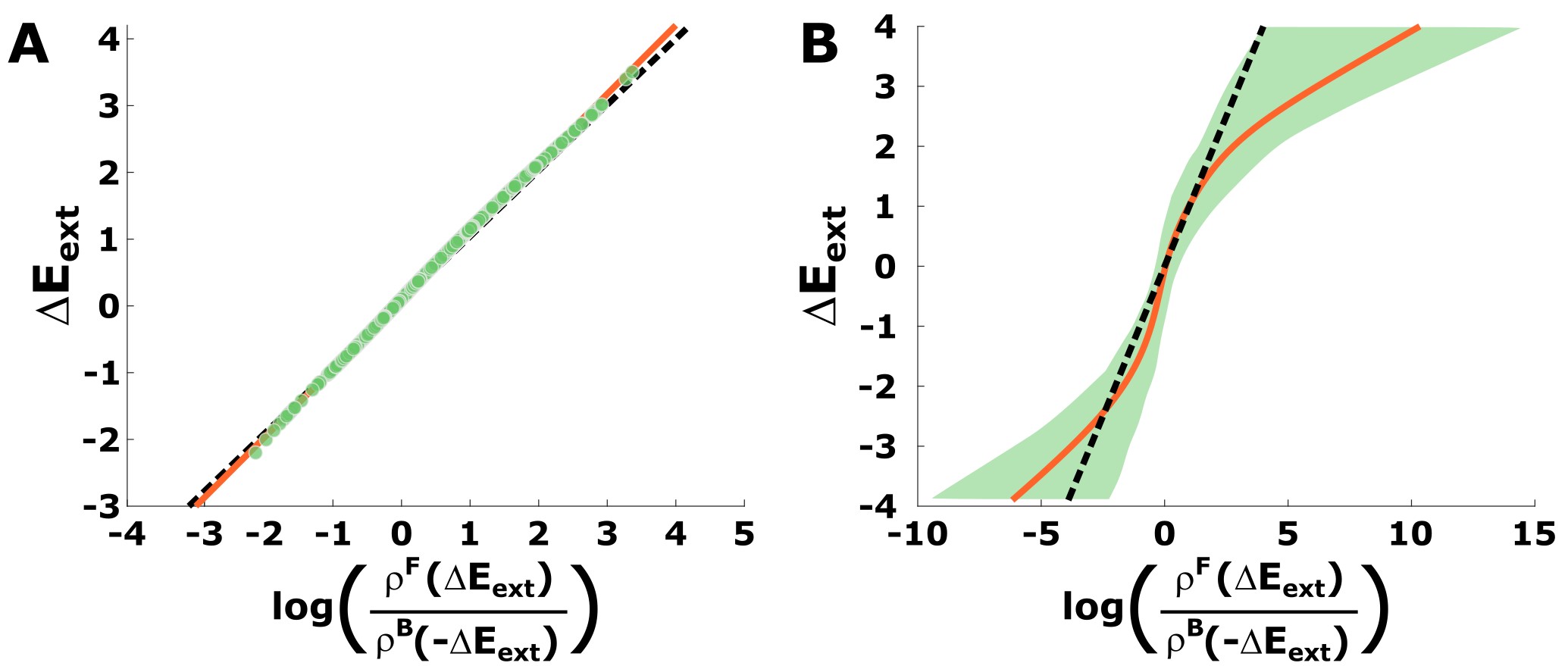}
    \caption{Simulation of Crooks' fluctuation theorem. \textbf{A} Simulation with 1000 cycles. We include the theoretical prediction (black), the linear regression for the simulated data (red) and the simulated points (green). \textbf{B} Simulation with 20 cycles and 1000 bootstraps. We include the theoretical prediction \eqref{prediction} (black) and both the mean (red) and the 99 \% confidence interval (shaded area) of \eqref{prediction} after the bootstraps.
    (Reproduced from \cite{hack2022thermodynamic}, licensed under a Creative Commons Attribution 4.0 International License (http://creativecommons.org/licenses/by/4.0/).)}
    \label{2 simus}
\end{figure}

\subsubsection{Simulations}

In order to simulate the fluctuation theorems, we use the Metropolis-Hasting algorithm to generate data following the specific procedure we detailed in \cite[Section A.2.1]{hack2022thermodynamic}. We address two scenarios:

\begin{enumerate}[label=(\roman*)]
\item We have an arbitrary amount of available data. (Thus, we expect both \eqref{prediction II} and \eqref{prediction} to hold.)
\item We imitate the experimental setup. (Hence, we obtain a rough estimation of what can be expected for our data.) 
\end{enumerate}

We include the result of both simulations in Figure \ref{2 simus}. Notice that, in the situation described by (i), we directly obtain a good adjustment regarding \eqref{prediction}, which automatically (see, for example, \cite{crooks1998nonequilibrium}) implies a good adjustment in \eqref{prediction II}. However, in (ii), we overcome the lack of data by bootstrapping the obtained points. (We follow the same procedure to evaluate the adjustment of the actual data, see \cite{hack2022thermodynamic} for more details.) In this situation, we obtain, after 1000 bootstraps, that Crooks' fluctuation theorem is consistent with the 99 \% confidence interval of our simulated data in the region where $|\Delta E_{\textrm{ext}}| \leq 4$ (i.e. the region where the experimental data lies), and that Jarzynski's equality with $\Delta F = 0$ is also consistent with $(0.48,\text{ }1.64)$, which is the $99\%$ confidence interval we obtain from the mean values we observe in each bootstrap. 

Now that we have gained some intuition regarding the the fluctuation theorems in our setup, we present the experimental results.

\begin{figure}[!tb]
\centering
    \includegraphics[scale=0.21]{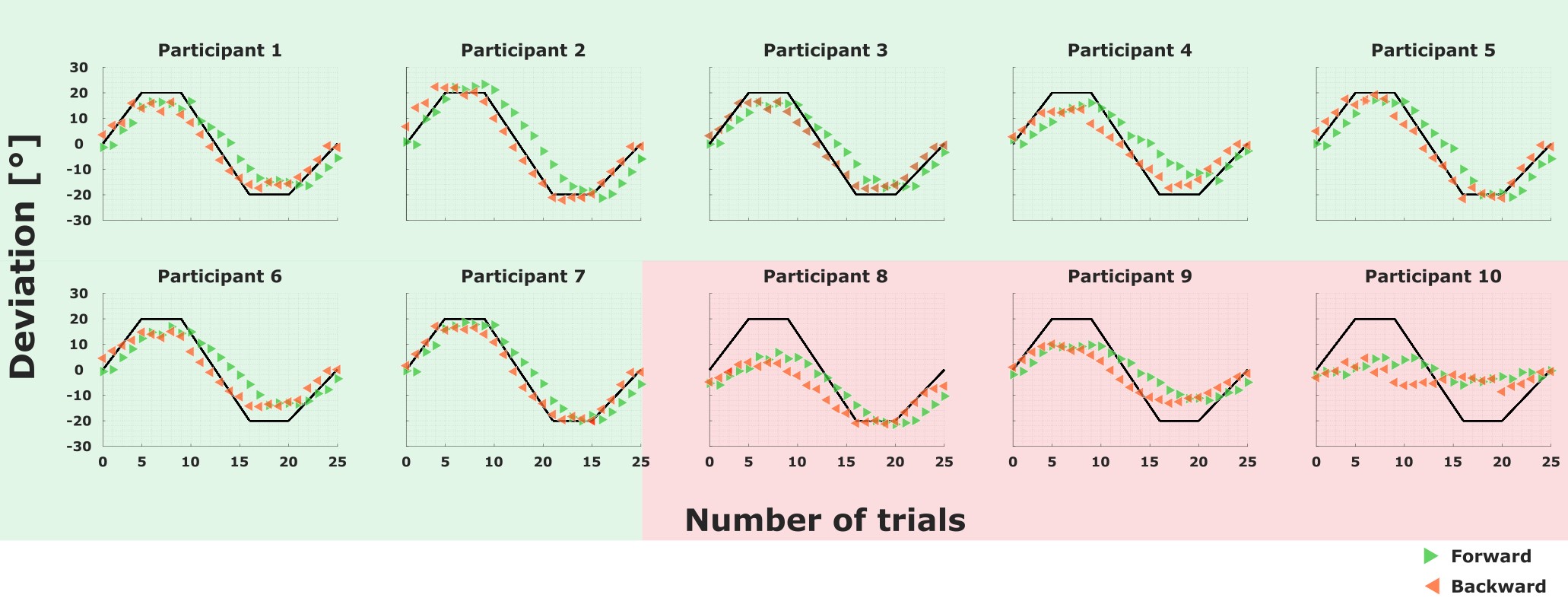}
    \caption{Mean adaptation. The filled triangles are the mean of the observed angles for both the forward (green) and backward (red) processes. The black line is the optimal deviation (the protocol). Participants that achieve at least $50\%$ mean adaptation are shaded by a green background color.
    (Reproduced from \cite{hack2022thermodynamic}, licensed under a Creative Commons Attribution 4.0 International License (http://creativecommons.org/licenses/by/4.0/).)}
    \label{hyste plot}
\end{figure}

\begin{table}[!tb]
\centering
 \begin{tabular}{c c  c c } 
 \hline\noalign{\smallskip}
 Participant & Confidence interval & Participant & Confidence interval \\
 \hline
 1 & \cellcolor[RGB]{175,234,180}(0.03,\text{ }48.59) & 6 & \cellcolor[RGB]{175,234,180} (0.04,\text{ }3.75) \\ 
 2 & \cellcolor[RGB]{175,234,180} (0.03,\text{ }137.58) & 7  &\cellcolor[RGB]{175,234,180} (0.01,\text{ }0.50)\\
 3 & \cellcolor[RGB]{175,234,180} (0.01,\text{ }3.63) & 8 &\cellcolor[RGB]{255, 182, 193}(1.98,\text{ }518130.21)\\
 4 & \cellcolor[RGB]{175,234,180}(0.49,\text{ }63.48) & 9  & \cellcolor[RGB]{255, 182, 193}(0.76,\text{ }77.24)\\
 5 &\cellcolor[RGB]{175,234,180} (0.46,\text{ }1.37) & 10  &\cellcolor[RGB]{255, 182, 193}(0.26,\text{ }48758.33)\\
 \hline
 \end{tabular}
 \caption{Experimental results for Jarzynski's equality. We include the confidence intervals for the left hand side of \eqref{prediction II}, which we obtain after 1000 bootstraps of the observed values of $\Delta E_{ext}(\vect{x})$ for the forward process, and the estimation of $\mathbb E \big[ e^{-\beta \Delta E_{ext}(\vect{X})} \big]$ by its mean for each set of bootstrapped data.
 Note, in our setup, the theoretical prediction requires values around $1.0$ for this estimate. Participants that achieve at least $50\%$ adaptation (c.f. Figure \ref{hyste plot}) are shaded by a green background color. 
 (Reproduced from \cite{hack2022thermodynamic}, licensed under a Creative Commons Attribution 4.0 International License (http://creativecommons.org/licenses/by/4.0/).)}
 \label{jarz participants}
 \end{table}
 
 \begin{figure}[!tb]
\centering
    \includegraphics[scale=0.21]{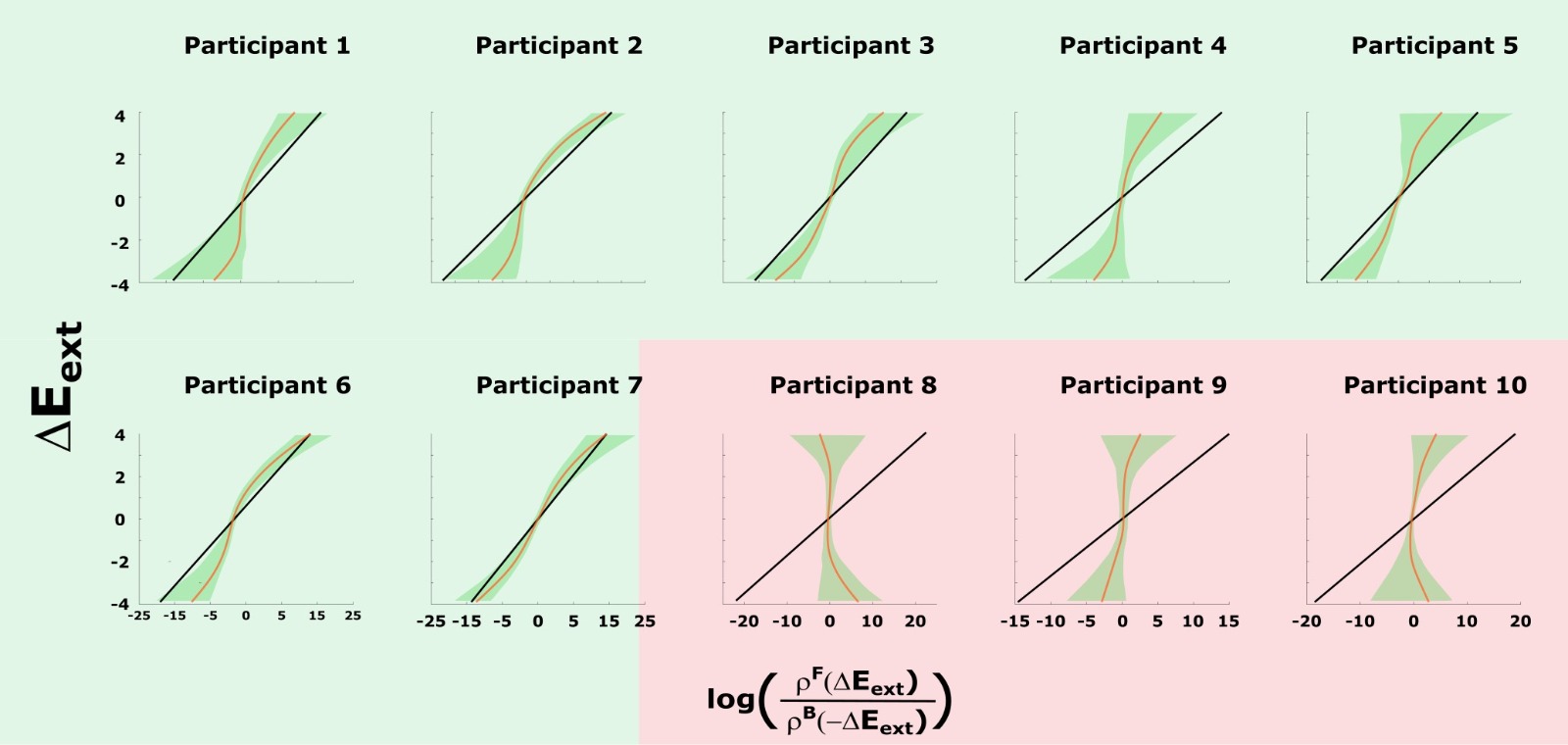}
    \caption{Experimental results for Crooks' fluctuation theorem when the sensorimotor loss behaves as an exponential quadratic error \eqref{utility}. We include the theoretical prediction of Crooks' fluctuation theorem \eqref{prediction} (black) and the mean path after 1000 bootstraps of the observed driving error values (red). Participants that achieve at least $50\%$ adaptation (c.f. Figure \ref{hyste plot}) are shaded by a green background color. The shaded areas inside the graphs are the 99\% confidence intervals which result from bootstrapping. Note that we fit the parameters for each participant according to \cite[Section A.3.3]{hack2022thermodynamic}.
    (Reproduced from \cite{hack2022thermodynamic}, licensed under a Creative Commons Attribution 4.0 International License (http://creativecommons.org/licenses/by/4.0/).)}
    \label{together}
\end{figure}

\subsubsection{Experimental results}

Our experiment consisted of ten participants $(P_i)_{i=1}^{10}$, with $P_i$ being authors of \cite{hack2022thermodynamic} for $i=1,2,3$. Among the participants, there were three ($i=8,9,10$) that presented quite a poor adaptation. In fact, they did not manage to achieve 50 \% adaptation (see Figure \ref{hyste plot}). we expect, thus, the agreement to be poorer in these cases.

Regarding Jarzynski's equality, we obtain consistency with the theoretical prediction for eight out of ten participants. (See Table \ref{jarz participants}.) Importantly, the confidence intervals for two of the participants that do not achieve 50 \% adaptation, $i=8,10$, become larger than the rest. As a consequence, one of them is not consistent with Jarzynski's equality.

Concerning Crooks' fluctuation theorem, the subjects that achieve 50 \% adaptation or higher follow the theoretical trend and, lie within or close to the prediction confidence interval
bounds for ample regions. However, the other subjects show bad adjustment.

In the following section, we challenge the results we have reported in this one by modifying the model we have considered so far.

\begin{figure}[!tb]
\centering
    \includegraphics[scale=0.21]{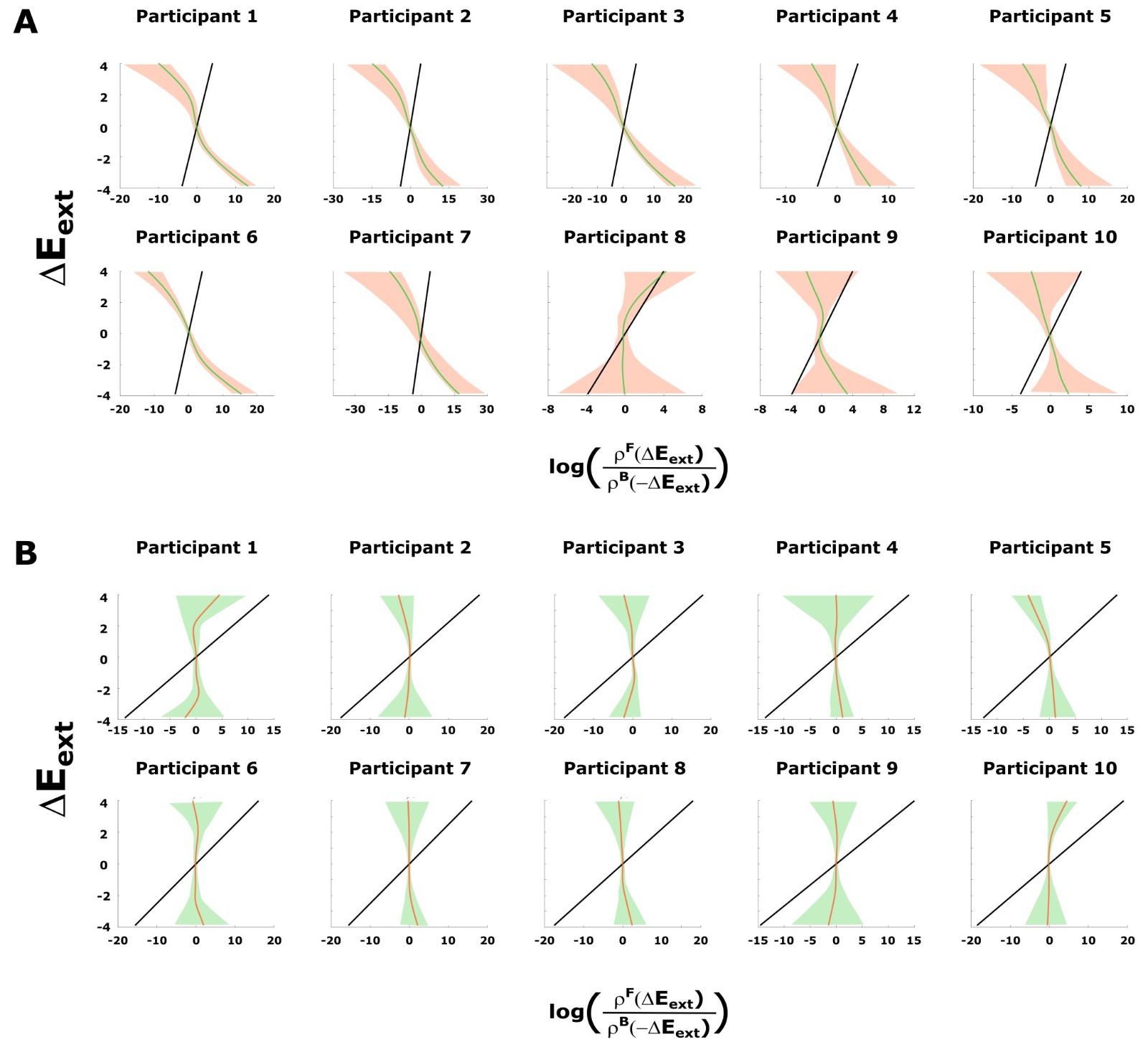}
    \caption{Control results for Crooks' fluctuation theorem in two scenarios: \textbf{A} the sensorimotor loss behaves like a Mexican hat function \eqref{mex hat} and \textbf{B} the sensorimotor loss behaves as an exponential quadratic error but we sample the observed angles randomly with repetition. We include the theoretical prediction of Crooks' fluctuation theorem \eqref{prediction} (black), the mean path after 1000 bootstraps (red) and the 99\% confidence intervals which result from bootstrapping (shaded area). For simplicity, in \textbf{A}, we assume $\beta=1$ for all participants. In \textbf{B}, we fit the parameters for each participant according to \cite[Section A.3.3]{hack2022thermodynamic}. 
    (Reproduced from \cite{hack2022thermodynamic}, licensed under a Creative Commons Attribution 4.0 International License (http://creativecommons.org/licenses/by/4.0/).)}
    \label{togetherB}
\end{figure}

\subsubsection{Robustness}

In order to assess our results, we evaluate them following the procedure above and varying our model's parameters. In particular, we consider two variants:
\begin{enumerate}[label=(\roman*)]
\item We assume absurd utility functions to illustrate the non-triviality of our results.
\item We fix the exponential quadratic error but change the parameters we use (which we fitted in our previous analysis) to determine their influence on the validity of the results.
\end{enumerate}

Regarding (i), we consider two scenarios: A sensorimotor loss that assigns low utility to the target region, in particular, a Mexican hat function that replaces $E_n$ in \eqref{utility} by
\begin{equation}
\label{mex hat}
   E'_n (x)= \frac{2}{\sqrt{3 \sigma} \pi^{\frac{1}{4}}} \bigg(1-\Big(\frac{x-\theta_n}{\sigma}\Big)^2\bigg) \exp \bigg( \frac{(x-\theta_n)^2}{2 \sigma^2} \bigg)
\end{equation}
with $\sigma=4$, and one where we maintain the quadratic exponential loss but the observed angles are randomized. As a result, the consistency in Jarzynski's equality diminishes considerably. More specifically, while there were eight consistent subjects before, there are only two when the Mexican hat function is taken. Moreover, the consistent participants ($i=8,9$) do not reach 50 \% adaptation. When randomizing the angles, after 1000 simulations, the expected number of consistent subjects is $2.33$, with the specific subjects being simulation dependent.
The agreement for Crooks' fluctuation theorem diminishes as well, with the consistency being severely reduced and the data not following the theoretical trend. (See Figure \ref{togetherB}.)

As a last test, we considered convex combinations of a reasonable loss functions (the quadratic exponential) and an absurd one (the Mexican hat) and measured the effect that a variation of the weights has on Crooks' theorem. As expected, the agreement improves as we move the weight towards the quadratic exponential. (See \cite[Table 3]{hack2022thermodynamic}.)

Regarding (ii), the participants' behavior is compatible with Crooks’ fluctuation theorem for a broad neighbourhood of parameter
settings. However, this is not the case for implausible parameters. For the sake of brevity, we do not include these values here. (They can be found in \cite[Figure 8 and Table 2]{hack2022thermodynamic}.)

\section{Summary}

In this chapter, we have extended the study of learning systems through uncertainty tools from optimal to non-optimal behaviour. In particular, we have clarified the applicability of two thermodynamic fluctuation theorems, Jarzynski's equality and Crooks' fluctuation theorem, in the context of learning systems and, moreover, we have obtained the first piece of evidence supporting their validity in a non-physical setup.

In the following chapter, we consider a more intuitive approach to uncertainty based on order structures. The substitution of such an approach by the maximum entropy principle in this chapter leads us to question the relation between order structures and real-valued functions that preserve them in some sense.

\newpage
\thispagestyle{empty}
\mbox{}
\newpage
\thispagestyle{empty}
\mbox{}
\newpage

\chapter{Real-valued representations of preorders}
\label{chap: real repre}

In the previous chapter, we exposed the connection between learning and thermodynamic systems in terms of uncertainty. Moreover, as we briefly noted in Chapter \ref{intro} and develop here in Section \ref{sec: intro pre}, this similarity can be extended to majorization, a fundamental notion of uncertainty that takes the form of a binary relation that ranks the probability distributions over a finite set.
Nonetheless, instead of dealing with it, this basic notion is usually replaced by a real-valued function when facing optimization, as in the maximum entropy principle. The question of when such a binary relation can be substituted by a function constitutes the motivation for Section \ref{opt preorders}, where we point out, in particular, the limitations of such a substitution in the well-known case of the maximum entropy principle. In Section \ref{sec:mu charact}, we complement the approach via functions to optimization over binary relations
by introducing families of functions that fully capture these binary relations. We think of these functions, thus, as complexity measures of the binary relations, given that we can use the number of them and the properties they preserve to compare different binary relations. We devote the following two sections, Sections \ref{sec: dimension} and \ref{sec:classi}, to the study of complexity, optimization and the connection between them for binary relations. Lastly, in Section \ref{sec:application}, we show some concrete applications in physics of the abstract approach taken throughout the chapter. Given that our original motivation comes from majorization and, moreover, the binary relations underlying several other areas of interest the same properties, we focus on the binary relation known as preordered spaces.


The reader should note that, in this chapter, we summarize and supplement our work in \cite{hack2022representing}, \cite{hack2022classification}, \cite{hack2022geometrical} and \cite{hack2022disorder}.

\section{Transitions in thermodynamic and learning systems: Preorders}
\label{sec: intro pre}

The first similarity we established between learning systems and thermodynamics was in terms of their optimal behaviour, which is a Boltzmann distribution since it is obtained following the same optimization principle. We can extend this similarity by assuming, as is usually done in both areas \cite{crooks1998nonequilibrium,crooks2000path,grau2018non}, that their dynamical evolution follows a Markov chain whose equilibrium distribution is a Boltzmann distribution. Notice, in thermodynamics, this is known as the \emph{principle of increasing mixing character} \cite{ruch1976principle,ruch1978mixing,ruch1975diagram} (see Section \ref{subsec: molec diff} for more details). In fact, this was a key property that we exploited in the derivation of both Jarzynski's and Crooks' fluctuation theorem. 

In general, if $d \in \mathbb P_\Omega$ is a Boltzmann distribution that was obtained via either of the optimization principles, it is common to assume, both in thermodynamics and in the study of learning systems,
that
the transition from a distribution $q \in \mathbb P_\Omega$ to another one $p \in \mathbb P_\Omega$ is possible provided that $p$ results from applying to $q$ a stochastic matrix that has $d$ as its stationary distribution. This relation between $p$ and $q$ is usually denoted by $\preceq_d$, that is, we have
\begin{equation}
\label{eq: q-majo}
    p \preceq_d q \iff p= \Pi_d q
\end{equation}
for all $p,q \in \mathbb P_\Omega$, where $\Pi_d$ is a stochastic matrix such that $\Pi_d d =d$ \cite{joe1990majorization}.
The transitions we defined in \eqref{eq: q-majo} are known as $d$-majorization \cite{joe1990majorization,ruch1978mixing}, since $\preceq_d$ is a generalization of majorization (as we detail in the following paragraphs). Crucially, $\preceq_d$ is a preorder, that is, the fundamental transitions that govern the evolution of both thermodynamics and learning systems form a preorder. The study of preorders and, more specifically, how they can be represented using real-valued functions, constitutes the central topic of this chapter. Before doing so, however, let us improve on the relation between majorization and $d$-majorization by giving an alternative definition of $\preceq_d$ taken from \cite{brandao2015second}. 

Take a distribution $d \in \mathbb P_\Omega$ such that $d(x)>0$ for all $x \in \Omega$. (Otherwise, eliminate the elements in $\Omega$ for which $d(x)=0$.) Moreover, assume $d(x) \in \mathbb Q$ for all $x \in \Omega$ (the general case follows from a simple limit of rational distributions). $\preceq_d$ aims to capture the notion of being uncertain relative to $d$. To achieve this, we map $d$ to a uniform on $\Omega'$ (with $|\Omega|< |\Omega'|$) and the other distributions over $\Omega$ to distributions in $\Omega'$ with the distributions whose components are further from $d$ being further from the uniform (in the majorization sense) in $\Omega'$. Concretely, take $\alpha \in \mathbb N$ fulfilling $\alpha d(x) \in \mathbb N$ for all $x \in \Omega$ and, for each $x \in \Omega$, a set $A_x$ with $|A_x|=\alpha d(x)$. Take then $\Omega'$ such that $|\Omega'|=\sum_{x \in \Omega} |A_x|=\alpha$ and $\Lambda_d: \mathbb P_{\Omega} \to \mathbb P_{\Omega'}$, where
\begin{equation*}
    \Lambda_d p(y) \coloneqq \frac{1}{\alpha} \frac{p(x)}{d(x)}
\end{equation*}
for all $y \in A_x$, $p \in \mathbb P_\Omega$ and $x \in \Omega$. We define then, for all $p,q \in \mathbb P_\Omega$,
\begin{equation}
\label{eq: q-majo II}
    p \preceq_d q \iff \Lambda_d p \preceq_M \Lambda_d q,
\end{equation}
where $\preceq_M$ is the majorization relation on $\mathbb P_{\Omega'}$. Hence, $\preceq_d$ orders the distributions in $\mathbb P_{\Omega}$ according to how close their components are to those in $d$. It is in this sense that $\preceq_d$ measures uncertainty relative to $d$. Importantly, definitions \ref{eq: q-majo} and \ref{eq: q-majo II} are equivalent \cite{gottwald2019bounded}.

 
To conclude, notice, if $d$ is the uniform distribution over $\mathbb P_\Omega$, then $\preceq_d$ equals $\preceq_M$ and \eqref{eq: q-majo} reduces to the well-known relation \cite{marshall1979inequalities}:
\begin{equation*}
    p \preceq_M q \iff p= \Pi q
\end{equation*}
for all $p,q \in \mathbb P_\Omega$, where $\Pi$ is a stochastic matrix for which the uniform distribution is stationary (sometimes called a \emph{bistochastic} or \emph{doubly stochastic} matrix).

In the following section, we consider the relation between the maximum entropy principle and majorization and, more in general, that of preordered spaces and optimization principles.


\section{The maximum entropy principle:\\ Optimization principles for preorders}
\label{opt preorders}

In the last section, we presented a more detailed picture of both thermodynamic and learning systems that substituted the simplified account in which transitions can be explained through a specific function for one in which a more fundamental reason for these transitions is given, in particular, a preorder. In fact, a preorder is only partially captured in general by a single function. In this section, we will briefly expose what properties such a function ought to have for it to be chosen as an optimization principle and, hence, motivate the study of real-valued representations of preorders, which is the topic of the next section. For simplicity, we use majorization as our running example throughout this section. 

In a dynamical interpretation of a preorder $(X,\preceq)$, we think of it as a
system that 
transitions
among states until some state is reached where such transitions are no longer possible, that is, for $(\mathbb P_\Omega,\preceq_U)$, until some $p \in \mathbb P_\Omega$ is attained for which there is no $p' \in \mathbb P_\Omega$ such that both $p \preceq_U p'$ and $\neg(p' \preceq_U p)$ hold. We refer to reaching such a $p$, thus, as \emph{maximizing} $\preceq_U$. Since we may be interested in maximizing $\preceq_U$ in several subsets $B \subseteq \mathbb P_\Omega$ and several maxima may exist in each $B$, it is useful to substitute the maximization of $\preceq$ by the maximization of some function $f: \mathbb P_\Omega \to \mathbb R$. This is what occurs in the maximum entropy principle \eqref{max entropy}, where the chosen function is Shannon entropy $H$. The key property that $H$ possesses is that, whenever $p \preceq_U p'$ and $\neg(p' \preceq_U p)$ hold, we have $H(p)<H(p')$. Hence, if a distribution $p$ maximizes $H$ in a subset $B$, then $p$ maximizes $\preceq$ since, otherwise, we violate this property.

If maximizing $f$ over any subset implies maximizing $\preceq$ over the same subset, we call $f$ is an \emph{optimization principle} or we say that $f$ \emph{represents maximal elements} of $\preceq$. It is not hard to see \cite{hack2022representing} that the existence of such an optimization principle is equivalent to that of a \emph{strict monotone} or \emph{Richter-Peleg} function, that is, a function that possesses the property we discussed for Shannon entropy, namely that, for all $x,y \in X$, $x \prec y$ (which is the notation we use to represent that both $x \preceq y$ and $\neg (y \preceq x)$ hold) implies $f(x)<f(y)$, and it is a \emph{monotone}, that is, if $x \preceq y$, then $f(x) \leq f(y)$ (which is also the case for $H$). Notice, however, that a function can be a monotone and not a strict monotone. Whenever $|\Omega| \geq 3$, it is easy to realize this in the case of majorization since, for all $i$ such that $1 \leq i \leq |\Omega|-1$, there exist $x^i,y^i \in \mathbb{P}_{\Omega}$ such that $x^i \prec_U y^i$ and $u_i(x^i) = u_i(y^i)$. Specific examples of such distributions can be found in \cite{hack2022disorder}.


Useful optimization principles
do not only provide distributions that are optimal in the preorder, they output a \emph{single} distribution with such a property. Take, for example, the case of the maximum entropy principle \eqref{max entropy}. The underlying goal of the maximum entropy principle is to provide a numerical procedure (the maximization of $H$) that, given a constraint of (typically) the form $\mathbb E_p[E] = c$ (where $E$ is a random variable), allows us to find a distribution that is maximal in the uncertainty preorder, that is, a distribution $p \in \mathbb P_\Omega$ for which no $q \in \mathbb P_\Omega$ satisfying both $\mathbb E_q[E] = c$ and $p \prec_U q$ exists. This principle works in the case of majorization because of two key properties: on the one hand, $H$ is a strict monotone and, on the other hand, $H$ is strictly concave and maximized over a convex set. The combination of these properties assures that the principle yields a unique distribution that is maximal in $\preceq_U$.

As we highlighted in the last paragraph, the existence of a principle like maximum entropy relies heavily on the convexity of the set where the optimization is carried. Nonetheless, we would like to consider the existence of optimization principles that retain as many properties of maximum entropy as possible in the context of general preordered spaces $(X,\preceq)$.\footnote{This is motivated by our general approach to decision-making via order-theoretic structures that is not limited to majorization and includes, for instance, the bisection method (see Section \ref{sec: uncert and comp}). In particular, we wish to avoid any kind of convexity assumptions.} To this end, we introduced injective monotones in \cite{hack2022representing}, which we define after introducing two notions that we use in their definition. We say a pair of elements $x,y \in X$ are \emph{equivalent}, which we denote by $x \sim y$, if $x \preceq y$ and $y \preceq x$ hold. Moreover, $X/\mathord{\sim}$ is the quotient space of $\sim$, $X/\mathord{\sim} = \{[x]|x\in X\}$, with $[x] = \{y\in X| y\sim x\}$ the equivalence classes.

\begin{defi}[Injective monotones] 
A monotone $f:X\to \mathbb R$ on a preordered space $(X,\preceq)$ is called an \emph{injective monotone}  if $f(x)=f(y)$ implies $x \sim y$, that is, if $f$ is injective considered as a function on the quotient set $X/\mathord{\sim}$.
\end{defi}

Aside from respecting the strict preferences like Shannon entropy, injective monotones associate \emph{incomparable} elements ($x,y \in X$ such that $\neg(x \preceq y)$ and $\neg(y \preceq x)$ hold, which we denote by $x \bowtie y$) to different real numbers.

Injective monotones have the property that, whenever they have optima on some subset, these optima are both equivalent and optima in the preorder. In fact, their existence is equivalent to that of functions that satisfy this property on any subset, namely, \emph{injective optimization principles} \cite{hack2022representing}. It should be noted that such principles extend (up to equivalence) the uniqueness of the maximum entropy principle to \emph{every} subset. This is not the case for Shannon entropy, as we showed in \cite[Lemma 4 $(i)$]{hack2022representing} for non-linear constraints.

Injective monotones are closer to the maximum entropy principle on \emph{partial orders}.\footnote{A partial order $(X,\preceq)$ is an \emph{antisymmetric} preorder, i.e., a preorder where, if $x \preceq y$ and $y \preceq x$ hold for all $x,y \in X$, then $x=y$.} In fact, if $\preceq$ is a partial order (e.g. the quotient set $X/\mathord{\sim}$ equipped with the order relation it inherits from $(X,\preceq)$), then injective monotones preserve all the optimization properties of the maximum entropy principle. In the more general case of preorders, however, the uniqueness cannot be emulated in general, given that, for preordered spaces, equivalent elements have no structural distinction and, thus, are impossible in general to functionally distinguish them only relying on $\preceq$.

In summary, in this section, we have started with an interest in the existence of optimization principles for preorders and have connected this initial aim to the existence of monotones. The relation between these functions is summarized in the following proposition.

\begin{prop}
\label{optimization charac of R-P}
If $(X, \preceq)$ is  a preordered space, then the following statements hold:
\begin{enumerate}[label=(\roman*)]
\item Strict monotones exist if and only if optimization principles do. 
\item Injective monotones exist if and only if injective optimization principles do.
\end{enumerate}
\end{prop}

It should be noted that we stated a slightly weaker version of Proposition \ref{optimization charac of R-P} in \cite[Proposition 3]{hack2022representing}. In fact, the proof of Proposition \ref{optimization charac of R-P} follows directly from that of \cite[Proposition 3]{hack2022representing}.

The simple relations in Proposition \ref{optimization charac of R-P} allow us to connect optimization to the way in which preorders are usually defined, namely, multi-utilities, which we will introduce in the following section. 
More specifically, Proposition \ref{optimization charac of R-P} will allow us to connect optimization with complexity, since the latter is closely related to multi-utilities.

\section{Multi-utility characterizations of preorders}
\label{sec:mu charact}

As we discussed in the previous section, Shannon entropy is a useful tool in the study of uncertainty and, in particular, in the optimization of uncertainty. Nonetheless, entropy does not capture majorization completely. For example, while one can always compare distributions in terms of Shannon entropy, one cannot always compare them in terms of majorization (if $|\Omega| \geq 3$). In particular,
$p=(2/3,1/6,1/6, 0,\dots,0)$ and
$q=(1/2,1/2,0,\dots,0)$,
where $p$ ($q$) has $|\Omega|-3$ ($|\Omega|-2$) zeros, are not related by $\preceq_U$. (See, for $|\Omega|=3$, Figure \ref{fig: example decision-making non-total}.) Nonetheless, this issue can be solved by considering a larger family of functions like
\begin{equation*}
\mathcal U \coloneqq (u_i)_{i=1}^{|\Omega|-1}  \, ,   
\end{equation*}
which we used to define the uncertainty preorder \eqref{uncert rela}. In fact, the existence of (countable) families of functions with properties analogous to those of $\mathcal U$ is closely related to the existence of optimization principles for general preordered spaces, as we will see later in this section. This constituted our original motivation for studying such families, which are central in the following.

\begin{figure}
    \centering
    \begin{tikzpicture}[scale=0.67,font=\fontsize{15}{15}\selectfont]
\begin{axis}[
    ymin=0, ymax=1,
    minor y tick num = 3,
    area style,
    xticklabels={,,},
    ylabel={Probability}
    ]
\addplot+[ybar interval,mark=no] plot coordinates { (0, 0.5) (1, 0) (2, 0.5) (3,0.5)};
\end{axis}
\end{tikzpicture}
\hspace{1cm}
    \begin{tikzpicture}[scale=0.67,font=\fontsize{15}{15}\selectfont]
\begin{axis}[
    ymin=0, ymax=1,
    minor y tick num = 3,
    area style,
    xticklabels={,,}
    ]
\addplot+[ybar interval,mark=no] plot coordinates { (0, 0.67) (1, 0.17) (2, 0.17) (3,0.17)};
\end{axis}
\end{tikzpicture}
\caption{Simple example where the left distribution $q$ and the right distribution $p$ are not related by $\preceq_U$. Despite this fact, the Shannon entropy of $q$ is strictly smaller than that of $p$.}
    \label{fig: example decision-making non-total}
\end{figure}
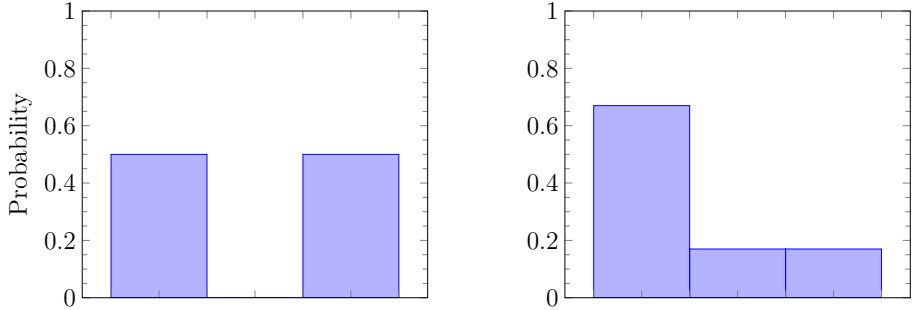

In general, several branches of science study the transitions in their systems of interest by using (a single or several) real-valued functions, which represent measurement devices. As we have seen, this manifests itself, for example, in a fundamental pillar of several disciplines, namely, optimization principles. Examples of areas that (fundamentally) rely to model systems' transitions and, only at a later stage, look for measurement devices that more or less capture their underlying order include (but are not limited to) \emph{thermodynamics} \cite{lieb1999physics,giles2016mathematical}, \emph{general relativity} \cite{bombelli1987space,minguzzi2010time}, \emph{quantum physics} \cite{nielsen1999conditions,brandao2015second} and \emph{economics} \cite{debreu1954representation,ok2002utility}. The same ideas are also relevant in \emph{multi-objective optimization}\cite{jahn2009vector,miettinen2012nonlinear}. Furthermore, the study of the association of real numbers to transitions has often been called \emph{measurement theory} \cite{roberts1979measurement,luce2002representational}.

As a result of the previous paragraph, the study of real-valued representations of preorders is what occupies us throughout the rest of this chapter, where we distinguish between two approaches: one in which multiple functions are used and the order structure is completely captured (e.g. $\mathcal U$ characterizes majorization), and one in which a single one is used and, in general, only part of the order structure is captured (e.g. Shannon entropy for $\preceq_U$ with $|\Omega| \geq 3$).

The fundamental idea regarding complete real-valued representations of a preorder $(X,\preceq)$ is to introduce coordinates into our problem, that is, to map, for some set $I$, the state space of our system into the Cartesian product of $I$ copies of the real line
\begin{equation*}
u: X \to \bigtimes_{i \in I} \mathbb R    
\end{equation*}
in a way such that the coordinates associated to a pair of elements can be compared to determine whether transitions between them are possible, i.e. $x \preceq y$ if and only if $u(x) \preceq_I u(y)$ for all $x,y \in X$, where $\preceq_I$ is some ordering on 
$\times_{i \in I} \mathbb R$. Usually, the \emph{standard} ordering is considered on
$\times_{i \in I} \mathbb R$, that is, $x \preceq_I y$ if and only if $x_i \leq y_i$ for all $i \in I$.\footnote{Although it is widely used in the literature \cite{bosi2020mathematical,evren2011multi,ok2002utility}, the standard ordering is not always considered. We will return to this point in Section \ref{other repres}.} The standard ordering gives rise to the notion of a \emph{multi-utility} \cite{evren2011multi,bosi2012continuous,kaminski2007quasi}. In particular, we say a family of real-valued functions $(u_i)_{i \in I}$, where $u_i: X \to \mathbb R$ for all $i \in I$, is a multi-utility for a preordered space $(X,\preceq)$ if we have 
\begin{equation}
\label{definition mu}
    x \preceq y \iff u_i(x) \leq u_i(y) \text{ for all } i \in I.
\end{equation}
(See Figure \ref{multi-ut 3} for a representation of the ordering when $|I|=2$.) Furthermore, in case $|I|=1$, we say $(u_i)_{i \in I}$ is a \emph{utility} function \cite{bridges2013representations,debreu1954representation}, which coincides with the definition of \emph{utility} that we introduced in Section \ref{thermo learning}. As an instance where the definition of multi-utility applies, we can take $\mathcal U$ and the uncertainty preorder $\preceq_U$. Another instance is
\begin{equation*}
\mathcal F \coloneqq \Big\{\sum_{n=1}^{|\Omega|} g(p_n) \, \Big| \, g \text{ concave} \Big\},    
\end{equation*}
which is also a multi-utility for the uncertainty preorder \cite{marshall1979inequalities}. 

The basic features of the functions that constitute a multi-utility $(u_i)_{i \in I}$ are the following: (i) $u_i$ is a  \emph{monotone} for all $i \in I$,
and (ii) all the involved monotones have to agree in order for some transition to be possible, that is, whenever $x \preceq y$ is false for some $x,y \in X$, there will be some $i_0 \in I$ such that $u_{i_0}(x)>u_{i_0}(y)$. It is in this latter sense (ii) that one may say the minimal multi-utility that exists for some preorder \emph{measures} its complexity, since it is the least number of coordinates that are required in order for the standard ordering on them to faithfully represent the original preorder. We will return to this point in Section \ref{sec: dimension}.

\begin{figure}[!tb]
\centering
\begin{tikzpicture}[scale=0.75, every node/.style={transform shape}]

\node[other node 1] at (0,-2) (1) {\textbf{A}};
\node[other node 2] at (-2,0) (3) {\textbf{C}};
\node[other node 2] at (2,0) (2) {\textbf{B}};
\node[other node 3] at (-2,3) (5) {\textbf{E}};
\node[other node 3] at (2,3) (4) {\textbf{D}};
\node[other node 4] at (0,5) (6) {\textbf{F}};

\node[other node 1] at (-2,-10.5) (1d) {$u$(\textbf{A})};
\node[other node 2] at (0,-10.5) (2d) {$u$(\textbf{B})};
\node[other node 3] at (2,-10.5) (3d) {$u$(\textbf{D})};
\node[other node 2] at (-2,-8) (1a) {$u$(\textbf{C})};
\node[other node 3] at (0,-8) (2a) {$u$(\textbf{E})};
\node[other node 4] at (2,-8) (3a) {$u$(\textbf{F})};

\path[draw,thick,-]
    (1) edge node {} (2)
    (1) edge node {} (3)
    (2) edge node {} (4)
    (3) edge node {} (5)
    (2) edge node {} (5)
    (4) edge node {} (6)
    (5) edge node {} (6)
    ;

    \draw[thick,->]
    (0,-4) -- (0,-6) node[midway,sloped,yshift=-3mm,font=\fontsize{15}{15}\selectfont, thick, rotate=90] {$u$};

\draw[thick,->]
    (-3,-11.5) -- (-3,-9.5) node[midway,sloped,yshift=3mm,font=\fontsize{15}{15}\selectfont, thick, rotate=-90] {$u_2$};
    
     \draw[thick,->]
    (-3,-11.5) -- (-1,-11.5) node[midway,sloped,yshift=-3mm,font=\fontsize{15}{15}\selectfont, thick, rotate=0] {$u_1$};
    
\end{tikzpicture}
\caption{Representation of \eqref{definition mu} with $|I|=2$. The top figure is the Hasse diagram of a preorder $(X,\preceq_0)$ with $X \coloneqq \{A,B,C,D,E,F\}$ and $A\preceq_0 B,C$; $C\preceq_0 E$; $B\preceq_0 D,E$ and $E,D\preceq_0 F$, plus the relations that hold by transitivity and reflexivity. (In a Hasse diagram \cite{harzheim2006ordered}, each point represents an element of the preorder and an arrow between a point $x \in X$ and a higher point $y \in X$ represents that $x \prec y$. Only the nearest neighbour relations are included in a Hasse diagram.) The bottom figure is the image of the elements in $(X,\preceq_0)$ via a multi-utility consisting of two functions $u=(u_1,u_2)$. Note that, since there are incomparable pairs of elements in $(X,\preceq_0)$, there is no multi-utility representation consisting of a single real-valued function.}
\label{multi-ut 3}
\end{figure}
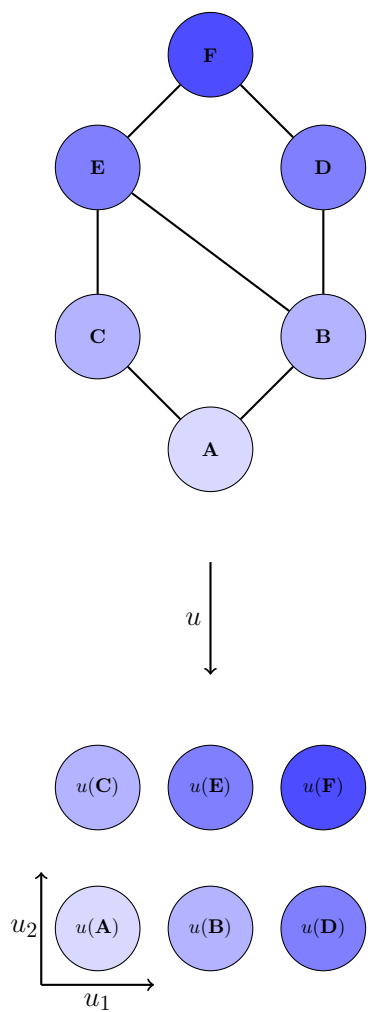

Although the number of functions that are required to form a multi-utility is a measure of the complexity of a preorder, it is not the only one. An important example of another sort of measure of complexity is given by \emph{strict monotone multi-utilities}, which are also known as \emph{Richter-Peleg multi-utilities}. These are multi-utilities $\mathcal V$ where each of its constitutive functions $v \in \mathcal V$ is a strict monotone. Equivalently \cite{alcantud2016richter}, a strict monotone multi-utility $\mathcal V$ is a family of real valued functions, $v:X \to \mathbb R$ for all $v \in \mathcal V$, such that, for all $x,y \in  X$,
\begin{equation}
\begin{cases}
    x \preceq y \iff v(x) \leq v(y) \ \forall v \in \mathcal V, \text{ and }\\
    x \prec y \iff v(x) < v(y) \ \forall v \in \mathcal V.
    \end{cases}
\end{equation}
Hence, strict monotone multi-utilities achieve a better characterization of the \emph{irreversible} transitions a system undergoes. In the case of majorization, for example, 
\begin{equation*}
\begin{split}
  &(f_{i,n})_{1 \leq i \leq |\Omega|-1, n \geq 0} \text{ with}\\
  &f_{i,n} \coloneqq u_i - q_n H
  \end{split}
\end{equation*}
is a strict monotone multi-utility \cite[Proposition 4.1]{alcantud2016richter}, where $(q_n)_{n \geq 0}$ is a numeration of the strictly positive rational numbers. The situation that occurs for majorization, where both multi-utilities and strict monotone multi-utilities coexist, does not happen in general. That is, there exist preorders for which multi-utilities exist while strict monotones do not. Moreover, even when they both exist, a different number of functions may be required in order to constitute them. In fact, this has important consequences in physics, as we will show in Section \ref{sec:application}.

In the same spirit in which we defined strict monotone multi-utilities by requiring the functions that conform a multi-utility to be strict monotones, we could introduce \emph{injective monotone multi-utilities} by requiring the functions to be injective monotones. However, given that they have a close connection to geometry, we postpone their definition until Section \ref{sec: dimension}, where we discuss what the most appropriate notion of dimension for partial orders is. Our interest in this topic also follows from optimization given that, as we will see, complexity (in the multi-utility sense) is closely related to optimization, and the different notions of dimension that we will deal with are also measures of complexity. 

It should be noted that, although we will interpret injective monotone multi-utilities in the light of our discussion concerning dimension, they can also be interpreted, using Proposition \ref{optimization charac of R-P}, via the \emph{description of possible self of the agent defined by $\preceq$} that was introduced in \cite{evren2014scalarization} and used in \cite{alcantud2016richter} to interpret strict monotone multi-utilities.

Before entering the debate concerning the notion of dimension on partial orders, let us briefly introduce some of the notions of order denseness and order separability we will use throughout this work. These definitions are relevant since they are closely related to the representation of preorders by real-valued functions. Moreover, as we will later clarify in Chapter \ref{chap:Uncertainty and computation}, these properties are also closely related to the order-theoretic approach that treats computation as uncertainty reduction. 

\section{Order denseness and separability properties}
\label{sec: order props}

The most fundamental preorder on an uncountable space is that given by the usual ordering on the real line $(\mathbb R, \leq)$. This preorder possesses a well-known countability restriction, namely, $\mathbb Q$ is dense in $\leq$ in the sense that, for all $x,y \in \mathbb R$, there exists some $q \in \mathbb Q$ such that $x < q < y$.

An abstraction of the basic denseness property of $\mathbb Q$ in $\mathbb R$ with its usual ordering is fundamental when studying the representation of preorders by real-valued functions. This extension is called \emph{Debreu separability}. More specifically, we say a preorder $(X,\preceq)$ is \emph{Debreu separable} if it contains a countable set $D \subseteq X$ such that, if $x \prec y$ holds for a pair $x,y \in X$, then there exists some $d \in D$ such that $x \preceq d \preceq y$ \cite{bridges2013representations,debreu1954representation}. Despite being a weaker version of the proper generalization of the density relation between $\mathbb Q$ and $\mathbb R$ (which corresponds to what is called \emph{order separability} \cite{mehta1986existence} or \emph{density in the sense of Cantor} \cite{bridges2013representations} and is defined as the existence of a countable subset $D \subseteq X$ such that, if $x \prec y$ for a pair $x,y \in X$, then there exists some $d \in D$ such that $x \prec d \prec x$), Debreu separability is key in the representation of preorders by real valued functions. In fact, its prominent role in the field is illustrated by
the following fundamental theorem by Debreu \cite{debreu1954representation} on the representation of \emph{total} preorders, that is, preorders $(X,\preceq)$ such that either $x \preceq y$ or $y \preceq x$ holds for any pair $x,y \in X$ \cite{bridges2013representations}.

\begin{theo}[{Debreu \cite[Theorem 1.4.8]{bridges2013representations}}]
\label{debreu thm}
If $(X,\preceq)$ is a total preorder, then it has a utility function $u:X \to \mathbb R$ if and only if it is Debreu separable.
\end{theo}

Note that, instead of Debreu separability, this result is sometimes stated using several equivalent order density properties (see \cite[Proposition 1.4.4]{bridges2013representations} for a list). The reason behind using Debreu separability instead of order separability is clarified by Theorem \ref{debreu thm}: the latter would exclude preorders which have utility functions, like $(\mathbb N, \leq)$. This contrasts with Debreu separability, which, since preorders are reflexive by definition, is fulfilled by any preorder on a countable set.

Given the strong relation between order denseness and the representation of total preorders by utility functions put forward by Debreu, it is natural to ask whether similar relations hold for non-total preorders.\footnote{Regarding applications, note that, although the interest was originally focused on total preorders, the study of non-total preorders was pioneered in \cite{aumann1962utility,peleg1970utility} (see, for example,  \cite{gilboa2010objective}).} Before doing so, however, a notion of order denseness that also addresses incomparable pairs of elements should be defined. Surprisingly, despite Debreu's result, the order denseness property considered in the literature to deal with such a scenario, namely \emph{upper separability} (see, for example, \cite{ok2002utility}), is an extension of order separability. That is, a preorder is said to be upper separable if there exists a countable subset $D \subseteq X$ by which $(X,\preceq)$ is order separable and, for all $x,y \in X$ such that $x \bowtie y$, there exists some $d \in D$ such that $x \bowtie d \prec y$.

Hence, a suitable extension of Debreu separability to non-total preorders is missing. Because of that, we introduced the following definition in \cite{hack2022representing}.

\begin{defi}[Debreu upper separable preorder]
\label{def: debreu upper sep}
We say a preorder $(X,\preceq)$ is Debreu upper separable if there exists a countable subset $D \subseteq X$ by which
$(X,\preceq)$ is Debreu separable and, moreover, there exists some $d \in D$ such that $x \bowtie d \preceq y$ for each pair $x,y \in X$ such that $x \bowtie y$.
\end{defi}

Note that Debreu upper separability is a natural extension of Debreu separability that also addresses incomparabilities in non-total preorders and, like Debreu separability, it is fulfilled by any preorder on a countable set. This contrasts with upper separability which, as a result of the strong requirements in its definition, is not fulfilled by a preorder like $(\mathbb N, =)$.

The advantage of Debreu upper separability lies in the fact we can recover analogous implications to those given by Debreu separability in Theorem \ref{debreu thm} that apply to a larger family of preorders (including some which are not upper separable). In this regard, it should be noted that, for total preorders, a utility function is equivalent to a strict monotone and to the existence of a countable multi-utility. Furthermore, although the equivalence is lost when considering non-total preorders (as we will show in Figure \ref{fig:classification density}), Debreu upper separability still implies the existence of the latter two. We state the connection between order denseness and the most common representations by real-valued monotones in the following proposition, where we call a subset $D \subseteq X$ \emph{Debreu upper dense} (which we defined in \cite{hack2022representing}) if, for any pair $x,y \in X$ such that $x \bowtie y$, there exists some $d \in D$ such that $x \bowtie d \preceq y$.

\begin{table}[!t]
\begin{center}
\begin{tabular}{ p{3cm} p{1.3cm} p{6.5cm} }
 \hline\noalign{\smallskip}
 Name& \hfil Object & \hfil Definition\\
 \hline\noalign{\smallskip}
order dense  & \hfil $Z \subseteq X$    & $\forall x,y \in X$ $x \prec y$ $\implies$ $\exists z \in Z$: $x \prec z \prec y$ \\
 Debreu dense&   \hfil $Z \subseteq X$  & $\forall x,y \in X$ $x \prec y$ $\implies$ $\exists z \in Z$: $x \preceq z \preceq y$  \\
 upper dense&   \hfil$Z \subseteq X$  & $\forall x,y \in X$ $x \bowtie y$ $\implies$ $\exists z \in Z$: $x \bowtie z \prec y$\\
 Debreu upper dense& \hfil$Z \subseteq X$  & $\forall x,y \in X$ $x \bowtie y$ $\implies$ $\exists z \in Z$: $x \bowtie z \preceq y$ \\
 order separable& \hfil $X$ & $\exists Z \subseteq X$ countable: $Z$ is order dense \\
 Debreu separable    & \hfil$X$ & $\exists Z \subseteq X$ countable: $Z$ is Debreu dense\\
 upper separable& \hfil$X$  & $\exists Z \subseteq X$ countable: $Z$ is order dense and upper dense\\
 Debreu upper separable& \hfil $X$  & $\exists Z \subseteq X$ countable: $Z$ is Debreu dense and Debreu upper dense \\
 \hline\noalign{\smallskip}
\end{tabular}
\caption{Separability properties of preordered spaces $(X,\preceq)$.
(Reproduced from \cite{hack2022representing}, licensed under a Creative Commons Attribution 4.0 International License (http://creativecommons.org/licenses/by/4.0/).)}
\end{center}
\label{dense table}
\end{table}

\begin{prop}
\label{deb no mu}
If $(X,\preceq)$ is a preordered space, then the following statements hold:
\begin{enumerate}[label=(\roman*)]
\item If $(X,\preceq)$ is Debreu separable, then it has a strict monotone. However, the converse is false.
\item If $(X,\preceq)$ is Debreu upper separable, then it has a countable multi-utility. However, the converse is false. Moreover, substituting the hypothesis by either Debreu separability or the existence of a countable Debreu upper dense subset is insufficient.
\item If $(X,\preceq)$ has a countable Debreu upper dense subset, then it has a strict monotone if and only if it has a countable multi-utility. Furthermore, the hypothesis is mandatory.
\end{enumerate}
\end{prop}

\begin{proof}

$(i)$ The first statement is well-known, see for example \cite{herden2012utility}. For the second statement, we can consider majorization with $|\Omega| \geq 3$ and note that Shannon entropy is a strict monotone although majorization is not Debreu separable, as we showed in \cite[Lemma 5 (ii)]{hack2022representing}.

$(ii)$ We showed the first statement in \cite[Proposition 9]{hack2022representing}. For the second statement, we can again consider majorization with $|\Omega| \geq 3$ and note that \eqref{uncert rela} is a countable multi-utility although majorization is not Debreu separable, as we showed in \cite[Lemma 5 (ii)]{hack2022representing}.

For the last statement, we can take any total preorder (which has a trivial countable Debreu upper dense subset) that is not Debreu separable (hence neither utility functions by Theorem \ref{debreu thm} nor countable multi-utilities by \cite[Section 4]{alcantud2016richter} exist) to obtain that we cannot weaken the hypothesis to the existence of a countable Debreu upper dense subset.
Moreover, to show that we cannot weaken the hypothesis to \emph{Debreu separability}, we
can take the trivial ordering on the power set of the real numbers $(\mathbb P (\mathbb R ), =)$, which is trivially Debreu separable and has no countable multi-utility since $|\mathbb P (\mathbb R )| > | \mathbb R|$ (see \cite{hack2022classification} for a proof of the bound on the set imposed by a countable multi-utility).
As a perhaps more interesting example of a Debreu separable preorder without a countable multi-utility, consider the set $X \coloneqq [0,1] \cup [2,3]$ equipped with the preorder $\preceq$, where 
\begin{equation*}
    x \preceq y \iff 
\begin{cases}
    x,y \in [0,1] \text{ and } x \leq y,\\
x  \in [2,3],\text{ } y \in [0,1] \text{ and } x-2 <y
\end{cases}
\end{equation*}
for all $x,y \in X$. (One can find a representation of $(X,\preceq)$ in Figure \ref{fig 4}.) Note that $D \coloneqq \mathbb Q \cap [0,1]$ is a countable Debreu dense subset. To conclude, we show that countable multi-utilities do not exist. In order to do so, we consider the characterization of countable multi-utilities by a countable family of increasing sets $(A_i)_{i \in I}$ that separate the elements according to  \cite[Proposition 7]{hack2022representing}. We will show any $(A_i)_{i \in I}$ with those separation properties is uncountable. Since $x \bowtie x+2$ for all $x \in [0,1]$, there exists some $A_x \in (A_i)_{i \in I}$ such that $x+2 \in A_x$ and $x \not \in A_x$. To conclude, it is sufficient to notice that, for all $x,y \in [0,1]$, $A_x = A_y$ implies that $x = y$. Hence, $(A_i)_{i \in I}$ is uncountable. Assume w.l.o.g. that $x < y$. Then, since $A_x$ is increasing, $y \in A_x=A_y$, which contradicts the definition of $A_y$.

$(iii)$ We showed the first statement in \cite[Proposition 10]{hack2022representing}. The second statement is implied by \cite[Proposition 8]{hack2022representing}. 
\end{proof}

As a last remark, note that we include the most common notions of \emph{separability} and \emph{order denseness} in Table \ref{dense table}.

In the following section, now that we have clarified the different notion of order density that we will consider throughout this work, we take our study of complexity on order structures beyond multi-utilities.

\begin{figure}[!tb]
\centering
\begin{tikzpicture}
    \node[main node] (1) {$x+2$};
    \node[main node] (2) [right = 2cm  of 1]  {$y+2$};
    \node[main node] (3) [right = 2cm  of 2]  {$z+2$};
    \node[main node] (4) [below = 2cm  of 1] {$x$};
    \node[main node] (5) [right = 2cm  of 4] {$y$};
    \node[main node] (6) [right = 2cm  of 5] {$z$};

    \path[draw,thick,->]
    (1) edge node {} (5)
    (2) edge node {} (6)
    (1) edge [dashed] node {} (6)
    (4) edge node {} (5)
    (5) edge node {} (6)
    (4) edge [bend right,dashed] node {} (6)
    ;
\end{tikzpicture}
\caption{Graphical representation of a preordered space, defined in Proposition \ref{deb no mu}, which has no countable multi-utility despite being Debreu separable. In particular,  we show $A \coloneqq [0,1]$ in the bottom line, $B \coloneqq [2,3]$ in the top line and how $x,y,z \in A$, $x<y<z$, are related to $x+2,y+2,z+2 \in B$. Note that an arrow from an element $w$ to an element $t$ represents $w \prec t$, where we indicate the relations between nearest neighbours by solid lines and the rest by dashed lines.}
\label{fig 4}
\end{figure}
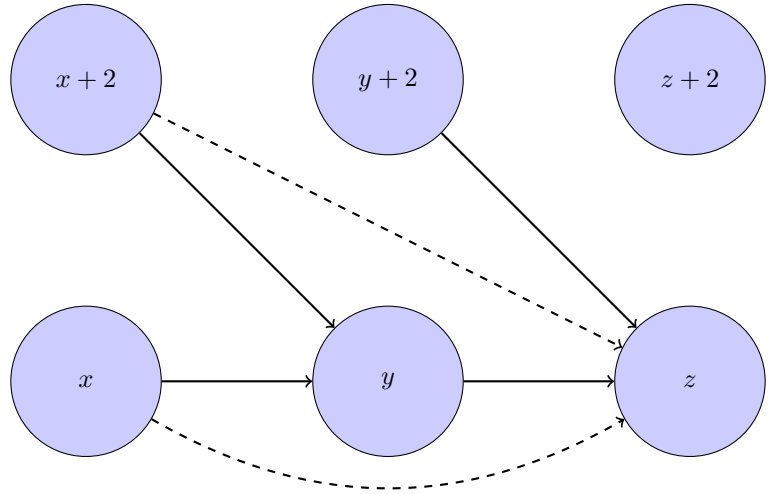

\section{The geometrical notion of dimension for partial orders}
\label{sec: dimension}

The notion of \emph{dimension} is usually attached to geometry.  Roughly, it is the number of copies of the real line required to represent certain mathematical object. As such, the dimension is a minimal translation of the object's properties into a familiar space like the Cartesian product of a set of copies of the real line. A well-known example of this are manifolds, whose dimension $n$ corresponds to the number of copies of the real line such that every point on has a neighborhood equivalent to a subset of $\mathbb R^n$. 

In the case of manifolds, the equivalence is topological. Similarly, in case the structure we are interested in is a partial order $(X,\preceq)$, we could consider $\times_{i \in I} \mathbb R$ with its standard ordering as our \emph{familiar} space.\footnote{In this section, we consider the more specific case of partial orders instead of preorders given that the literature on order dimension usually deals with partial orders \cite{dushnik1941partially,harzheim2006ordered,rival2012ordered}. We comment on the case of preorders later on.} Hence, as we did in the spirit of mathematical economics in \cite{hack2022geometrical} (see, for example, \cite{evren2011multi}), we can define the \emph{geometrical dimension} of $(X,\preceq)$ as follows:

\begin{defi}[Geometrical dimension of a partial order]
\label{geo dim}
If $(X,\preceq)$ is a partial order, then its geometrical dimension is the cardinality of the smallest set $I$ for which
$(X,\preceq)$ has a multi-utility $(u_i)_{i \in I}$.
\end{defi}

The geometrical dimension is not the usual notion of dimension for partial orders \cite{dushnik1941partially,ore1987theory,harzheim2006ordered}. In fact, two equivalent definitions were coined by Dushnik and Miller \cite{dushnik1941partially} and Ore \cite{ore1987theory}, respectively. We introduce these notions in the following section.

\subsection{The Dushnik-Miller dimension and its relation to the geometrical dimension}

  In order to introduce the Dushnik-Miller and Ore dimensions, we need some preliminary concepts. Two partial orders $(X,\preceq)$ and $(Y,\preceq')$ are \emph{order isomorphic} if there exists a bijection $f:X \to Y$ such that $x \preceq y \iff f(x) \preceq' f(y)$
 \cite{ore1987theory}. If $(X,\preceq)$ is a partial order, we say a partial order $(X,\preceq')$ is an \emph{extension} of $(X,\preceq)$ if $x \preceq y$ implies $x \preceq' y$ for all $ x,y \in X$ \cite{harzheim2006ordered}.
Furthermore, if the extension is total, we call it  \emph{linear}. Moreover, a family of linear extensions $(\preceq_i)_{i \in I}$ is called a \emph{realizer} of $\preceq$ provided
\begin{equation*}
    x \preceq y \iff x \preceq_i y \text{ for all } i \in I.
\end{equation*}
In addition, the \emph{Dushnik-Miller dimension} of $(X,\preceq)$ is the cardinality of the smallest $I$ for which a realizer $(\preceq_i)_{i \in I}$ exists \cite{dushnik1941partially}.
A closely related notion is
the \emph{Ore dimension} of $(X,\preceq)$, that is, the cardinality of the smallest $I$
 such that $(X,\preceq)$
 is order-isomorphic to
 $(S,\preceq')$, where $S \subseteq \times_{i \in I} C_i$, $(C_i)_{i \in I}$ is a set of chains and $\preceq'$ is the product-induced order for $(C_i)_{i \in I}$ \cite{ore1987theory,rival2012ordered}.\footnote{The \textit{product-induced} order $\preceq'$ for a family of preordered spaces $((X_i,\preceq_i))_{i\in I}$ is the preorder defined on $\times_{i\in I} X_i$, where $x\preceq' y$ if and only if $x_i \preceq_i y_i$ for all $i \in I$. That is, the product-induced order generalizes the transition from the usual order on $\mathbb R$, $\leq$, to the standard ordering on $\times_{i\in I} \mathbb R$.}
(The Ore and Dushnik-Miller dimensions are known to be equivalent, see \cite[Theorem 10.4.2]{ore1987theory}.)
 
We can think of the geometrical dimension as a special case of the Ore (or, equivalently, Dushnik-Miller) dimension where we impose $C_i = \mathbb R$ and $\preceq_i= \leq$ for all $i \in I$. In fact, if $X$ is countable, it is easy to see that these notions are equivalent. Nonetheless, as we argued in \cite{hack2022geometrical} via \cite[Section 1.4]{bridges2013representations}, they may be significantly different provided $X$ is uncountable. Take, for example, the \emph{lexicographic plane} $(\mathbb{R}^2,\preceq_L)$ \cite{debreu1954representation,bridges2013representations}, where
 \begin{equation}
 \label{def:lexico}
(x,y) \preceq_L (z,t)  \iff
\begin{cases}
    x < z, \text{ or }\\
    x=z  \text{ and } y \leq t
    \end{cases} 
\end{equation}
for all $ x,y,z,t \in \mathbb{R}$. Since $(\mathbb{R}^2,\preceq_L)$ is total, its Dushnik-Miller dimension equals one. Despite this, its geometrical dimension is actually uncountable \cite{hack2022geometrical}. Hence, there is a considerable gap between the geometrical and the Dushnik-Miller dimension. To overcome this difference, we introduced the Debreu dimension in \cite{hack2022geometrical}. In fact, the following section is devoted to the Debreu dimension.

\begin{rem}[Order dimension literature]
Regarding the literature related to the study of order dimension, it should be noted that it is usually interested in the finite dimensional case \cite{trotter1975inequalities,trotter1974dimension,hiraguchi1955dimension,kelly1982dimension}. In fact, the topic is studied in relation to other areas of interest, like graph theory \cite{schnyder1989planar,hocsten1999order} or computer science \cite{felsner2017complexity}, and possesses several interesting open questions, like the \emph{removable pair conjecture} \cite{trotter1992combinatorics}.
\end{rem}

\subsection{The Debreu dimension}

 The disagreement between the geometrical and the Dushnik-Miller dimension originates from the relation between the first one and the real numbers, which is missing for the second one. In order to overcome this issue, we can consider the classical result by Debreu \cite{debreu1954representation} that we stated in Theorem \ref{debreu thm}.
 This theorem equates Debreu separable total preorders with those having a utility function.
 Debreu separability, thus, provides the connection to the real numbers that we were looking for. With this in mind, we introduced in \cite{hack2022geometrical} the \emph{Debreu dimension}, a modification of the Dushnik-Miller dimension.
 \begin{defi}[Debreu dimension]
 \label{def: deb dim}
 If $(X,\preceq)$ is a partial order, then its Debreu dimension is the cardinality of the smallest set $I$ for which
$(X,\preceq)$ has a Debreu separable realizer $(\preceq_i)_{i \in I}$, where we say a realizer is \emph{Debreu separable} if all the linear extensions that conform it are.
\end{defi}

As an important constraint originated by the relation with the real line, the Debreu dimension is only defined for partial orders $(X,\preceq)$ where $X$ has, at most, the cardinality of the continuum.

\begin{rem}[The Debreu dimension for preorders]
In the more general case where $(X,\preceq)$ is a preorder instead of a partial order, it would make sense to generalize Definition \ref{def: deb dim} by asking for the components of the realizer to be \emph{preorders} instead of partial orders. In this case, it is not hard to see (using Theorem \ref{debreu thm}) that the Debreu dimension would be equivalent to the geometrical dimension. We are, however, concerned with partial orders, since that is the case usually considered in the literature on order dimension and, moreover, it is there where interesting differences arise, as we will see in the following section.
\end{rem}

\subsection{Relating the different notions of dimension}

In the remainder of this section, we state our main results regarding the Debreu dimension and, more specifically, its relation to both the geometrical and Dushnik-Miller dimensions. (Notice that we included these results, and others, in \cite{hack2022geometrical}.) In order to state them, however, we need first to formalize the notion of \emph{limit} for a sequence of partial orders.


\begin{defi}[Limit of a sequence of binary relations]
If $(\preceq_n)_{n \geq 0}$ is a sequence of binary relations on a set $X$, the limit of $(\preceq_n)_{n \geq 0}$ is the binary relation $\preceq$ on $X$ that fulfills
\begin{equation*}
\label{def limit}
    x \preceq y \iff \exists \ n_0 \geq 0 \text{ such that } x \preceq_n y \text{ } \forall n \geq n_0
\end{equation*}
for all $ x,y \in X$.
We also denote $\preceq$ by 
$\lim_{n \to \infty} \preceq_n$.
\end{defi}

In case $\preceq_n$ is a partial order for all $n \geq 0$, it is easy to see that $\lim_{n \to \infty} \preceq_n$ is also a partial order. It should be noted that, although we have not encountered this notion previously in the order-theoretic literature, $\lim_{n \to \infty} \preceq_n$ coincides \cite{hack2022geometrical} with the set-theoretic \emph{limit infimum} of the sequence $(A_n)_{n \geq 0}$ \cite{resnick2019probability}, where $A_n \coloneqq \{(x,y) \in X \times X| x \preceq_n y\}$ for all $ n \geq 0$.

Now that we have gathered all the relevant definitions, we can state the main results in \cite{hack2022geometrical}. The first one is concerned with constructing the building blocks of the Debreu dimension, namely Debreu separable linear extensions, when countability constraints are imposed. In particular, we can assume either Debreu upper separability or the weaker assumption that countable multi-utilities exist, as the following theorem states. 

\begin{theo}
\label{count mu has D-ext}
If $(X,\preceq)$ is a partial order, then the following statements hold:
\begin{enumerate}[label=(\roman*)]
\item If $(X,\preceq)$ is Debreu upper separable and $x,y \in X$ are two incomparable elements $x \bowtie y$, then there exists a sequence $(\preceq_n)_{n \geq 0}$ of extensions of $\preceq$,
such that
$\preceq' \coloneqq \lim_{n \to \infty} \preceq_n$ is a Debreu separable extension of $\preceq$ where $x\preceq' y$.
\item If $(X,\preceq)$ has a countable multi-utility and $x,y \in X$ are two incomparable elements, $x \bowtie y$,
then there exists a sequence $(\preceq_n)_{n \geq 0}$ of extensions of $\preceq$,
such that
$\preceq' \coloneqq \lim_{n \to \infty} \preceq_n$ is a Debreu separable extension of $\preceq$ where $x \preceq' y$.
\end{enumerate}
\end{theo}

As a consequence of the last result, we can show that the Debreu dimension actually achieves our original aim: it reduces the gap between the geometrical and Dushnik-Miller dimensions. In particular, the following statement holds: 

\begin{coro}
\label{why deb dim}
If $(X,\preceq)$ is a partial order, then its geometrical dimension is countable if and only if its Debreu dimension is countable.
\end{coro}

In the more specific case where the dimensions are finite, the relation between them is more involved. In this scenario, let us first consider the relation between the geometrical and the Dushnik-Miller dimension. As we discussed for the lexicographic plane \eqref{def:lexico}, there exist instances where the Dushnik-Miller dimension is finite and the geometrical dimension is not. This contrasts with the converse, which actually holds: Provided $(X,\preceq)$ has a countable multi-utility, and given a monotone $u: X \to \mathbb R$, it is possible to use $u$ to build a linear extension of $(X,\preceq)$ that establishes the same preference among incomparable elements than $u$. The statement of the following proposition clarifies what we mean by this.   

\begin{prop}
\label{before teo II}
If $(X,\preceq)$ is a partial order with a countable multi-utility and $u:X \to \mathbb{R}$ is a monotone, then there exists a sequence of extensions of $\preceq$, $(\preceq_n)_{n\geq0}$, such that $\preceq_u \coloneqq \lim_{n\to\infty} \preceq_n$ is a linear extension of $\preceq$ where $x \prec_u y$ for all $ x,y \in X$ such that both $x \bowtie y$ and $u(x)<u(y)$ hold.
\end{prop}

Given that the treatment of incomparable elements is what determines the dimension in any of the three definitions we have considered here, the previous proposition allows us to recover a result that was shown (following a different method) by Ok in \cite{ok2002utility}.

\begin{coro}[{\cite[Proposition 1]{ok2002utility}}]
\label{teo II}
If $(X,\preceq)$ is a partial order with a finite multi-utility $(u_i)_{i=1}^N$, then it has a finite realizer $(\preceq_i)_{i=1}^N$. In particular, any partial order whose geometrical dimension is finite has a finite Dushnik-Miller dimension.
\end{coro}

The last result cannot be emulated when considering the Debreu instead of the Dushnik-Miller dimension. In fact, although Corollary \ref{why deb dim} closely relates the Debreu and geometrical dimensions, large discrepancies (although smaller than those between the geometrical and Dushnik-Miller dimensions) can arise between the geometrical and Debreu dimensions, as the following theorem states.

\begin{theo}
\label{finite geo infinite deb}
There exist partial orders with finite geometrical dimension and countably infinite Debreu dimension. In particular, such partial orders whose geometrical dimension is $2$ exist.
\end{theo}

\begin{figure}[!tb]
\centering
\begin{tikzpicture}
    \node[main node] (1) {$x$};
    \node[main node] (2) [right = 2cm  of 1]  {$y$};
    \node[main node] (3) [right = 2cm  of 2]  {$z$};
    \node[main node] (4) [below = 2cm  of 1] {$-x$};
    \node[main node] (5) [right = 2cm  of 4] {$-y$};
    \node[main node] (6) [right = 2cm  of 5] {$-z$};

    \path[draw,thick,->]
    (4) edge node {} (5)
    (5) edge node {} (6)
    (1) edge node {} (2)
    (2) edge node {} (3)
    (4) edge node {} (1)
    (6) edge node {} (3)
    (5) edge node {} (2)
    (4) edge [dashed] node {} (2)
    (4) edge [dashed] node {} (3)
    (5) edge [dashed] node {} (3)
    ;
\end{tikzpicture}
\caption{Representation of the partial order $(X,\preceq)$ we used in \cite[Theorem 2]{hack2022geometrical} to prove Theorem \ref{finite geo infinite deb}. In particular, we represent
$X \coloneqq \mathbb{R}\setminus\{0\}$ equipped with $\preceq$ where, for all $x,y \in X$, $x \preceq y$ if and only if $|x| \leq |y|$ and $ sgn(x) \leq sgn(y)$, where $|x|$ and $sgn(x)$ denote, respectively, the absolute value of $x$ and the sign of $x$ ($sgn(x) \coloneqq 1$ if $x> 0$ and $sgn(x) \coloneqq -1$ if $x<0$).
As we showed in \cite[Theorem 2]{hack2022geometrical}, the geometrical dimension of $(X,\preceq)$ is finite (it equals $2$) and its Debreu dimension is countably infinite. In this representation, we include three positive real numbers $0<x<y<z$ in the top line and their negatives $-x,-y,-z$ in the bottom line and we draw
an arrow from an element $w$ to an element $t$ whenever $w \prec t$ (we indicate the relations between nearest neighbours by solid lines and the rest by dashed lines).
Aside from helping with the distinction between the different notions of dimension for partial orders, this partial order also illustrates that the existence of a finite multi-utility does not imply that of finite strict monotone multi-utilities, which allows us to improve on the classification of preorders by real-valued monotones (included in Figure \ref{fig:classification density}). Lastly, this partial order will prove to be useful for applications (see Remark \ref{useful monos} in Section \ref{subsec: molec diff}). (Modified from our work in \cite{hack2022geometrical}.)}
\label{fig:counterex}
\end{figure}
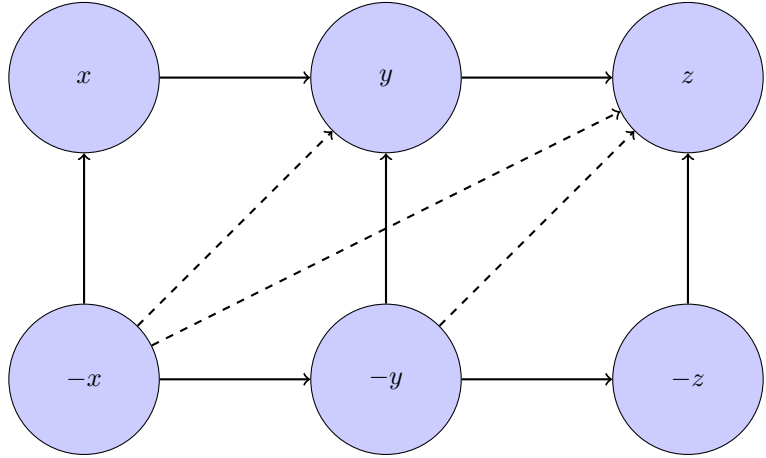

By Theorem \ref{debreu thm}, the geometrical dimension equals one if and only if the Debreu dimension equals one. Hence, the counterexample by which we proved the last theorem is minimal. (See \cite[Theorem 2]{hack2022geometrical} and Figure \ref{fig:counterex} for a representation.) Another instance of this result (with the geometrical dimension being $n-1$) is, for all $n \geq 3$, (the quotient of) majorization with $|\Omega|=n$. Hence, for any natural $m >1$, there are partial orders of geometrical dimension $m$ and countably infinite Debreu dimension. Given that this result is also interesting from the point of view of physics, we wait until Section \ref{sec:application} to formulate it more carefully.

Despite the coincidence between finite geometrical and Debreu dimensions not being fulfilled in general, it can be recovered for dimension two by requiring the partial order in question to be Debreu separable, as we state in the following theorem. In fact, by adding another hypothesis, we can recover it for the general finite case

\begin{theo}
\label{equi finite geo and deb}
If $(X,\preceq)$ is a Debreu separable partial order and $N < \infty$, then the following statements hold:
\begin{enumerate}[label=(\roman*)]
    \item The geometrical dimension of $(X,\preceq)$ equals two if and only if its Debreu dimension equals two.
    \item If there exists a countable subset $B \subseteq X$ such that, for all $x,y \in X$ where $x \bowtie y$, we either have $x \in B$ or $y \in B$, then the geometrical dimension of $(X,\preceq)$ equals $N$ if and only if its Debreu dimension equals $N$.  
    \end{enumerate}
\end{theo}

Note, on the one hand, that the partial order in Figure \ref{fig:counterex}, which we used to show Theorem \ref{finite geo infinite deb}, is not Debreu separable and, importantly, it possesses a countably infinite Debreu dimension despite its geometrical dimension being two. Hence, Debreu separability is not trivial in Theorem \ref{equi finite geo and deb} $(i)$. On the other hand, the hypothesis in Theorem \ref{equi finite geo and deb} $(ii)$ is not satisfactory, although it is only the first attempt towards a hypothesis that equates the geometrical and Debreu dimensions. (For a \emph{nicer} equivalence between finite multi-utilities and finite strict monotone multi-utilities see \cite[Theorem 3 $(i)$]{hack2022geometrical}.) Lastly, note that, under the hypotheses in Theorem \ref{equi finite geo and deb} $(ii)$, we can obtain a result analogous to Proposition \ref{before teo II} but for \emph{Debreu separable} linear extensions.

In the following section, now that we have introduced the sort of real-valued characterizations of preorders that are of interest to us and we have extensively discussed their relation to the notion of dimension for partial orders, we present the classification of preorders in terms of real-valued monotones. 

\section{The classification of preorders in terms of real-valued monotones}
\label{sec:classi}

In this section, we present the classification of preordered spaces in terms of real-valued monotones. Aside from the classes we defined in Section \ref{sec:mu charact}, we also consider injective monotone multi-utilities, that is, multi-utilities such that each function is an injective monotone. We have delayed their introduction since, if $(X,\preceq)$ is a partial orders, then there exists an injective monotone multi-utility $(u_i)_{i \in I}$ if and only if there exists a Debreu separable realizer $(\preceq_i)_{i \in I}$ (see \cite[Proposition 4]{hack2022geometrical}). The case of preorders is equivalent, it just requires to redefine the notion of realizer.

Before we present the classification of preorders according to real-valued monotones, we characterize several of the classes of interest through families of increasing sets.

\subsection{Characterization by families of increasing sets}

The preordered spaces with certain real-valued monotones can be characterized by families of increasing sets that separate the elements in some sense. We call a subset $A\subseteq X$ \emph{increasing} if, for all $x\in A$, $x\preceq y$ implies that $y\in A$ \cite{mehta1986existence}. Furthermore, a family $(A_i)_{i\in I}$ of subsets $A_i\subseteq X$ \emph{separates $x$ from $y$}, if there exists $i\in I$ with $x\not\in A_i$ and $y\in A_i$. In the literature, such families have been recognized as useful to characterize the existence of several real-valued representations of preorders \cite{herden1989existence,mehta1981recent,alcantud1999characterization,alcantud2013representations,bosi2013existence}. The characterizations that we use later on are stated in the following proposition.

\begin{prop}
\label{prop: sets charact}
Let $(X,\preceq)$ be a preordered space.
\begin{enumerate}[label=(\roman*)]
\item  For any infinite set $I$, there exists a multi-utility with the cardinality of $I$ if and only if there exists a family of increasing subsets $(A_i)_{i \in I}$
that $\forall x,y\in X$ with $x \prec y$ separates $x$ from $y$, and $\forall x,y\in X$ with $x \bowtie y$ separates both $x$ from $y$ and $y$ from $x$.
 \item There exists a strict monotone if and only if there exists a countable family of increasing subsets that $\forall x,y\in X$ with $x \prec y$ separates $x$ from $y$.
 \item There exists an injective monotone if and only if there exists a countable family of increasing subsets that $\forall x,y\in X$ with $x \prec y$ separates $x$ from $y$ and $\forall x,y\in X$ with $x \bowtie y$ separates either $x$ from $y$ or $y$ from $x$.
\end{enumerate}
\end{prop}

It should be noted that the proof of $(i)$ can be found in \cite{bosi2013existence}, while we built upon the characterization in the literature and showed $(ii)$ and $(iii)$ in \cite{hack2022representing}. 

The characterization in Proposition \ref{prop: sets charact} is useful in order to picture counterexamples that will allow us to present the most complete characterization of preordered spaces in terms of real-valued monotones so far.

\subsection{The classification of preorders}
\label{subsec:classi}

For the sake of brevity, instead of stating the different results separately, we profit from Figure \ref{fig:classification density} in order to state our main result regarding the classification of preorder. More specifically, the following theorem holds.

\begin{theo}
\label{theo class}
The classification of preorders in terms of real-valued monotones, the cardinality of the quotient set and order density properties is given by Figure \ref{fig:classification density}.
\end{theo}

Our contribution to the classification of preordered spaces can be found in \cite{hack2022classification,hack2022geometrical,hack2022representing} and Propositions \ref{deb no mu} and \ref{new classi 2}. The contribution in \cite{hack2022representing} amounts to the introduction of injective monotones and their relation to the other classes. Furthermore, the contribution in \cite{hack2022geometrical} consists in the distinction between  the existence of finite multi-utilities and that of finite strict monotone multi-utilities (which we obtained as a consequence of the proof of Theorem \ref{finite geo infinite deb}, as we showed in \cite[Corollary 3]{hack2022geometrical}). Aside from that, our contributions in \cite{hack2022classification} to the classification, and a general overview of the relation between the classes can be found in \cite[Section 3 and Figure 2]{hack2022classification}.


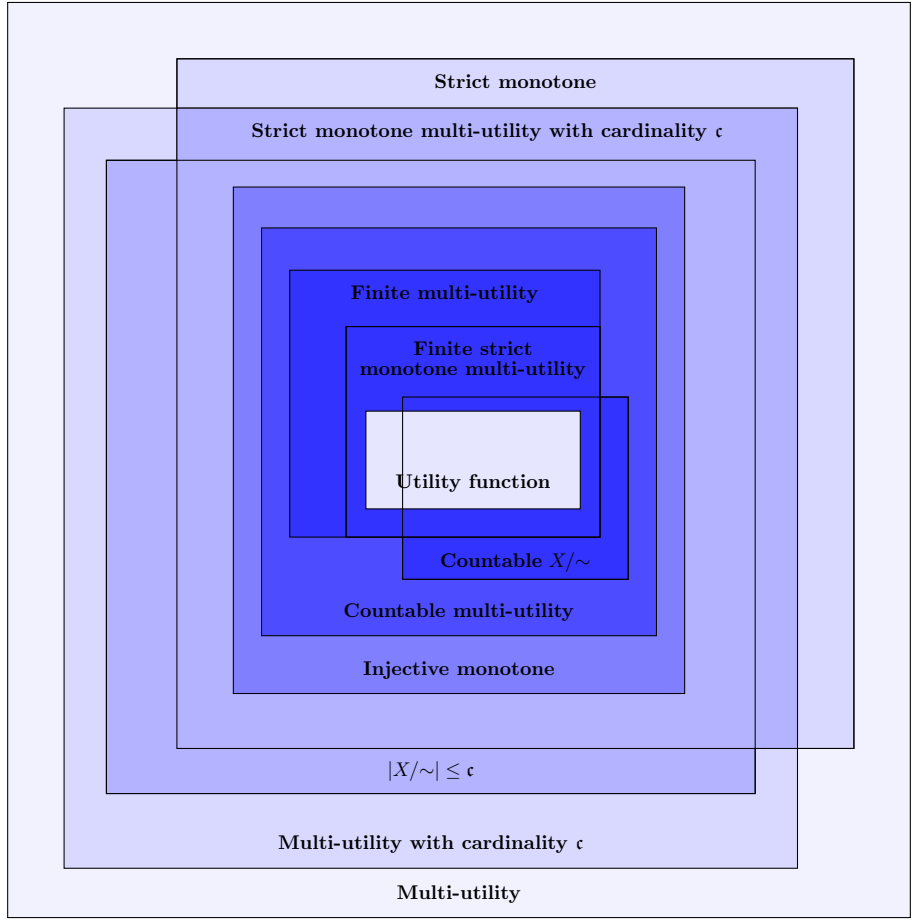
\begin{figure}[!tb]
\centering
\begin{tikzpicture}[scale=0.75, every node/.style={transform shape}]
\node[ draw, fill=blue!5 ,text height = 16cm,minimum     width=16cm, label={[anchor=south,above=1.5mm]270: \textbf{Multi-utility}}]  (main) {};
\node[ draw, fill=blue!15, text height =12cm, minimum width = 12cm,xshift=1cm,yshift=1cm,label={[anchor=south,below=1.5mm]90: \textbf{Strict monotone}}] at (main.center)  (semi) {};
\node[ draw, fill=blue!15, text height =13.25cm, minimum width = 13cm,xshift=-0.5cm,yshift=-0.5cm,label={[anchor=south,above=1.5mm]270: \textbf{Multi-utility with cardinality $\mathfrak{c}$}}] at (main.center)  (semi) {};
\node[ draw, text height =12cm, minimum width = 12cm,xshift=1cm,yshift=1cm] at (main.center)  (semi) {};
\node[ draw,fill=blue!30, text height = 11cm, minimum width = 11.5cm,xshift=-0.5cm,yshift=-0.3cm,label={[anchor=south,above=1mm]270:$|X/\mathord{\sim}| \leq \mathfrak{c}$}] at (main.center) (active) {};
\node[ draw,fill=blue!30,yshift= 0.5625cm,xshift=0.5cm,text height = 11.125cm, minimum width = 11cm,label={[anchor=north,below=1.5mm]90:\textbf{Strict monotone multi-utility with cardinality $\mathfrak{c}$}}] at (main.center) (active) {};
\node[ draw, text height = 11cm, minimum width = 11.5cm,xshift=-0.5cm,yshift=-0.3cm] at (main.center)  (semi) {};
\node[ draw,fill=blue!50, text height = 8.75cm, minimum width = 8cm,yshift=0.35cm,label={[anchor=south,above=1.5mm]270:\textbf{Injective monotone}}] at (main.center) (active) {};
\node[ draw,fill=blue!70, text height = 7cm, minimum width = 7cm,yshift=0.5cm,label={[anchor=south,above=1.5mm]270:\textbf{Countable multi-utility}}] at (main.center) (non) {};
\node[ draw,fill=blue!80, text height = 3cm, minimum width = 4cm,yshift=-0.5cm,xshift=1cm,label={[anchor=south,above=0.1mm]270: \textbf{Countable} $X/\mathord{\sim}$}] at (main.center) (non) {};
\node[ draw,fill=blue!80, text height = 4.5cm, minimum width = 5.5cm,yshift=1cm,xshift=-0.25cm,label={[anchor=north,below=1.5mm]90:\textbf{Finite multi-utility}}] at (main.center) (non) {};
\node[ draw,fill=blue!80, text height = 3.5cm, minimum width = 4.5cm,yshift=0.5cm,xshift=0.25cm,label={[anchor=north,below=1.5mm]90:\textbf{Finite strict}}] at (main.center) (non) {};
\node[ draw, text height = 3.5cm, minimum width = 4.5cm,yshift=0.5cm,xshift=0.25cm,label={[anchor=north,below=5mm]90:\textbf{monotone multi-utility}}] at (main.center) (non) {};
\node[ draw,fill=blue!10, text height = 1.5cm, minimum width = 3.8cm,xshift=0.25cm,label={[anchor=north,below=1cm]90:\textbf{Utility function}}] at (main.center) (non) {};
\node[ draw, text height = 3cm, minimum width = 4cm,yshift=-0.5cm,xshift=1cm] at (main.center) (non) {};
\end{tikzpicture}
\caption{Classification of preordered spaces according to the existence of real-valued monotones and the cardinality of the quotient space. (Based on \cite[Figure 1]{hack2022classification} and \cite[Figure 8]{hack2022geometrical}.)}
\label{fig:classification}
\end{figure}

Aside from the categories we considered in \cite{hack2022classification,hack2022geometrical,hack2022representing}(which we reproduce in Figure \ref{fig:classification}), we include order density properties in Figure \ref{fig:classification density}. To convince ourselves about the validity of this classification, we simply ought to show the following proposition.

\begin{prop}
\label{new classi 2}
The following statements hold:
\begin{enumerate}[label=(\roman*)]
\item If a preordered space has an injective monotones, then it has a subset whose cardinality is smaller or equal to $\mathfrak{c}$ that is both Debreu dense and Debreu upper dense.
\item There exist preordered spaces that have no strict monotone despite having subsets that are both Debreu dense and Debreu upper dense whose cardinality is lower or equal than $\mathfrak{c}$.
\item There exist preordered spaces that have strict monotones and whose subsets that are both Debreu dense and Debreu upper dense have cardinality larger than $\mathfrak{c}$.
    \item There exist preordered spaces with multi-utilities of cardinality $\mathfrak{c}$ that have no strict monotones and whose subsets that are both Debreu dense and Debreu upper dense have cardinality larger than $\mathfrak{c}$.
\end{enumerate}
\end{prop}

\begin{proof}
$(i)$ This holds since having an injective monotone implies the cardinality of the quotient of the ground set under $\mathord{\sim}$ is bounded by $\mathfrak{c}$. Hence, by reflexivity of $\preceq$, the set constituted by one representative for each equivalence class has the desired properties.

$(ii)$ We can simply take the preordered space $(X,\preceq)$ that we introduced in \cite{hack2022classification}[Proposition 3] and notice that $D \coloneqq X$ is both Debreu dense and Debreu upper dense by reflexivity of $\preceq$ and that $|D| \leq \mathfrak{c}$.

$(iii)$ We can take the preordered space we introduced in \cite[Proposition 4 $(ii)$]{hack2022classification}. Since we already know from there that we have strict monotone multi-utilities with the desired cardinality, we simply ought to notice that, for any ground set equipped with the trivial ordering, any Debreu upper dense subset has the cardinality of the ground set.

$(iv)$ We can simply take a preordered space that is the union of those in \cite[Corollary 4 and Proposition 4 $(ii)$]{hack2022classification} (with the elements from different sets being unrelated) and use their proof  plus the argument in $(iii)$. 
\end{proof}

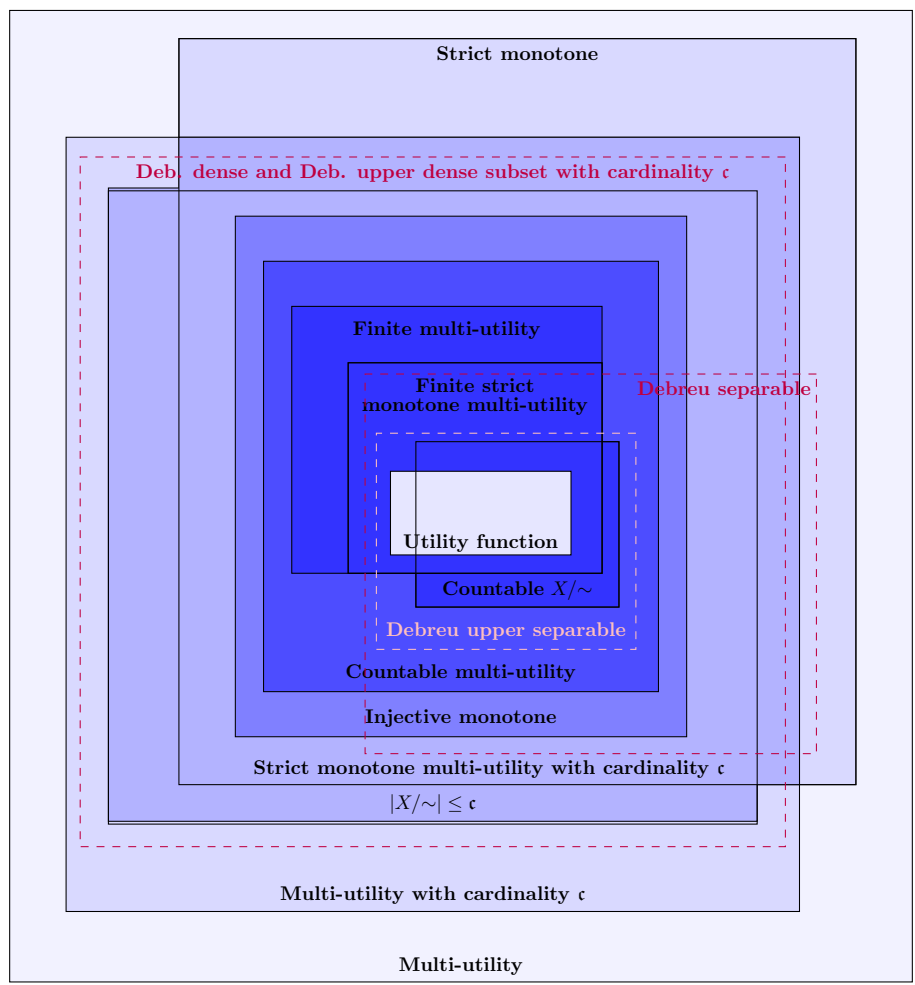
\begin{figure}[!tb]
\centering
\begin{tikzpicture}[scale=0.75, every node/.style={transform shape}]
\node[ draw, fill=blue!5 ,text height = 17cm,minimum   width=16cm, label={[anchor=south,above=.1mm]270: \textbf{Multi-utility}}]  (main) {};
\node[ draw, fill=blue!15, text height =13cm, minimum width = 12cm,xshift=1cm,yshift=1.5cm,label={[anchor=south,below=.1mm]90: \textbf{Strict monotone}}] at (main.center)  (semi) {};
\node[draw, fill=blue!15, text height =13.5cm, minimum width = 13cm,xshift=-0.5cm,yshift=-0.5cm,label={[anchor=south,above=.1mm]270: \textbf{Multi-utility with cardinality $\mathfrak{c}$}}] at (main.center)  (semi) {};
\node[ draw, text height =13cm, minimum width = 12cm,xshift=1cm,yshift=1.5cm] at (main.center)  (semi) {};
\node[ draw,fill=blue!30, text height = 11cm, minimum width = 11.5cm,xshift=-0.5cm,yshift=-0.15cm,label={[anchor=south,above=.1mm]270:$|X/\mathord{\sim}| \leq \mathfrak{c}$}] at (main.center) (active) {};
\node[ draw,fill=blue!30,yshift= 0.625cm,xshift=0.5cm,text height = 11.25cm, minimum width = 11cm,label={[anchor=south,above=.1mm]270:\textbf{Strict monotone multi-utility with cardinality $\mathfrak{c}$}}] at (main.center) (active) {};
\node[ draw, text height = 11cm, minimum width = 11.5cm,xshift=-0.5cm,yshift=-0.2cm] at (main.center)  (semi) {};
\node[ draw,fill=blue!50, text height = 9cm, minimum width = 8cm,yshift=0.35cm,label={[anchor=south,above=.5mm]270:\textbf{Injective monotone}}] at (main.center) (active) {};
\node[ draw,fill=blue!70, text height = 7.4cm, minimum width = 7cm,yshift=0.35cm,label={[anchor=south,above=0.5mm]270:\textbf{Countable multi-utility}}] at (main.center) (non) {};
\node[ draw,fill=blue!80, text height = 2.7cm, minimum width = 3.6cm,yshift=-0.5cm,xshift=1cm,label={[anchor=south,above=0.1mm]270: \textbf{Countable} $X/\mathord{\sim}$}] at (main.center) (non) {};
\node[ draw,fill=blue!80, text height = 4.5cm, minimum width = 5.5cm,yshift=1cm,xshift=-0.25cm,label={[anchor=north,below=1.5mm]90:\textbf{Finite multi-utility}}] at (main.center) (non) {};
\node[ draw,fill=blue!80, text height = 3.5cm, minimum width = 4.5cm,yshift=0.5cm,xshift=0.25cm,label={[anchor=north,below=1.5mm]90:\textbf{Finite strict}}] at (main.center) (non) {};
\node[ draw, text height = 3.5cm, minimum width = 4.5cm,yshift=0.5cm,xshift=0.25cm,label={[anchor=north,below=5mm]90:\textbf{monotone multi-utility}}] at (main.center) (non) {};
\node[ draw,fill=blue!10, text height = 1.25cm, minimum width = 3.2cm,yshift=-0.3cm,xshift=0.35cm,label={[anchor=north,below=1cm]90:\textbf{Utility function}}] at (main.center) (non) {};
\node[ draw, text height = 2.7cm, minimum width = 3.6cm,yshift=-0.5cm,xshift=1cm] at (main.center) (non) {};
\node[ draw, dashed, pink, text height = 3.6cm, minimum width = 4.6cm,yshift=-0.8cm,xshift=0.8cm,label={[anchor=south,above=0.5mm]270:\textcolor{pink}{\textbf{Debreu upper separable}}}] at (main.center) (non) {};
\node[draw,dashed,purple, text height = 12cm, minimum width = 12.5cm,xshift=-0.5cm,yshift=-0.1cm,label={[anchor=south,below=.1mm]90:\textcolor{purple}{\textbf{Deb. dense and Deb. upper dense subset with cardinality 
$\mathfrak{c}$}}}] at (main.center) (non) {};
\node[draw, dashed, purple, text height = 6.5cm, minimum width = 8cm,yshift=-1.2cm,xshift=2.3cm,label={[anchor=south,below=.1mm]55:\textcolor{purple}{\textbf{Debreu separable}}}] at (main.center) (non) {};
\end{tikzpicture}
\caption{Classification of preordered spaces according to the existence of real-valued monotones, the cardinality of the quotient set and order density properties. Aside from the results we mentioned in the caption of Figure \ref{fig:classification}, we include those in Propositions \ref{deb no mu} and \ref{new classi 2}. (Based on \cite[Figure 1]{hack2022classification} and \cite[Figure 8]{hack2022geometrical}.)}
\label{fig:classification density}
\end{figure}

Although the real-valued characterizations we have encountered until now are the ones that are most widely used, some other sorts of characterizations have been considered in the literature. In the following section, we complement the characterizations that appeared previously in this chapter by introducing another type of characterization and giving some basic relations between them.

\subsection{Real-valued representations of preorders in thermodynamics}
\label{other repres}

The study of physical systems using order structures has a long tradition, with thermodynamics being one of the major areas where they have been deployed. In fact, these approaches include
several attempts to the formalization of thermodynamics as a partial order
\cite{lieb1999physics,giles2016mathematical,landsberg1970main}.

Such approaches usually include stronger requirements. In particular, the order structure $(X,\preceq)$ is complemented by an algebraic structure $(X,\circ)$ that is connected to it and models the composition of physical systems. Aside from the algebraic properties that the real-valued functions are required to fulfill, these approaches also differ from the ones we presented in Section \ref{subsec:classi} in an order-theoretic sense.
More specifically, they are concerned with \emph{components of content}, that is, real-valued functions $g:X \to \mathbb{R}$ such that, for all $x,y \in X$, $x \preceq y$ implies $g(x)=g(y)$ \cite{roberts1968axiomatic}. The interest in these functions lies in the fact they represent the quantities that are \emph{conserved} during the transitions $\preceq$ that the system can undergo. Furthermore, another class of real-valued functions of interest there are the so-called \emph{entropy functions},
that is, real-valued functions $s:X \to \mathbb R$ such that, for all $x,y \in X$, $x \preceq y$ implies $s(x) \leq s(y)$ and $x \prec y$ implies $s(x) < s(y)$.\footnote{Note that, usually, the requirements on $g$ and $s$ are actually stronger, since some restrictions regarding the algebraic structure are also involved in their definition \cite{roberts1968axiomatic}.}

These order-theoretic approaches to thermodynamics are concerned with the existence of both a finite set of components of content $g_1,..,g_N$ and an entropy function $s$ such that, together, these functions characterize the system's transitions in the following sense \cite{giles2016mathematical,roberts1968axiomatic}:
\begin{equation}
\label{eneg cons repre}
    x \preceq y \iff
     \begin{cases}
    g_i(x) =g_i(y) \text{ } \forall i=1,..,N \text{, and } \\ s(x) \leq s(y),
    \end{cases}
\end{equation}
where the construction of the entropy function $s$ is always based on some sort of countability restriction (which can take several forms \cite{cooper1967foundations,giles2016mathematical,lieb1999physics}).

In the spirit of Section \ref{subsec:classi}, we define the thermodynamic representations following \eqref{eneg cons repre}. 

\begin{defi}[Thermodynamic representation]
\label{def: thermo repre}
A family of functions $G \cup \{ s \}$, where $g:X \to \mathbb R$ for all $g \in G$ and $s: X \to \mathbb R$, is a thermodynamic representation of a
preordered space $(X,\preceq)$ if it fulfills that
\begin{equation*}
    x \preceq y \iff
    \begin{cases}
    g(x) =g(y) \text{ } \forall g \in G \text{, and } \\s(x) \leq s(y),
    \end{cases}
\end{equation*}
for all $x,y \in X$.
\end{defi}

The purpose of this section is simply to point out, for completeness, the existence of a different kind of real-valued characterization of preorders, namely the thermodynamic representation. (We include a comparison between the two representations in Figure \ref{thermo repre}.) However, since they are nor usually addresses in the study of real-valued characterizations of preordered spaces, we take the opportunity, in the following proposition, to make some preliminary remarks regarding their relation to multi-utilities. 
Before we do so, we introduce a weaker notion of totality for preorders, namely conditional connectedness, and comment on its role in this work in general.

\begin{figure}[!tb]
\centering
\begin{tikzpicture}
    \node[other node] (1) {};
    \node[other node] (2) [right = 2cm  of 1] {}; 
    \node[other node] (3) [right = 2cm  of 2]{};  
    \node[other node] (4) [below = 1.5cm  of 1]{}; 
    \node[other node] (5) [right = 2cm  of 4]{}; 
    \node[other node] (6) [right = 2cm  of 5]{}; 
    
      \node[other node] (1B) [below = 2cm of 4]{};
    \node[other node] (2B) [right = 2cm  of 1B] {}; 
    \node[other node] (3B) [right = 2cm  of 2B]{};  
    \node[other node] (4B) [below = 1.5cm  of 1B]{}; 
    \node[other node] (5B) [right = 2cm  of 4B]{}; 
    \node[other node] (6B) [right = 2cm  of 5B]{}; 

   \path[draw,thick,->]
    (1) edge node {} (2)
    (2) edge node {} (3)
    (4) edge node {} (5)
    (5) edge node {} (6)
    (1) edge [bend left,dashed] node {} (3)
    (4) edge [bend right,dashed] node {} (6)
    
    (1B) edge node {} (2B)
    (2B) edge node {} (3B)
    (4B) edge node {} (5B)
    (5B) edge node {} (6B)
    (1B) edge [bend left,dashed] node {} (3B)
    (4B) edge [bend right,dashed] node {} (6B)
    
    (4B) edge node {} (1B)
    (4B) edge [dashed] node {} (2B)
    (4B) edge [dashed] node {} (3B)
    (5B) edge node {} (2B)
    (5B) edge [dashed] node {} (3B)
    (6B) edge node {} (3B)
    ;
    
     \draw[thick,->]
    (-1,-4.5) -- (-1,-1.5) node[midway,sloped,yshift=3mm,font=\fontsize{15}{15}\selectfont, thick, rotate =-90] {$g$};
    
    \draw[thick,->]
    (-1,-4.5) -- (2,-4.5) node[midway,sloped,yshift=-3mm,font=\fontsize{15}{15}\selectfont, thick] {$s$};
    
     \draw[thick,->]
    (-1,-11) -- (-1,-8) node[midway,sloped,yshift=3mm,font=\fontsize{15}{15}\selectfont, thick, rotate=-90] {$u_1$};
    
    \draw[thick,->]
    (-1,-11) -- (2,-11) node[midway,sloped,yshift=-3mm,font=\fontsize{15}{15}\selectfont, thick] {$u_2$};
\end{tikzpicture}
\caption{Difference between the thermodynamic and multi-utility representations with two functions. We present, above, the partial order on $\mathbb R^2$ where the preorder $(X,\preceq)$ can be embedded if a thermodynamic representation consisting of two functions $\{g,s\}$ exists and, below, the identical situation if a multi-utility $\{u_1,u_2\}$ exists. Note that an arrow from an element $w$ to an element $t$ represents $w \prec t$, where we indicate the relations between nearest neighbours by solid lines and the rest by dashed lines.}
\label{thermo repre}
\end{figure}
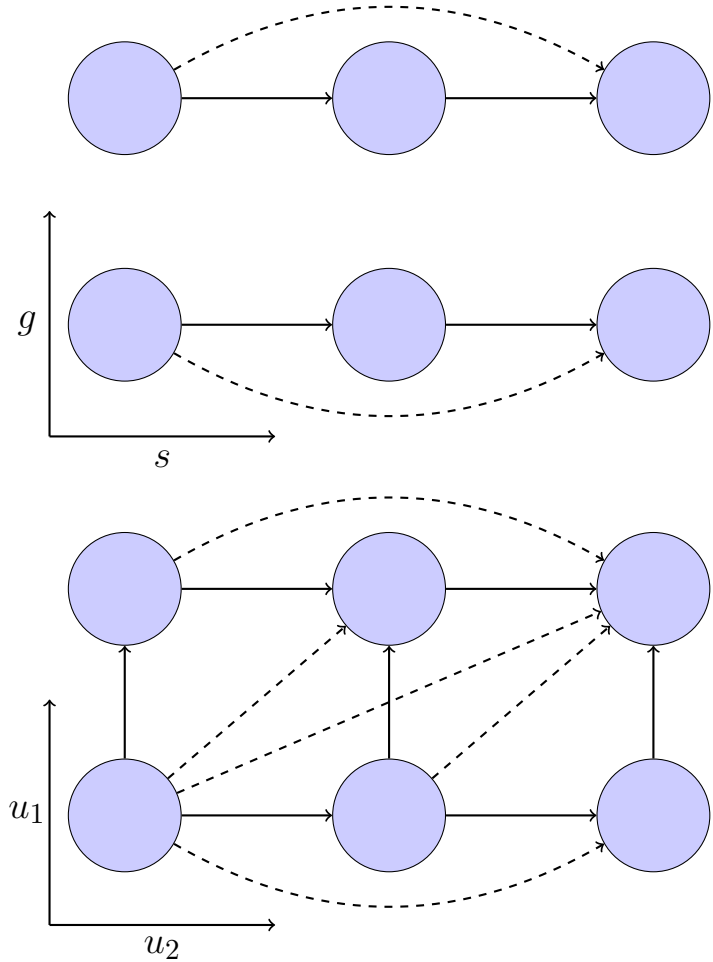

\begin{defi}[Conditionally connected preorder]
\label{def: cond connect}
A preorder $(X,\preceq)$ is conditionally connected if, for any pair $x,y \in P$ for which there exists some $z \in P$ such that $x \preceq z$ and $y \preceq z$, we have that $x$ and $y$ are comparable, i.e. $\neg(x \bowtie y)$.
\end{defi}

\begin{rem}[The role of conditional connectedness]
Although the definition of conditional connectedness is usually presented in a slightly different way in thermodynamics \cite{landsberg1970main,roberts1968axiomatic,giles2016mathematical}, we have a good reason for the specific form of Definition \ref{def: cond connect}. In particular, while the closely related notion of conditional connection between allowed transitions plays a key role in thermodynamics, we introduced the one in Definition \ref{def: cond connect} while studying an order-theoretic approach to computability (see Section \ref{sec: count const turing} and \cite{hack2022relation}). Hence, we believe it is interesting to remark how this property arises in these different disciplines.
\end{rem}

\begin{prop}
\label{thermo repre prop}
If $(X,\preceq)$ is a preordered space, then the following statements hold:
\begin{enumerate}[label=(\roman*)]
\item If $(X,\preceq)$ has a thermodynamic representation consisting of $n+1$ functions,
then it has a multi-utility with $2n+1$ functions. Moreover, if it has a thermodynamic representation consisting of infinitely many functions,
then it has a multi-utility with the same cardinality.
\item If $(X,\preceq)$ has a thermodynamic representation, then it is conditionally connected. In particular, there exist preorders that have strict monotone multi-utilities and no thermodynamic representation.
\item If $(X,\preceq)$ has a countable thermodynamic representation, then it has a thermodynamic representation consisting of two functions. 
\item If $(X,\preceq)$ is total, then it has a thermodynamic representation with a single function if and only if it has a utility function.
\end{enumerate}
\end{prop}

\begin{proof}
$(i)$ We only show the finite case, since the infinite case is equivalent. To show the result, it suffices to notice that $U \coloneqq (u_m)_{m=1}^{2n+1}$ is a multi-utility, where $(g_m)_{m=1}^n \cup \{s\}$ is a thermodynamic representation of $(X,\preceq)$, $u_m = g_m$ for $m=1,\dots,n$, $u_{n+m} = -g_m$ for $m=1,\dots,n$ and $u_{2n+1}=s(x)$. Since all functions involved in $U$ are monotonic by definition and $x \prec y$ implies $u_{2n+1}(x)< u_{2n+1}(y)$ for all $x,y \in X$, we simply have to notice that, if $x \bowtie y$ for a pair $x,y \in X$, then there exists some $m$, $1 \leq m \leq n$, such that $g_m(x) \neq g_m(y)$ and, thus, both $u_m(x)>u_m(y)$ and $u_{n+m}(x)<u_{n+m}(y)$ (or vice versa) hold. 


$(ii)$ We see directly that $(X,\preceq)$ is conditionally connected since, if $x,y \preceq z$ for $x,y,z \in X$ and $(g_i)_{i \in I}  \cup \{s\}$ is a thermodynamic representation, then we have $g_i(x)=g_i(z)=g_i(y)$ for all $i \in I$ and, thus, $\neg (x \bowtie y)$. Moreover, we can profit from this property to easily show there exist preorders with strict monotone multi-utilities that are not conditionally connected and, thus, have no thermodynamic representation.  The preorder in \cite[Proposition 2]{hack2022classification} is a particular example of this.

$(iii)$ Take a thermodynamic representation $(g_n)_{n \geq 0} \cup \{s\}$
and consider $(A_n)_{n \geq 0}$ a numeration of the family of sets $(A_{m,q})_{m \geq 0, q \in \mathbb Q}$, where $A_{m,q} \coloneqq \{x \in X|g_m(x) \geq q\}$ for all $m \geq 0$ and $q \in \mathbb Q$. By \cite[Lemma 1]{hack2022representing}, we have that $g': X \to \mathbb R$, where $g'(x) \coloneqq \sum_{n \geq 0} r^{-n} \chi_{A_n} (x)$ for all $x \in X$ and $r \in (0,\frac{1}{2})$, fulfills the property that, for all $x,y \in X$, $g'(x)=g'(y)$ if and only if $g_n(x) = g_n(y)$ for all $n \geq 0$. Hence, we have that
\begin{equation*}
 x \preceq y \iff
 \begin{cases}
    g'(x) =g'(y) \text{, and } \\s(x) \leq s(y).
    \end{cases}
 \end{equation*}
Thus, $\{g',s\}$ constitutes a thermodynamic representation.

$(iv)$ Straightforward.
\end{proof}

Note that the idea behind $(iii)$ is that the components of content $(g_n)_{n \geq 0}$ form a hyperplane in the representation space $(g_n)_{n \geq 0} \cup \{s\}$ such that the only non-trivial or \emph{irreversible} transitions (i.e. those where $x \preceq y$ and $\neg(y \preceq x)$ hold) take place in the direction that is perpendicular to that hyperplane. In particular, since we can map the hyperplane into the real line injectively, we can convert any countable set of components of content into a single one. However, the argument in $(iii)$ does not work in the case of multi-utilities, since we ought to have a pair of monotones $u,v$ inside the multi-utility such that $u(x) < u(y)$ and $v(x)>v(y)$ whenever $x \bowtie y$. Moreover, there exist preorders with countably infinite multi-utilities and no finite multi-utility, like the counterexample in \cite[Proposition 7]{hack2022classification}.
Lastly, note that $(iii)$ does not necessarily reduce the study of thermodynamic representations to the case with two functions. More specifically, the representations of interest in axiomatic thermodynamics \cite{roberts1968axiomatic,giles2016mathematical} require the components of content to satisfy some extra conditions regarding the algebraic structure $(X,\circ)$, which lie outside our scope here.


In the following section, we abandon the abstract approach that has characterized most of this chapter and provide three specific applications in order to illustrate the usefulness of our more general approach.

\section{Applications}
\label{sec:application}

In this section, we show that the abstract approach we have considered throughout this chapter can have concrete applications. Following our original motivation, we state a straightforward application to decision-making and learning systems (Section \ref{application: learning}). Moreover, we include applications in physics. In particular, we relate the results reported in this chapter with molecular diffusion (Section \ref{subsec: molec diff}) and resource theory (Section \ref{subsec: resource th}).

\subsection{The maximum entropy principle}
\label{application: learning}

As a first application, we return to the topic in Chapter \ref{chapter 2}, that is, decision-making and learning systems, and provide a specific application of our abstract approach that concerns the
maximum entropy principle.
Let us start by addressing some issues that arise when substituting the intuitive picture based on the uncertainty preorder by the maximum entropy principle.

\subsubsection{The issue with the maximum entropy principle}
\label{subsec: maximum entr}

The \emph{principle of insufficient reason} \cite{bernoulli1713ars} states that, when there is no evidence available, we should not establish preferences among different possible manifestations of a phenomenon. In particular, if we ought to assign probabilities to the outcomes of a random variable, we should assume the probability distribution that, while being consistent with the information we have about the random variable, is less biased (i.e. less concentrated around some subset of the outcome space).

Majorization allows us to determine how biased a distribution is, as we detailed in Chapter \ref{intro} when interpreting it as the preference of a casino owner. Hence, we can substitute the principle of insufficient reason by the \emph{principle of minimal majorization} or \emph{maximal uncertainty}. This does not, however, yield a unique distribution in general, that is, when given \emph{any} information that restricts the set of possible distributions. Hence, for commodity, it is usually substituted by the principle of maximum entropy, which, by neglecting the other distributions that are maximal in $\preceq_U$, outputs a single distribution that is maximal in the uncertainty preorder. (This holds provided the constraint is linear, it is not true in general \cite{hack2022representing}.).

The fact that maximum entropy is usually used has resulted in a somewhat artificially prominent role of the Boltzmann distribution in several areas of science. (A few more words in this regard can be found in Remark \ref{shannon entropy and uncertainty}.) In fact, a non-Boltzmann distribution $p$ that is maximal in $\preceq_U$ and fulfills the typical restriction that defines a Boltzmann distribution by means of the maximum entropy principle, namely a linear constraint, can be found in Figure \ref{fig:maxent}. However, applying the maximum entropy principle, the Boltzmann distribution is favored over $p$. Moreover, even when the constraint is linear, the bias introduced by the Shannon entropy may favor the Boltzmann distribution among an uncountable set of distributions that are maximal in $\preceq_U$, as the following proposition shows.

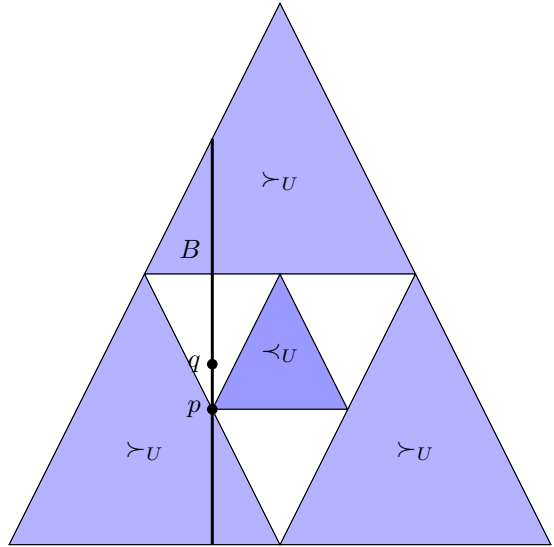
\begin{figure}[!tb]
\centering
\begin{tikzpicture}[scale=0.6]

\draw (0,6)--(-6,-6)--(6,-6)--cycle;
\draw (-3,0)--(3,0)--(0,-6)--cycle;
\draw (0,0)--(-1.5,-3)--(1.5,-3)--cycle;
\draw [fill=blue!30] (-3,0)--(0,6)--(3,0)-- node[above=1cm] {\textbf{$\succ_U$}} cycle;
\draw [fill=blue!30] (3,0)--(6,-6)-- node[above=1cm] {\textbf{$\succ_U$}} (0,-6)--cycle;
\draw [fill=blue!30] (-6,-6)-- node[above=1cm] {\textbf{$\succ_U$}} (0,-6)--(-3,0)--cycle;
\draw [fill=blue!40] (0,0)--(-1.5,-3)-- node[above=0.5cm] {\textbf{$\prec_U$}}(1.5,-3)--cycle;
\draw [fill=white] (0,0)--(-3,0)--(-1.5,-3)--cycle;
\draw [fill=white] (0,0)--(3,0)--(1.5,-3)--cycle;
\draw [fill=white] (0,-6)--(-1.5,-3)--(1.5,-3)--cycle;

\draw [line width=0.4mm] (-1.5,-6)-- (-1.5,3);
\draw [line width=0.0mm] (-1.5,1)-- node[ below left=0.05cm] {\textbf{$B$}} (-1.5,1);
\filldraw (-1.5,-3) circle (3pt) node[left=0.3mm] {$p$};
\filldraw (-1.5,-2) circle (3pt) node[left=0.3mm] {$q$};

\end{tikzpicture}
\caption{The maximum entropy principle does not yield all maximal elements of $\preceq_U$ in some $B \subseteq P_\Omega$ defined by a linear constraint. Here, we show a visualization of the 2-simplex, that is, the set of probability distributions in $P_\Omega$ with $|\Omega|=3$. Take the energy function $E$ with $E(x_1)\coloneqq 1$, $E(x_2)\coloneqq-1$, and $E(x_3)\coloneqq 0$, and let $B \coloneqq \{p \in \mathbb P_\Omega | \langle E \rangle_p = \frac{1}{4}\}$ (the vertical line). While $p=(1/2, 1/4, 1/4)$ is a maximal element of $B$, $q=(9/20,4/20,7/20)$ is also in $B$ and $H(p) < H(q)$. As a result, $p$ is a maximal element of $B$ which is not obtained via the maximum entropy principle.
(Reproduced from \cite{hack2022representing}, licensed under a Creative Commons Attribution 4.0 International License (http://creativecommons.org/licenses/by/4.0/).)}
\label{fig:maxent}
\end{figure}

\begin{prop}
\label{uncount max ent}
If $|\Omega| \geq 3$, then there exist subsets $B \subseteq \mathbb P_\Omega$ that are defined by a linear constraint and possess an uncountable set of maximal elements in $\preceq_U$.
\end{prop}

\begin{proof}
We simply take $B$ the set defined in Figure \ref{fig:maxent} (adding to its definition, if needed, $E(x_i)=0$ for all $i$ such that $3<i \leq |\Omega|$). It is then easy to see that, if we parameterize $B=(p_\lambda)_{\lambda \in [0,3]}$ by
\begin{equation*}
p_\lambda = \frac{1}{8}(5-\lambda,3-\lambda,2 \lambda,0,\dots,0) \text{ for all } \lambda \in [0,3]    
\end{equation*}
and let the number of zeros be equal to $|\Omega|-3$, then $B^\preceq_M=(p_\lambda)_{\lambda \in [1,5/3]}$, where $B^\preceq_M$ stands for the set of distributions that are maximal in the uncertainty preorder when restricted to $B$.
\end{proof}

If $|\Omega|=2$, then $(\mathbb P_\Omega, \preceq_U)$ has no incomparable elements and, hence, Proposition \ref{uncount max ent} does not hold. To conclude this section, let us briefly note how one of the general results we obtained allows us to extend the uniqueness property of the maximum entropy principle.

\subsubsection{Improving on the maximum entropy principle}

One of the relations in Theorem \ref{theo class} exemplifies that we can improve on the maximum entropy principle. In particular, combining it with Proposition \ref{optimization charac of R-P}, we have that injective optimization principles exist for preorders with a countable multi-utility. Hence, there are functions that (up to equivalence) preserve the uniqueness property of the maximum entropy principle for any subset of $\mathbb P_\Omega$. Moreover, applying \cite[Proposition 11 $(ii)$]{hack2022representing}, such functions yield maxima when optimized over any compact set, in particular, over a linear subset of $\mathbb P_\Omega$. Hence, these functions allow us to (essentially) preserve the desirable properties of the maximum entropy principle for any bound that defines a compact subset of $\mathbb P_\Omega$. (Note that Shannon entropy does not possess this uniqueness property on \emph{any} compact set, as we showed in \cite[Lemma 4 (ii)]{hack2022representing}.)   

\subsection{Entropy and the second law in molecular diffusion}
\label{subsec: molec diff}

As a second application of the abstract results in this chapter, we consider the relation between the second law of thermodynamics and molecular diffusion. In particular, in this section, we will briefly enter the realm of physics in order to motivate the use of certain order-theoretic tools in thermodynamics. Hence, a remark regarding the scope of our work seems appropriate.

\begin{rem}[The scope of this work regarding physics]
\label{no physics scope}
Although we have referred to thermodynamics in several sections of this work, the reader may notice we have only done so in order to either expose the origin of certain tools or to exemplify a field where these concepts can be applied. In fact, the specific physical picture that supports the use of these tools lies outside the scope of this work. This limitation is prominent in this section, where we briefly explore the complexity of a non-standard version of the second law of thermodynamics. We are aware of the subtleties involving the second law and, hence, we do not intend to make any claim in this direction. However, we believe it is important to give some naive physical background, since the purpose of this section is to show that, actually, our tools can be useful in the discussion of certain proposals that have been put forward by the physical community \cite{ruch1975diagram,ruch1976principle,mead1977mixing,ruch1978mixing}. In summary, this section should be interpreted as a contribution towards understanding the structure of a proposed physical model with the sole aim of illustrating an application of the tools we developed along this section.  
\end{rem}

Before introducing the physical applications of the order-theoretic framework we developed in this chapter, let us introduce two elementary physical notions that are important for molecular diffusion. It should be noted that we do not intend to provide a precise definition, rather to give the reader a rough estimate of the idea that is sufficient for our purposes in the following.

\subsubsection{Molecular diffusion}

The basic picture we have in mind when referring to molecular diffusion is that of a box that contains a gas, i.e. a set of \emph{balls} moving within the box and colliding with each other (and with the walls of the box). We say our system (the box with the gas) is \emph{isolated} if there is no external influence acting on it, that is, if the balls have no preference regarding what position in the box they occupy. On the contrary, we say our system is \emph{in contact with a heat bath} whenever these balls do have a preference (say in one axis) concerning what position in the box they occupy. Moreover, we consider an (imaginary) division of the box into a finite set of compartments $\Omega$, and take the distributions $p \in \mathbb P_\Omega$ as description of the system's state, where $p_i$ is the fraction of molecules in compartment $i$ for $i=1,\dots,|\Omega|$.
As we will see in the following, for our purposes here, the difference between an isolated system and one in contact with a heat bath is simply a variation of the model we choose for the transitions between the distributions that the gas may undergo.

\subsubsection{Isolated systems: The principle of increasing mixing character}

In this section, we use the term \emph{second law} of thermodynamics to refer to the statement that, in an isolated system described by $\mathbb P_\Omega$, a thermodynamic transition between probability distributions $p,p' \in \mathbb P_\Omega$ takes place if and only if there is an increase in Shannon entropy $H$. Starting from this premise, we expose a critique of this statement that led Ruch \cite{ruch1975diagram,ruch1976principle,ruch1978mixing} and others \cite{mead1977mixing,zylka1985note,alberti1982stochasticity,uhlmann1971satze,alberti1981dissipative}) to propose majorization instead of the second law as the fundamental model that governs thermodynamic transitions. (It should be noted that this critique is centered around the case with a few molecules. The reader interested in more details, which lie outside the scope of this work, may find \cite{gour2015resource,lostaglio2019introductory} of interest.) After briefly providing some physics' background, we proceed to the main result in this section: We show that, when considering the uncertainty preorder (which we call \emph{disorder}) as the driving force behind molecular diffusion, an infinite number of functions are required to replicate the properties attributed to Shannon entropy via the second law of thermodynamics. Hence, we show that, in this model, the complexity of molecular diffusion is infinite.

Although our purpose here is not to argue against the second law of thermodynamics (which, as discussed in Remark \ref{no physics scope}, lies outside our scope here), but to explore the complexity of a preorder that has been proposed to fundamentally underpin  molecular diffusion, we (briefly) recall the reasoning that led Ruch to formulating a statement stronger than the increase of Shannon entropy, namely, the \emph{principle of increasing mixing character} \cite{ruch1976principle,ruch1978mixing,ruch1975diagram}. 

The principle of increasing mixing character postulates that, as time evolves, the distribution that describes an isolated thermodynamic system tends to become more mixed, that is, any statistical preference over specific states that the system may have held tends to vanish. Thus, its distribution becomes less biased over time. More specifically, Ruch proposed that the transitions an isolated thermodynamic system may undergo are those allowed by the uncertainty preorder, since it relates the distributions over some finite set according to how biased or concentrated around some points these distributions are. In particular, a transition from $p \in \mathbb P_\Omega$ to $p' \in \mathbb P_\Omega$ is possible provided $p' \preceq_M p$, where $\preceq_M$ denotes the majorization preorder.

\begin{rem}[Majorization and molecular diffusion]
The reader may object that some of the transitions allowed by majorization would, actually, not take place in molecular diffusion. In particular, for the gas with three compartments we have considered all along, we have that $p \preceq q$ for $p=(1/4,1/2,1/4)$ and $q=(3/8,1/2,1/8)$. Nonetheless, reaching $p$ from $q$ by molecular diffusion does not seem likely, since we would expect the number of molecules in the central compartment to decrease given it is in contact with the other two compartments, which actually exchange molecules. However, the main result in this section (Theorem \ref{dim majo}) still holds true, namely, the analogous of a finite family of second laws of disorder (which we define below) do not exist for molecular diffusion. This is the case since the transitions in the ordered subset of majorization that we used to prove Theorem \ref{dim majo} in \cite[Theorem 1]{hack2022disorder} remain the same in this improved model of molecular diffusion where exchanges of molecules between non-adjacent compartments are forbidden.
\end{rem}

In order to discuss the physical intuition that led to proposing majorization instead of the second law to study thermodynamic transitions, consider a simple example by Mead \cite{mead1977mixing} (which we reproduced in \cite{hack2022disorder}). Take a gas in a box that is divided into three compartments by two walls. Consider two possible states \textbf{A} and \textbf{B} (see Figure \ref{gas trans} for a representation): 
\begin{enumerate}[label=\textbf{(\Alph*)}]
\item Half of the molecules are in the first and the other half are in the second compartment. 
\item $2/3$ of the molecules are in the first and $1/6$ are in both the second and third compartments.
\end{enumerate}
In this scenario, we can ask the following question:
\newline

If we prepare the system in the state \textbf{A} and eliminate the walls that separate the compartments, will the system evolve to state \textbf{B}?
\newline

We would answer negatively, since it is very unlikely that molecular diffusion concentrates the \emph{balls} in the first compartment.
However, the Shannon entropy of \textbf{B} is larger than that of \textbf{A}. Hence, the \emph{second law} would suggest that a very unlikely transition happens spontaneously.

\begin{figure}[!tb]
\centering
\begin{tikzpicture}
\node[rounded corners, draw, text height = 2cm, minimum width = 9cm] {};
\node[rounded corners, draw, text height = 2cm, minimum width = 9cm,yshift=-3.5cm] {};

\draw [dashed] (-1.5,1) -- (-1.5,-1);
\draw [dashed] (1.5,1) -- (1.5,-1);
\draw [dashed] (-1.5,-2.5) -- (-1.5,-4.5);
\draw [dashed] (1.5,-2.5) -- (1.5,-4.5);

\node[small dot blue] at (-3,0) {};
\node[small dot blue] at (-4,0.75) {};
\node[small dot blue] at (-2.5,0.75) {};
\node[small dot blue] at (-3.5,-0.85) {};
\node[small dot blue] at (-1.8,-0.3) {};
\node[small dot blue] at (-3.8,-0.3) {};

\node[small dot blue] at (-3+3,0) {};
\node[small dot blue] at (-4+3,0.75) {};
\node[small dot blue] at (-2.5+3,0.75) {};
\node[small dot blue] at (-3.5+3,-0.85) {};
\node[small dot blue] at (-1.8+3,-0.3) {};
\node[small dot blue] at (-3.8+3,-0.3) {};

\node[small dot blue] at (-3,0-3.5) {};
\node[small dot blue] at (-4,0.75-3.5) {};
\node[small dot blue] at (-2.5,0.75-3.5) {};
\node[small dot blue] at (-3.5,-0.85-3.5) {};
\node[small dot blue] at (-1.8,-0.3-3.5) {};
\node[small dot blue] at (-3.8,-0.3-3.5) {};
\node[small dot blue] at (-2.3,-4) {};
\node[small dot blue] at (-2.3,-3.3) {};

\node[small dot blue] at (-3+3,0-3.5) {};
\node[small dot blue] at (-4+3,0.75-3.5) {};

\node[small dot blue] at
(-3.5+6,-3.3) {};
\node[small dot blue] at (-1.8+6,-0.3-3.5) {};

\draw [-implies,double equal sign distance] (-1,-1.2) -- (-1,-2.2);

\draw (-1.3,-2) -- (-0.7,-1.4);

\draw[thick,->]
    (1,-1.2) -- (1,-2.2) node[midway,sloped,yshift=-3mm,font=\fontsize{15}{15}\selectfont, thick, rotate=-90] {$H$};

\end{tikzpicture}
\caption{Two molecular states for a gas in a box where the increase in Shannon entropy $H$ allows transitions which are very unlikely. More specifically, while the top state has lower Shannon entropy, we do not expect to observe a transition from it to the bottom state. (Reproduced from \cite{hack2022disorder}.)}
\label{gas trans}
\end{figure}
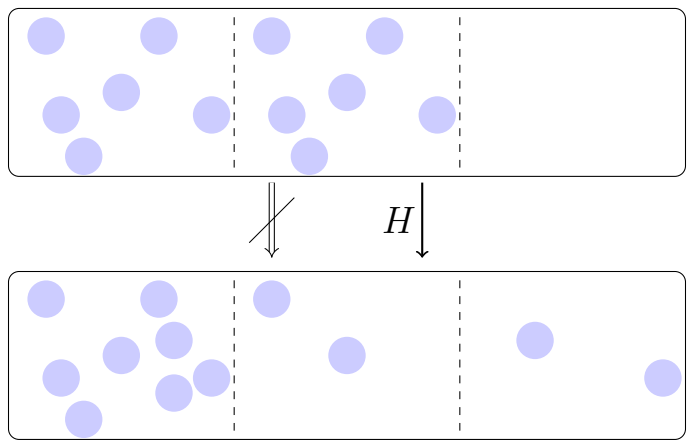

\subsubsection{The second laws of disorder}

Now that we have exposed the improvement on the second law proposed by disorder, we would like to obtain the second law of disorder, that is, a family of functions that emulate the the role of entropy in the second law. We aim to keep the following properties \cite{hack2022disorder}: 
\begin{enumerate}[label=(\roman*)]
\item if a transition is reversible, then entropy remains unchanged, 
\item if a transition is irreversible, then entropy increases, and
\item if a transition is impossible, then entropy decreases.
\end{enumerate}
Following these properties, we are looking for a family of functions $(f_i)_{i \in I}$, $f_i: \mathbb P_\Omega \to \mathbb R$ for all $i \in I$, such that $f_i(p) = f_i(q)$ for reversible transitions $p\sim_M q$, $f_i(p) < f_i(q)$ for irreversible transitions $p \prec_M q$ and at least one function strictly decreases (i.e. there exists some $i\in I$ such that $f_i(p) > f_i(q)$) for forbidden transitions $\neg( p \preceq_M q)$. This motivates the following definition we introduced in \cite{hack2022disorder}.

\begin{defi}[Family of second laws of disorder]
\label{family second}
If $\Omega$ is a finite set, a family of functions $(f_i)_{i \in I}$, $f_i: \mathbb P_\Omega \to \mathbb R$ for all $i \in I$, is a family of second laws of disorder provided it fulfills the following properties:
\begin{enumerate}[label=(\roman*)]
\item if $p\sim_M q$, then $f_i(p)=f_i(q)$ for all $i\in I$, 
\item if $p\prec_M q$, then $f_i(p)<f_i(q)$ for all $i\in I$, and
\item if $p \bowtie_M q$, then there exist $i,j\in I$ with $f_i(p)>f_i(q)$ and $f_j(p)<f_j(q)$.
\end{enumerate}
\end{defi}

We can think of Definition \ref{family second} as a (partial) extension of what Lieb and Ingvason call the \emph{entropy principle} \cite{lieb1998guide,lieb1999physics}. More specifically, as an extension where, instead of only considering pairs of states related by the underlying order relation, we consider arbitrary pairs of states.  

As we remarked in \cite{hack2022disorder}, a family of second laws of disorder for $\mathbb P_\Omega$ is equivalent to a strict monotone (or Richter-Peleg) multi-utility. Hence, the question we ask is whether strict monotone multi-utilities exist for disorder and, moreover, what the minimal number of functions in such a family would be. For example, the first candidate for such a family that comes to mind, $\mathcal U$, is not a strict monotone multi-utility, as we remarked in Section \ref{opt preorders}. In fact, if $|\Omega| =2$, then the scenario where only one function is needed, like in the second law, is recovered. However, if $|\Omega| \geq 3$, then a countably infinite number of functions is actually needed in order for Definition \ref{family second} to be fulfilled. We state that this in the following theorem.

\begin{theo}
\label{dim majo}
The following statements hold:
\begin{enumerate}[label=(\roman*)]
\item If $|\Omega|=2$, then 
\begin{equation*}
    p \preceq_M q \ \iff \ -H(p) \leq -H(q).
\end{equation*}
\item If $|\Omega| \geq 3$, then the smallest family of second laws of disorder is countably infinite.
\end{enumerate}
\end{theo}

\begin{rem}[Applications of the classification of preorders]
\label{useful monos}
Notice Theorem \ref{dim majo} exemplifies the usefulness of
the general classification of preorders we considered during this chapter. In particular, the path that led to it started with the construction of a partial order $(X,\preceq)$ that had a finite geometrical dimension and a countably infinite Debreu dimension (as reported in Theorem \ref{finite geo infinite deb}, see Figure \ref{fig:counterex} for a representation) and naturally led to the question whether our running example throughout this work, majorization, presented the same properties. The last step to obtaining it was, then, to construct a subset of majorization order-isomorphic to $(X,\preceq)$ (see \cite{hack2022disorder}). Hence, it was the interest in the more general question in Theorem \ref{finite geo infinite deb} which led to obtaining Theorem \ref{dim majo}.
\end{rem}


It should be noted that the proof of Theorem \ref{dim majo}, which we split between \cite{hack2022geometrical} and \cite{hack2022disorder}, is purely combinatorial, given that we did not require any sort of continuity property on the functions that constitute the family of second laws of disorder. However, in case we were only interested in that particular case, we can apply \cite[Proposition 5.2]{alcantud2016richter} (which is a consequence of \cite[Theorem 1]{schmeidler1971condition}) to show that finite families of second laws of disorder composed of continuous functions do not exist. (The notion of continuity we are referring to is the one obtained when taking the standard topology for both $\mathbb P_\Omega$ and $\mathbb R$.) 
In the following section, we address the same scenario we considered in this one for a system that, instead of being isolated, is in contact with a heat bath. 

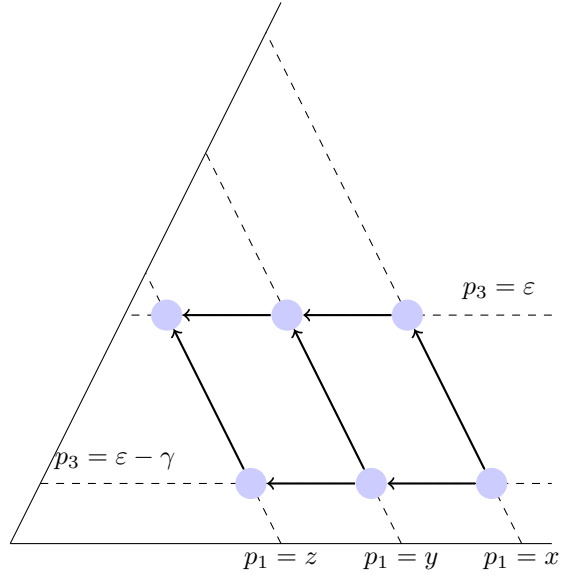
\begin{figure}[!tb]
\centering
\begin{tikzpicture}[scale=0.8]

\draw (-5,-5)--(4,-5);

\draw (-5,-5)--(-0.5,4);


\draw [dashed] (3.5,-5) node[below] {$p_1=x$} -- (-3/4,3/2+2);

\draw [dashed] (1.5,-5) node[below] {$p_1=y$} -- (-7/4,7/2-2);

\draw [dashed] (-0.5,-5) node[below] {$p_1=z$} -- (-11/4,5-11/2);

\draw [dashed] (4,-1.2) node[above left=0.13cm] {$p_3=\varepsilon$} -- (-3.1,-1.2);

\draw [dashed] (-9/2,-4) node[above right=0.1 cm] {$p_3=\varepsilon - \gamma$}-- (4,-4);


\node[small dot 3] at (-12/5,-1.2) (1u) {};

\node[small dot 3] at (-2/5,-1.2) (2u) {};

\node[small dot 3] at (8/5,-1.2) (3u) {};


\node[small dot 3] at (-1,-4) (1d) {};


\node[small dot 3] at (1,-4) (2d) {};

\node[small dot 3] at (3,-4) (3d) {};

\path[draw,thick,->]
    (1d) edge node {} (1u)
    (2d) edge node {} (2u)
    (3d) edge node {} (3u)
    (3d) edge node {} (2d)
    (2d) edge node {} (1d)
    (3u) edge node {} (2u)
    (2u) edge node {} (1u)
    ;

\end{tikzpicture}
\caption{Representation of the subset of majorization that is order isomorphic to the partial order in Figure \ref{fig:counterex}. We introduced this subset in \cite[Theorem 1]{hack2022disorder} in order to prove Theorem \ref{dim majo}. 
We represent the subset inside the 2-simplex as a directed graph, where $p$ is majorized by $q$ whenever there exists a path following the arrows from $p$ to $q$.
This figure can also be thought to represent the ordered subset we introduced in \cite[Theorem 2]{hack2022disorder} in order to prove Theorem \ref{d-majo dim} $(ii)$. (Reproduced from \cite{hack2022disorder}.)}
\label{fig:majo}
\end{figure}

\subsubsection{Systems in contact with a heat bath: The principle of decreasing mixing distance}

Continuing along the lines we just presented above, we consider now the situation where the thermodynamic system we are interested in is not isolated but in contact with a heat bath. In this case, analogously to the situation where the system is isolated, Ruch proposed the \emph{principle of decreasing mixing distance} \cite{ruch1976principle,ruch1978mixing}. The fundamental idea behind this principle is that, in the presence of a heat bath, the long-term behaviour of the gas follows a Boltzmann distribution $p^*$ \eqref{boltzmann} and the transitions the system undergoes decrease the uncertainty relative to $p^*$. Hence, like majorization describes the principle of increasing mixing character, $d$-majorization with $d=p^*$ describes the principle of decreasing mixing distance. Hence, in the following, we refer to $d$-majorization as $d$-disorder.

\subsubsection{The second laws of d-disorder}

In the context of $d$-disorder, we are also interested in finding a family of functions that emulate the role of entropy in the second law. As we argued before definition \ref{family second} in the case of disorder, we arrived at the following definition in \cite{hack2022disorder}.

\begin{defi}[Family of second laws of $d$-disorder]
\label{d second laws}
If $\Omega$ is a finite set, $d \in \mathbb P_\Omega$ and $d(x)>0$ for all $x \in \Omega$, then a family of functions $(f^d_i)_{i \in I}$, $f^d_i: \mathbb P_\Omega \to \mathbb R$ for all $i \in I$, is a family of second laws of $d$-disorder provided it fulfills the following properties:
\begin{enumerate}[label=(\roman*)]
\item if $p\sim_d q$, then $f^d_i(p)=f^d_i(q)$ for all $i\in I$, 
\item if $p\prec_d q$, then $f^d_i(p)<f^d_i(q)$ for all $i\in I$, and
\item if $p \bowtie_d q$, then there exist $i,j\in I$ with $f^d_i(p)>f^d_i(q)$ and $f^d_j(p)<f^d_j(q)$,
\end{enumerate}
where $\preceq_d$ denotes $d$-majorization.
\end{defi}

Definition \ref{d second laws} corresponds, analogously to the case of disorder, to a strict monotone multi-utility for $\preceq_d$. Hence, the question we ask is whether strict monotone multi-utilities exist for $d$-disorder and, moreover, what the minimal number of functions in such a family would be. In Theorem \ref{d-majo dim}, we state that, if either $|\Omega| =2$ and $d$ is not the uniform distribution or
$|\Omega| \geq 3$, then a countably infinite number of functions is actually needed. 

\begin{theo}
\label{d-majo dim} 
If $d$ is a distribution over a finite set $\Omega$ such that $d(x)>0$ for all $x \in \Omega$, then the following statements hold:
\begin{enumerate}[label=(\roman*)]
\item If $|\Omega|=2$ and $d$ is not the uniform distribution, then the smallest family of second laws of $d$-disorder is countably infinite.
\item If $|\Omega| \geq 3$, then the smallest family of second laws of $d$-disorder is countably infinite.
\end{enumerate}
\end{theo}

\begin{rem}[Applications of the classification of preorders]
Note that, in parallel to our comment in Remark \ref{useful monos}, Theorem \ref{d-majo dim} also exemplifies the usefulness of the general classification of preorders we considered during this chapter. The details can be found in \cite[Theorems 2 and 3]{hack2022disorder}.
\end{rem}

Note that the requirement that $d(x)>0$ for all $x \in \Omega$ is not actually a restriction since, in that scenario, we simply eliminate the states where that is not the case. Moreover, if we assume that $d$ is a Boltzmann distribution $d=p^*$, then $p^*(x)=0$ for some $x \in \Omega$ corresponds to a state that can be discarded for physical reasons. In particular, $x$ is a state with infinite energy.

\subsection{Quantum resource theories}
\label{subsec: resource th}

In the previous section, we exemplified in the context of molecular diffusion the usefulness of the general approach we have developed along this chapter. Here, we continue in the realm of physics by briefly addressing its applicability in quantum resource theories.

Resource theory constitutes an area of quantum physics where both incomplete characterizations of preorders by single monotones and multi-utilities have received increasing attention over the last couple of decades \cite{jonathan1999entanglement,nielsen1999conditions,brandao2015second,gour2015resource,winter2016operational}. The basic setup in any resource theory consist of a set of possible states a system $X$ can be in together with a a set of operations one is allowed to perform, which determine the conversion between these states one can achieve. The connection between such an approach and our more abstract framework in this chapter comes from the fact the conversions between states often form a preorder and, hence, a resource theory is based on a preordered space $(X,\preceq)$. 
Resource theories are actively used in quantum physics where, in particular, they are deployed in order to understand several quantum effects
and to draw relations among them by finding formal similarities in their structures. 
In the following, we briefly address one of the main resource theories, namely entanglement  \cite{nielsen1999conditions,nielsen2001majorization,nielsen2002introduction}.

\subsubsection{Entanglement}

In order to avoid an introduction to the concepts in quantum mechanics, which lie outside the scope of this work, we simply include, roughly, how majorization is related to entanglement.
A well-known resource theory, which is relevant in the study of the quantum phenomenon known as entanglement \cite{nielsen2002introduction}, restricts the operations one can perform on the quantum systems to the so-called
\emph{local operations and classical communications}. Roughly, the usual setup there consists of two parties that possess quantum states which are correlated and are only allow to (i) perform operations on their individual systems and (ii) communicate with each other via classical communication channels \cite{nielsen2002quantum}.
These operations are connected to majorization by Nielsen's theorem \cite{nielsen1999conditions} which, basically, shows they are equivalent to majorization for certain probability vectors that are associated to the quantum state of the system. Thus, our results regarding molecular diffusion and majorization are also applicable to the study of entanglement. In conclusion, our tools also turn out to be useful in studying the complexity of the conversion between quantum states. The reader interested in the topic can consult a recent review by Sagawa \cite{sagawa2022entropy} which deals with the role of majorization in physics from a broad perspective.

To conclude, let us include a few words on a variation of entanglement that is regarded as ill-understood and where our abstract approach could be useful, namely entanglement catalysis \cite{jonathan1999entanglement,turgut2007necessary}.

\subsubsection{Entanglement catalysis}

In the same spirit of entanglement, entanglement catalysis 
considers the conversion between quantum states through local operations and classical communications and, moreover, allows the use of an auxiliary system or \emph{catalyst}. As in the case of entanglement and majorization, this resource theory can be reduced to a preorder, known as \emph{trumping}. In particular, for a finite $\Omega$, trumping $\preceq_T$ is a preorder on
$\mathbb{P}_{\Omega}$ where
     \begin{equation}
     \label{trumping def}
        p \preceq_T q \iff \exists r \in \mathbb{P}_{\Omega'} \text{ such that } p \otimes r \preceq_M q \otimes r,
\end{equation}
for all $p,q \in \mathbb{P}_{\Omega}$, where $p \otimes r:= (p_1r_1,\dots,p_1r_{\Omega'},\dots,p_{\Omega}r_1,\dots,p_{\Omega} r_{\Omega'})$ for all $p \in \mathbb{P}_{\Omega}$ and $r \in \mathbb{P}_{\Omega'}$, and $\Omega'$ is some finite set. (Note that $r$ in \eqref{trumping def} is known as a \emph{catalyst}.)

Given that trumping remains to be well understood in terms of real-valued functions \cite{turgut2007catalytic,turgut2007necessary}, it is our hope that abstract approaches like the one in this chapter may shed some light on it.  
As a final remark, note that trumping represents an anomaly regarding preorders since, unlike majorization for instance, it is not directly defined through a family of real-valued functions.

\section{Summary}

In this chapter, we addressed the substitution of order structures by real-valued functions. We started considering the substitution of the uncertainty preorder by the maximum entropy principle. This led us to the study of optimization, complexity and their interplay in order structures. Despite having quite an abstract approach to the subject, we showed how such a treatment may be helpful in concrete applications. More specifically, we related our more general framework to the maximum entropy principle and the uncertainty-based approaches to both molecular diffusion and quantum resource theories.

In the following chapter, we consider an abstract order-theoretic approach to computation based on the intuitive idea of uncertainty reduction. In particular, we study the specific requirements on such a structure that allow it to provide a notion of \emph{finite description} which is based on Turing machines. Hence, continuing our abstract approach to decision-making, we considering the interplay between uncertainty and computation through an order structure. 

\newpage
\thispagestyle{empty}
\mbox{}
\newpage
\chapter{Uncertainty and computation}
\label{chap:Uncertainty and computation}


The link we established between the local transitions in thermodynamics and learning systems was in terms of uncertainty. In particular, we modelled both transitions through a specific preorder on the space of probability distributions, namely majorization. This led us to equating decision-making with uncertainty reduction and dealing with them abstractly through order structures. In this chapter, we follow this intuition and use order structures to capture a notion of relative uncertainty between elements that, ultimately, allows us to introduce a notion of \emph{finite description} for decision-making. In particular, we translate computability from Turing machines (i.e. from the natural numbers) to other spaces where we intend to perform some decision-making process (e.g. determine a zero of a function), like the real numbers. 

We can think of any algorithm as a sequence of steps that consecutively gather information about the answer to some problem of interest. The strategy followed by the algorithm to find a solution can be considered in a similar way to how we pictured the bisection method through
$(\mathcal I, \preceq_{\mathcal I})$. That is, we think of the algorithm as a series of transitions between the elements of some representation of the solution space. Moreover, we assume that these transitions behave as a preorder that relates the elements in the representation in terms of how much \emph{information} or relative \emph{uncertainty} they possess. Aside from these considerations, for a set of instructions to be considered an algorithm, there must be a finite description of them that allows one to perform them without any ingenuity (i.e. they must be simple and clear enough). 
In this regard, we can
strengthen the order structure to add finiteness to our view of preorders as abstract relations representing core strategies that can be followed by an algorithm. In particular, we can build an order-theoretic picture were notions of computability (i.e. finite description) for elements and functions are defined. The purpose of this chapter is to explore the requirements on such an order structure.

The importance of developing an order-theoretic framework for computability lies in the fact it could help to sharpen the notion of computation on uncountable spaces, which is known to presents difficulties (see Section \ref{sec: issue uncoun comp} for more details).

The reader should note that, in this chapter, we summarize and supplement our work in \cite{hack2022computation} and \cite{hack2022relation}.

\section{Computation on countable and uncountable sets}
\label{compu count and uncount}

\subsection{Computation on countable sets}

The classical notion of computation was introduced by Turing
\cite{turing1937computable,turing1938computable}. Turing's approach is based on an idealized machine, namely the \emph{Turing machine}. Roughly speaking (a formal definition can be found in \cite{hopcroft2001introduction}), a Turing machine consists of:
\begin{enumerate}[label=(\roman*)]
    \item A tape, which is divided into cells with each cell containing a single symbol from a finite alphabet $\Sigma$.
    \item A head, that can read and write symbols on the tape's cells.
    \item A state register, which stores the state that the Turing machine has at every step and, hence, the actions regarding the head that the machine may preform.
    \item A finite table of instructions, which contains the rules regarding what the following state of the Turing machine may be depending on its current state and the tape's input it obtains from the head. 
\end{enumerate}
The fundamental idea that the Turing machine aims to formalize is that of \emph{finite instructions}, which is fundamental in (iv). In particular, although the tape is allowed to be infinite and, hence, both the input and the output of the machine may be infinite, the instructions that regulate its performance are fixed and finite. This corresponds to the fundamental idea that one can obtain any arbitrary amount of finite information about several infinite objects provided they are sufficiently \emph{regular}, that is, provided they can be described in some finite manner. An example of this would be the set of natural numbers or the set of even natural numbers, since there exists a recursive way of defining both of them and, hence, there is a straightforward Turing machine for each set that (i) outputs only elements that belong to the set and (ii) eventually generates any element in the set.

Turing's effort was successful on the natural numbers. In fact, Turing machines are still regarded as the canonical model of computation on them \cite{weihrauch2012computability,weihrauch2012computable} (see the \emph{Church-Turing} thesis below). We summarize the formalization of computability on the natural numbers they provide in the remainder of this section (see, for example, \cite{rogers1987theory} for a more detailed account). A function $f:\mathbb{N} \rightarrow \mathbb{N}$ is \emph{computable} if there exists a Turing machine which, for all $ n \in \mathbb{N}$, \emph{halts} (i.e. finishes after some finite time) on input $n$ and returns $f(n)$ \cite{rogers1987theory}.\footnote{What we call a \emph{computable} function is known as a \emph{total recursive function}, which contrast with functions where $\text{dom} (f) \subset \mathbb{N}$ holds \cite{rogers1987theory}, which we call \emph{partially computable}. More in general, $f:X \to Y$ is \textit{partial} if it is not defined on the whole $X$.}
Moreover, a subset $A \subseteq \mathbb{N}$ is said to be \emph{computable} if there exists a computable function $f:\mathbb{N} \rightarrow \mathbb{N}$ such that $f(x)=1$ if $x \in A$ and $f(x)=0$ if $x \not \in A$. (For instance, both the natural numbers and the even natural numbers are computable sets.) 

There is a weaker notion of computable sets.
A subset $A \subseteq \mathbb{N}$ is said to be \emph{recursively enumerable} or \emph{computably enumerable} if $A=\emptyset$ or there exists a computable $f$ such that $A=f(\mathbb{N})$. The recursively enumerable subsets of $\mathbb{N}$ are those whose elements can be produced in finite time.
It should be noted that recursively enumerable sets which are not computable exist \cite{rogers1987theory}.

\subsubsection{The Church-Turing thesis}

The established status of Turing's approach is due to the fact the subsequent attempts to formalize computation on the natural numbers, like universal register machines \cite{cutland1980computability} or lambda calculus \cite{barendregt1984lambda}, turned out to be equivalent to it \cite{rogers1987theory}.
The equivalence between these different formalizations led to the formulation of the so-called \emph{Church-Turing thesis} or \emph{Church thesis} \cite{rogers1987theory}. This \emph{thesis} states that a function between the natural numbers $f: \mathbb N \to \mathbb N$ is \emph{effectively calculable} if and only if it is computable in the Turing machine sense. Note that we say a function is effectively calculable \cite{copeland1996computation,turing1937computable} if
\begin{enumerate}[label=(\roman*)]
\item it consists of a finite number of exact instructions such that each of them can be expressed by a finite amount of symbols,
\item it produces the desired result after a finite number of steps,
\item it can be carried, in principle, by a human using only paper and pencil and
\item it requires no ingenuity from the human carrying out the method.
\end{enumerate}
Thus, the definition of effective calculability aims to clarify our intuitive notion of \emph{finite instructions} and, hence, the Church-Turing thesis connects (in fact, equates) this informal notion to the formal one provided by Turing machines. 

Although the computability landscape for countable sets is well-established, the situation regarding uncountable sets is quite different. We give  a brief account of this in the following section.

\subsection{The issue with computation on uncountable sets}
\label{sec: issue uncoun comp}

Despite the approach-invariant nature of computation on the natural numbers, uncountable spaces like the real numbers behave differently. In particular, significant differences between certain proposed models have been unveiled. Notably, the naive approach to computability on the real numbers through their decimal representation deems as uncomputable some elementary functions which turn out to be computable in other approaches.

Originally, Turing proposed to use the decimal expansion of the real numbers and, consequently, defined the computable real numbers as those whose \emph{decimal expansion can be calculated by finite means} \cite{turing1937computable}. However, as realized by him, some elementary functions like multiplication are not computable in this scenario. For instance, as argued in \cite{di1996real}, one can directly see there is no algorithm that computes the multiplication by 3: If it existed, such an algorithm should generate the first digit in the decimal expansion of the output after observing a finite amount of digits in the input. Let us assume, with input $0.3333\dots$, the output is $1.0000\dots$ (we can make the same argument if we choose $0.9999\dots$ as output). In this scenario, the single one that belongs to the output ought to be generated after examining, say, $n$ digits from the input. Given the input $0.3^{n)}\dots300\dots$ that contains exactly $n$ threes, the algorithm would again choose one as the first digit of the output. This results in an incorrect computation, since the output should be $0.9^{n)}\dots900\dots$, with exactly $n$ nines. The same holds true for any representation that only differs from the decimal in the base, like the binary representation used in digital computers.

The problem in the previous paragraph can be fixed by using a representation of the real numbers that belongs to a different class. For example, we can take one belonging to the \emph{negative-digit representations} \cite{di1996real,wiedmer1980computing,avizienis1961signed}. The difference between these representations and the ones considered above is that they allow negative digits in their expansions. Because of that, multiplication by 3 turns out to be computable in the latter. In fact, roughly, the main practical difference between these approaches lies in the bounds on the final result established by the outputs obtained at finite times (the representations with negative digits allow us to reduce these bounds in the following steps by outputting negative numbers). Note that this property regarding bounds is what we exploit in the example above, where we also use the fact that, in order to converge, the algorithm must fix each digit of the output after only examining a finite subset of the representation of the input. 


Currently, as a result of the discussion in the previous paragraph, there is no canonical model for computation on the real numbers \cite{weihrauch2012computability,weihrauch2012computable}. Because of this, the study of the different approaches towards computation on uncountable spaces continues to be pursued. It is here where order structures play a role, as we will see in the following section. 

\section{An order-theoretic approach to computation on uncountable sets}

Among the proposed models for computation on uncountable sets, like the decimal or negative-digit representations, there exist some which are based on preorders \cite{scott1970outline,edalat1999domain,edalat1997domains}. The basic idea in this approach is to use an order structure to translate the formal notion of computation from the natural numbers (i.e. from Turing machines) to uncountable sets. The fundamental structure, thus, consists of a preordered set $(X,\preceq)$, where $X$ is a representation set of the solution space to some problem and $\preceq$ relates the elements in $X$ in terms of information. These approaches belong to a field called \emph{domain theory} \cite{abramsky1994domain,scott1970outline}.

Despite being a successful theory, we believe that the intuitive introduction of the order-theoretic framework in domain theory does not justify the strong order-theoretic requirements it imposes. This led us to proposing a more general order-theoretic framework in which the assumed structure is better justified. As we will see in Section \ref{discuss order comp}, where we compare our approach to domain theory, we gain on the intuitive side, although we loose some useful properties. This will lead us to establishing formal reasons why the stronger assumptions in domain theory are required. Before doing so, in the following section, we give an intuitive explanation of what the order-theoretic approaches to computation aim to capture and we introduce our more general framework.

\begin{rem}[Domain theory outside computer science]
Following the interdisciplinary aim of this work, the reader should note that, as in the case of real-valued representations of preorders (see Section \ref{sec:mu charact}), domain theory has been applied in disciplines outside computer science. For example, in general relativity \cite{martin2006domain,martin2010domain,martin2012spacetime}, quantum mechanics \cite{coecke2010partial,rennela2014towards} and information theory \cite{martin2008topology,chatzikokolakis2008monotonicity}.
In fact, some of these applications have been compiled in a recent book concerned with new mathematical structures that may become useful in physics \cite{coecke2011new}.
\end{rem}

\subsection{The fundamental intuitive picture:\\ Countable distinguishing properties}
\label{intuition compu}

Contrary to the rest of this work, instead of taking a preordered space $(X,\preceq)$ as starting point, in this section we begin by considering a set of features that distinguish the elements in $X$ and define the order structure $\preceq$ through them.

The fundamental trait of some uncountable space $X$ that allows us to describe it is the existence of a countable family of \emph{features} or \emph{properties} $(O_n)_{n \geq 0}$ that enable us to distinguish among the elements in $X$. Following the idea behind Dedekind cuts \cite{dedekind1901essays}, an example of this is the characterization of any real number by listing all the rational numbers below it.
More specifically, an element $x \in X$ should be the only element in $X$ that possesses the features in (and only in) some subfamily $O_x \subseteq (O_n)_{n \geq 0}$. Hence, determining the element $x$ is equivalent to listing all (and only) the features that define it, that is, the properties that belong to $O_x$. For instance, listing the rational numbers below some real number is equivalent to determining it. In fact, we can use the countable family of features $(O_n)_{n \geq 0}$ to build a binary relation $\preceq$ among the elements of $X$ in the following way:
\begin{equation*}
x \preceq y \iff (\text{for all } n \geq 0) \text{ } y \text{ fulfills } O_n \text{ provided that } x \text{ fulfills } O_n.    
\end{equation*}
We call $\preceq$ an \emph{information} relation on $X$, given that, since we think of $O_n$ as a feature for each $n \geq 0$, $x \prec y$ implies that $y$ possesses all the features that define $x$ and at least one more of them is needed in order to determine $y$. Hence, $y$ requires more \emph{information} to be specified. This information relation will play a key role in the reminder of this chapter.

As explained in the previous paragraph, provided there exists a countable set of features that allows us to distinguish between the elements in an uncountable space we are interested in, determining an element is equivalent to providing a list consisting of its features. Another way of determining an element $x \in X$ would consist of a procedure that, for any arbitrary but finite number of features of $x$, provides us an element $y \in X$ that possesses all these features (and no features that $x$ does not share). In this scenario, the features defining $x$ would be the intersection of all the features defining the elements provided by the procedure. Although we will sharpen the notion of \emph{procedure} later on, for the moment, let us just say it consists of a sequence of elements $(x_n)_{n\geq 0} \subseteq X$ where each $x_n$ can be represented using finite space (this will play a key role in the following) and such that, for any finite subfamily $O'_x \subseteq O_x$, there exists some element $x_n$ that possesses the properties in $O'_x$.
However, in case we have such a procedure $(x_n)_{n \geq 0}$, how do we associate it to an element $x$? Following the idea behind the \emph{principle of insufficient reason} (see Section \ref{subsec: maximum entr}), we associate to $(x_n)_{n \geq 0}$ the element $x$, if it exists, that, for all $n \geq 0$, fulfills the properties that $x_n$ does and is \emph{minimal} (i.e. any other element sharing these properties would require more features than $x$ to be determined). We denote such an $x$ by $\sqcup (x_n)_{n \geq 0}$. Hence, we have a way of associating an element $x$ (which may have no finite representation) to a list $(x_n)_{n \geq 0}$ conformed by elements which have a finite representation. This brings us closer to obtaining a notion of \emph{finite instructions} for $x$, which is key in the following sections of this chapter.

Despite the connection between $(x_n)_{n \geq 0}$ and $x$ in the previous paragraph, we continue having the issue that, whatever finite number of features $O'_x$ we choose, we would like our procedure to contain some element $x_n$ that fulfills all these features. To remedy this, 
we will only consider a specific sort of procedures. If we were only interested in the case where $|O'_x|=1$, we could require the family $(O_n)_{n \geq 0}$ to fulfill that, if $x$ has some feature $O \in (O_n)_{n \geq 0}$ and a procedure $(x_n)_{n \geq 0}$ is associated to $x$, then we ought to have some element that possesses that feature, that is, if $\sqcup (x_n)_{n \geq 0} = x$, then there exists some $n \geq 0$ such that $x_n$ fulfills $O$. Nonetheless, we are interested in finding elements that have any finite number of properties of $x$ and, if we follow the previous proposal, then we may encounter a situation where there is no way of obtaining an element with more than one feature of $x$ (an instance of this can be found in Remark \ref{why directed}). To avoid this, we only consider certain procedures. More specifically, those $(x_n)_{n \geq 0}$ where, for any pair of elements $x_n,x_m \in (x_n)_{n \geq 0}$, there exists a third element $x_p \in (x_n)_{n \geq 0}$ that gathers all their features $x_n,x_m \preceq x_p$. We refer to those procedures as \emph{directed}. For such procedures, we can obtain elements with any finite number of features of $x$ by simply requiring the features to fulfill the property regarding $|O'_x|=1$ that we explained above. In conclusion, we require the features to fulfill that, if $(x_n)_{n \geq 0}$ is a directed procedure, $\sqcup (x_n)_{n \geq 0}=x$, and $O$ is a feature of $x$, then there exists some $n \geq 0$ such that $x_n$ fulfills $O$.

This section illustrates again the uncertainty basis of our computational picture. In particular, we think of a procedure $(x_n)_{n \geq 0}$ as a \emph{computation} of $x=\sqcup (x_n)_{n \geq 0}$ that proceeds by unveiling the features $O_x \subseteq (O_n)_{n \geq 0}$ that define $x$, that is, by \emph{reducing uncertainty} with respect to $x$. In this regard, we can think of computation as a decision-making process \cite{gottwald2019bounded}. 

In the following section, instead of the family of features $(O_n)_{n \geq 0}$, we take $\preceq$ as the starting point and we deal with a specific computational example to refine our notion of \emph{procedure}.

\subsection{An intuitive example: The bisection method}
\label{bisection ex}

To complement the introduction to the order-theoretic approaches to computation from Section \ref{intuition compu}, we return to the example of the bisection method and $(\mathcal I,\preceq_{\mathcal I})$ (see Chapter \ref{intro}). In particular, we provide some more details regarding this example in order to easily transition from it to the general case.

The outcome of each computation step in the bisection method is an interval $I_n = [a_n,b_n] \in \mathcal I$ with $a_n,b_n\in\mathbb Q$. Hence, a computation outputs a sequence $(I_n)_{n \geq 0}$ inside $\mathcal I$. In the following two paragraphs, we detail the two key properties of such sequences.

The sequences $(I_n)_{n \geq 0}\subset B$, where $B \subseteq \mathcal I$ is the set of interval with rational endpoints, that result from  applying the bisection method have a \emph{finite} description in the sense that:
\begin{enumerate}[label=(\roman*)]
\item each element $I_n$ has a finite description (its endpoints), and 
\item the sequence $(I_n)_{n \geq 0}$ is generated by using the bisection method, which also has a finite description.
\end{enumerate}
To generalize (i) to arbitrary spaces the existence of an effectively calculable surjective map $\alpha:\mathbb N\to B$ is key.
Regarding (ii),
, and following the introduction to Turing machines from Section \ref{compu count and uncount},
$\alpha$ enables us to translate the finite description (via Turing machines) of sequences in $\mathbb N$ to sequences in $B$.  

We can call a number $x\in\mathbb R$ \emph{computable} via the bisection method if an instance of the method $(I_n)_{n \geq 0}$ \textit{convergences} (or reduces  \emph{information} or \textit{uncertainty}) to $x$. By this, we mean that $I_{n+1}\subseteq I_n$ for all $n \geq 0$ and that, if we associate $x\in \mathbb R$ with $[x,x]\in \mathcal I$, then $[x,x] = \inf_{n \geq 0} I_n$, where we understand the \emph{infimum} in the set inclusion sense. Hence, we relate sequences in $B$ with real numbers.

To translate the previous notions to the general case, we need to replace \cite{hack2022computation}:
\begin{enumerate}[label=(\roman*)]
\item $\mathcal I$ and $\mathbb R$ by general sets $Y$ and $X$, respectively,
\item the partial map $f:\mathcal I \to \mathbb R$ by a surjective partial map $\rho:Y\to X$,
\item $B$ by a countable $B_Y \subseteq Y$, that contains finite sequences $(b_n)_{n \geq 0}$ which we can associate to $y \in Y$ such that $\rho(y)$ exists,
\item the finite instructions by Turing machines which translate finiteness to $B_Y$ via some effectively calculable $\alpha:\mathbb N\to B_Y$,
\item set inclusion on $\mathcal I$ by a partial order $\preceq$ on $Y$ which connects $Y$ with sequences in $B_Y$.
\end{enumerate}


Given a surjective partial map $\rho:Y\to X$ and the points (i)-(v), we say $x\in X$ is computable if there exists some $y \in Y$ such that $x=\rho(y)$ and there is some $(b_n)_{n \geq 0}\subseteq B_Y$ with a finite description (that is, $\alpha^{-1}((b_n)_{n \geq 0})$ is the output of some Turing machine) that converges in $\preceq$ to $y$.

In the following section, we formalize the intuitive introduction that we have attempted in this one.

\subsection{The formal order-theoretic approach}
\label{order in compu}

Since computation on $X$ is introduced via Turing machines, it needs to be somehow related to the subsets of $\mathbb N$ which are the output of some Turing machine. Such a relation is obtained through a \emph{finite map} $\alpha: \mathbb{N} \to B$, which we defined in \cite{hack2022computation}.

\begin{defi}[Finite map]
\label{def:finite map}
A map $\alpha: \text{dom}(\alpha) \to A$, where $\text{dom}(\alpha) \subseteq \mathbb{N}$, is finite for $A$ or simply a finite map if $\alpha$ is bijective and both $\alpha$ and $\alpha^{-1}$ are effectively calculable.\footnote{Note that $\alpha$ does not need to be defined on all natural numbers.}
\end{defi}

We define the \emph{computable} subsets of some countable $B$ as those $B' \subseteq B$ such that $\alpha^{-1}(B')$ is the output of some Turing machine.
Through a finite map, we translate computability from Turing machines to the subsets of $B \subseteq X$. However, to translate it to the (potentially uncountable) $X$, 
we use $\preceq$ and associate subsets of $B$ to elements in $X$. To make this association, we first require the computational processes that consistently gather information in $X$ (i.e. those that eventually reduce uncertainty) to converge in $X$. That is, we require $(X,\preceq)$ to be \emph{directed complete} \cite{abramsky1994domain}. (If $A \subseteq X$ is a \emph{directed set}, the we require it to have a \emph{supremum} $\sqcup A$, where we say a set $A \subseteq X$ is \emph{directed} if, for all $a,b \in A$, there exists some $c \in A$ such that $a \preceq c$ and $b \preceq c$, and we say an element $\sqcup A \in X$ is the \emph{supremum} of $A$ provided we have (i) $a \preceq \sqcup A$ for all $a \in A$ and (ii), given some $b \in X$ such that $a \preceq b$ for all $a \in A$, then $\sqcup A \preceq b$ also holds.) Hence, since $\preceq$ represents information, $\sqcup A$ is the \emph{minimal} element (in information) that possesses all the information in $A$. Hence, we consider any computational process outputting $A$ to be computing $\sqcup A$.

To conclude the definition, we limit the computational processes we consider to those in $B$. More specifically, we consider a partial order $(X,\preceq)$ that is directed complete (a \emph{dcpo}) and has a countable $B \subseteq X$ (related to $\mathbb N$ through some finite $\alpha: \mathbb N \to B$) and, for each $x \in X$, potentially contains a computational process that leads to $x$. We specify this in the following definition \cite{hack2022computation}.

\begin{defi}[Weak basis]
\label{def:weak basis}
A subset $B \subseteq P$ of a dcpo $P$ is a weak basis if, for each $x \in P$, there exists a directed set $B_x \subseteq B$ such that $x=\sqcup B_x$.
\end{defi}

We can define computability on $X$: If $X$ is a dcpo with a countable weak basis $B \subseteq X$, then $x \in X$ is \emph{computable} if there exists a Turing machine whose output is $\alpha^{-1}(B_x)$ for some subset $B_x \subseteq B$ such that $\sqcup B_x=x$. However, we still need to address an important aspect:
To gather the information in $\sqcup A$ from some $A \subseteq B$,
we need to know how the different chunks of information in $A$ are related to each other in order to approximate the information in $\sqcup A$. (That is, we pretend to provide, after any finite amount of time, the best approximation of $x$ so far by considering the information relationship among the elements in $A$ that have been outputted.) 
To achieve this, we ask for the dcpo to have a specific sort of weak basis, which we call \emph{effective} in the following definition \cite{hack2022computation}.

\begin{defi}[Effective weak basis]
\label{def:eff weak basis}
A countable weak basis $B \subseteq P$ of a dcpo $P$ is effective if there exist both a finite map for $B=(b_n)_{n\geq0}$ and a computable function $f:\mathbb{N} \to \mathbb{N}$ such that $f(\mathbb{N}) \subseteq \{\langle n,m \rangle| b_n \preceq b_m\}$ and, for each $x \in P$, there is a directed set $B_x \subseteq B$ such that $\sqcup B_x = x$ and, if $b_n,b_m \in B_x\setminus\{x\}$,
then there exists some $b_p \in B_x$ such that $b_n,b_m \prec b_p$ and $\langle n,p \rangle, \langle m,p \rangle \in f(\mathbb{N})$.\footnote{$\langle \cdot,\cdot\rangle $ corresponds to a \emph{pairing function}, i.e., a computable bijection $\langle \cdot,\cdot\rangle: \mathbb{N} \times \mathbb{N} \rightarrow \mathbb{N}$ (which can be defined similarly to how we defined computable functions $f:\mathbb{N} \rightarrow \mathbb{N}$). In the following, we take $\langle n,m \rangle \coloneqq \frac{1}{2} ( n^2+ 2nm + m^2 + 3n + m)$, the \emph{Cantor pairing function}. Note that we will denote the inverses of the Cantor pairing function by $\pi_1$ and $\pi_2$, i.e., $\pi_1(\langle n,m \rangle)=n$ and   $\pi_2(\langle n,m \rangle)=m$.}
\end{defi}

Now that we have clarified the properties we require, we formalize the definition of computable elements \cite{hack2022computation}.

\begin{defi}[Computable element]
\label{def:compu ele}
If $P$ is a dcpo, $B\subseteq P$ is an effective weak basis and $\alpha$ is a finite map for $B$, then
an element $x \in P$ is computable if there exists some $B_x \subseteq B$ such that the properties in Definition \ref{def:eff weak basis} are fulfilled and $\alpha^{-1}(B_x) \subseteq \mathbb{N}$ is recursively enumerable.
\end{defi}

In summary, we start from $\preceq$ and a countable $B \subseteq X$and obtain a notion of computable element for an uncountable space $X$ that is grounded on Turing machines.

Before we continue by providing some examples, let us make a brief remark regarding directed sets.

\begin{rem}[Why directed sets?]
\label{why directed}
The fundamental idea in our approach to computability is that of an information relation $\preceq$ which defines a notion of convergence $\sqcup $ that is key in order to deal with uncountable sets. However, having a set $B_x$ that converges to some element $x$ is not sufficient to actually compute it, since we should be able to provide elements that are arbitrarily close to $x$ to any degree of precision required. In order to do so, we need to be able to distinguish the internal structure of $B_x$ in terms of $\preceq$. In particular, we should be able to differentiate, among the elements in $B_x$, those that contain more information about $x$. This is why we require some finiteness constraint on $\preceq$ in Definition \ref{def:eff weak basis}. Hence, if we simply ask for the existence of a subset $B_x \subseteq B$ such that $\sqcup B_x=x$ for $x$ to be computable, we may have no way of distinguishing among the elements in $B_x$ and, hence, we may not able to provide elements that have arbitrary information regarding $x$. A very pathological example of such a behaviour would be the dcpo $(\mathbb N \cup p,\preceq)$, where $p$ has no finite representation and $x \preceq y$ if and only if $x=y$ or $y=p$. In this scenario, although $B_p \coloneqq \mathbb N$ converges to $p$, there is no way of providing approximations to $p$ with arbitrary precision. 

We should also note that the notion of approximation we use is based on the Scott topology (see Section \ref{uniform compu}). Hence, it is intimately connected to directed sets. However, there does not seem to be a satisfactory notion of convergence that can be similarly defined via the suprema of subsets in general (which would correspond the the notion of computation we have put forward in this remark). 
Finally, note that the case where the condition on directed sets is substituted by one on  chains is discussed in \cite{abramsky1994domain}.
\end{rem}

\subsubsection{Examples}

An example of the objects in this section is given by the notion of computable elements that $(\mathcal I, \preceq_{\mathcal I})$ (a directed complete partial order such that the subset of compact intervals with rational endpoints can be used as effective weak basis $B$) introduces into $\mathbb R$. As a matter of fact, the notion of computation on $(\mathcal I, \preceq_{\mathcal I})$ from this section includes all possible applications of the bisection method. To illustrate this, take $P_{[0,1]}$ the family of polynomials with rational coefficients that have a unique zero in $[0,1]$. The first thing to notice is that we can take as finite map $\alpha$, an enumeration of the rational intervals in $[0,1]$ that can be encountered when running the bisection method with $[0,1]$ as starting interval on any polynomial $p_0 \in P_{[0,1]}$. It is then not difficult to construct for each $p_0 \in P_{[0,1]}$ a Turing machine $f_{p_0}: \mathbb{N} \to \mathbb{N}$ whose output has the property that $\alpha(f_{p_0}(\mathbb{N}))$ is the set of intervals that are obtained when running the bisection method on $p_0$ starting with $[0,1]$. This occurs since the conditions in the bisection method can be rewritten in a way such that they only involve operations on the natural numbers.  

Aside from $(\mathcal I, \preceq_{\mathcal I})$, an important example is the \emph{Cantor domain} or \emph{Cantor set model} \cite{blanck2008reducibility,martin2000foundation}. If $\Sigma$ is a finite set of symbols, $\Sigma^*$ is the set of finite strings in $\Sigma$ and $\Sigma^\omega$ is the set of countably infinite sequences, then their union together with the prefix order $(\Sigma^{\infty},\preceq_C)$
is called the \emph{Cantor domain}
\begin{equation}
\label{Cantor domain}
\begin{split}
    \Sigma^\infty &\coloneqq \Big\{x\Big|x:\{1,..,n\} \to \Sigma,\text{ }0\leq n\leq \infty\Big\}, \\
    x \preceq_C y &\iff |x| \leq |y| \text{ and } x(i)=y(i) \text{ } \forall i \leq |x|,
   \end{split}
\end{equation}
where $|s|$ is the cardinality of the domain of $s \in \Sigma^\infty$.
Note that  $\Sigma^*$ is an effective weak basis (see \cite[Proposition 8]{hack2022computation}). If $\Sigma = \{0,1\}$, then the Cantor domain is the binary representation of real numbers in $[0,1]$. (See Figure \ref{cantor dom} for a representation of the Cantor domain with $\Sigma = \{0,1\}$.)


\begin{figure}[!tb]
\centering
\begin{tikzpicture}
    \node[other node] (1) {$\perp$};
    \node[other node] (2) [above right = 1cm and 1.5cm  of 1]  {$1$};
    \node[other node] (3) [above left = 1cm and 1.5cm  of 1]  {$0$};
    \node[other node] (4) [above left = 1cm  of 3]  {$00$};
    \node[other node] (5) [above right = 1cm  of 3]  {$01$};
    \node[other node] (6) [above left = 1cm  of 2]  {$10$};
    \node[other node] (7) [above right = 1cm  of 2]  {$11$};

   \path[draw,thick,->]
    (1) edge node {} (2)
    (1) edge node {} (3)
    (3) edge node {} (4)
    (3) edge node {} (5)
    (2) edge node {} (6)
    (2) edge node {} (7)
    ;
    \path[draw,thick,dotted]
    (4) edge node {} (-6,7)
    (4) edge node {} (-3,7)
    (5) edge node {} (-3,7)
    (5) edge node {} (0,7)
    (6) edge node {} (0,7)
    (6) edge node {} (3,7)
    (7) edge node {} (3,7)
    (7) edge node {} (6,7)
    ;
    \draw [line width=0.75mm] (-6.0,7.0) -- (6.0,7.0);
    \node [font=\fontsize{15}{15}] at (0,7.5) {$\{0,1\}^*$};
\end{tikzpicture}
\caption{Cantor domain with binary alphabet $\Sigma=\{0,1\}$. An arrow from an element $w$ to another $t$ represents $w \preceq_C t$ (we only include the relations between nearest neighbors). The points where the dotted lines starting at an element $w$ and the horizontal line cross delimit the elements of $\{0,1\}^*$ that are above $w$ in the $\preceq_C$ sense. (Reproduced from \cite{hack2022computation}.)}
\label{cantor dom}
\end{figure}
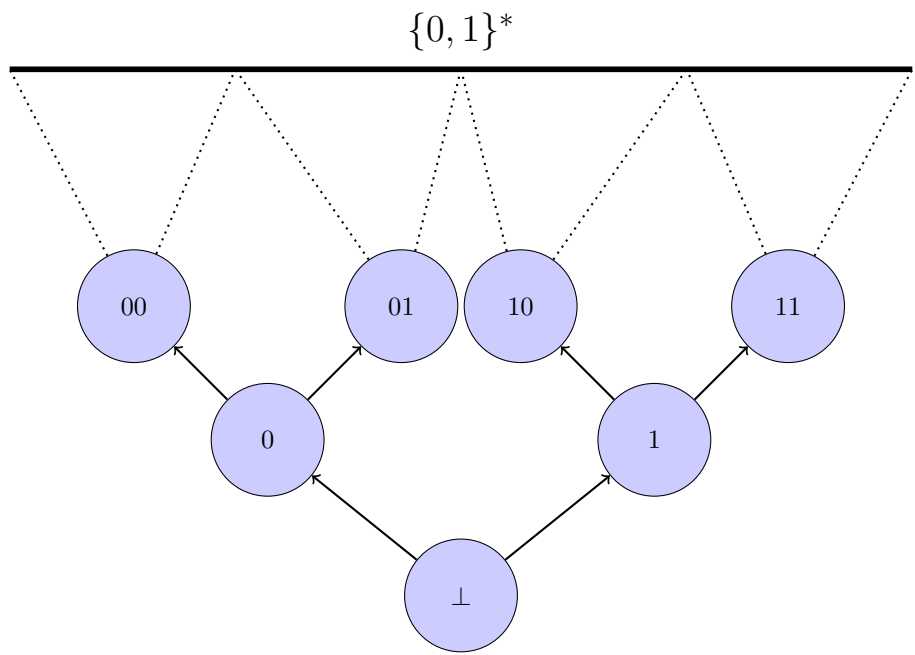

Although well motivated and intuitive, our order-theoretic introduction of computation on uncountable sets is not commonplace in domain theory. In fact, domain theory incorporates more involved assumptions.
In the following section, we present its usual framework and compare it to the one we introduced above.



\section{Properties of order-theoretic approaches to computation}
\label{discuss order comp}

The order structure we introduced does not correspond to the one usually used in domain theory. Hence, the purpose of this section is to improve on our understanding regarding the specific order-theoretical structure needed to enable the definition of computation in uncountable spaces.

\subsection{The uniform order-theoretic approach}
\label{uniform compu}

In this section, we recall the order-theoretic approach to computation in domain theory (see \cite{abramsky1994domain,stoltenberg1994mathematical,stoltenberg2001notes,scott1982lectures,cartwright2016domain} for an introduction and \cite{edalat1999domain,martin2000foundation,mislove1998topology} for further information).
In the following, we refer to the framework in domain theory as \emph{uniform computability} and to ours as the
\emph{non-uniform} approach. We make this distinction since, in our approach, computability was defined as the existence of \emph{any} computational path leading to an element, whereas, in domain theory, its computability is equivalent to that of a \emph{specific} computational path. 

In Section \ref{order in compu}, we considered an element $x$ to be computable provided there was some directed $B_x \subseteq B$ such that both $\sqcup B_x =x$ and $\alpha^{-1}(B_x)$ was recursively enumerable.
We recall now the stronger approach introduced in \cite{scott1970outline} by Dana Scott. More specifically, we
associate to each $x \in P$ a unique $B_x \subseteq B$ such that the recursive enumerability of $\alpha^{-1}(B_x)$ is equivalent to $x$ being computable.
To achieve this goal, $B_x$ must be closely related to $x$. More specifically, the information in each $b \in B_x$ must be obtained by any computational process leading to $x$. To formalize this, we recall the notion of \emph{essential} information (see, for instance, \cite{abramsky1994domain}), which we define in the following. Before doing so, we ought to introduce the Scott topology.

\subsubsection{The Scott topology}

If $P$ is a dcpo, then $O \subseteq P$ is \emph{open} in the \emph{Scott topology} \cite{abramsky1994domain,scott1970outline}
 if it is (i) \emph{upper closed} (if $x \in O$ and $y \in P$ fulfills $x \preceq y$, then $y \in O$) and (ii) \emph{inaccessible by directed suprema} (if $A\subseteq P$ is directed and $\sqcup A \in O$, then $A \cap O \neq \emptyset$). Given the dcpo $P$, we denote by $\sigma(P)$ its Scott topology. Crucially, 
the Scott topology characterizes the order structure in $P$:
\begin{equation}
\label{charac order by topo}
    x \preceq y \iff  x \in O \text{ implies } y \in O \text{ for all } O \in \sigma(P)
\end{equation}
\cite[Proposition 2.3.2]{abramsky1994domain}. In fact, a countable topological basis of $\sigma(P)$, $(O_n)_{n \geq 0}$, coincides with the set of features on which we based our intuitive picture in Section \ref{intuition compu}.
Moreover, the Scott topology fulfills the $T_0$ separation axiom \cite{kelley2017general}: 
If $x,y \in P$ and $x \neq y$, then, by antisymmetry, either $\neg(x\preceq y)$ or $\neg(y \preceq x)$ holds. Thus, by \eqref{charac order by topo}, there is some $O \in \sigma(P)$ such that $x \in O$ and $y \not \in O$ (or vice versa).

As we just pointed out, we can think of the Scott topology as a family of features that distinguish the points in some set $P$. Importantly, by definition, any feature possessed by an element $x \in P$ is gathered (at some point) by a computational processes approaching $x$. Furthermore, by the $T_0$ property, these features are enough to distinguish among the elements of $P$. This property is key for computability in an uncountable set $P$. More specifically, if the Scott topology is second countable, then the elements in $P$ can be distinguished through a countable set of features.

Now that we have introduced the Scott topology, we can define the key property that distinguishes the uniform from the non-uniform approach, namely, essential information.

\subsubsection{Essential information}

To equate the computability of an element to that of a subset of $\mathbb{N}$ (i.e. to obtain a \emph{uniform} computability notion), we introduce the way-below relation. If $x,y \in P$, we say $x$ is \emph{way-below} $y$ and denote it by $x \ll y$ if, for each directed  $A\subseteq P$ such that $y \preceq \sqcup A$, there is an $a \in A$ for which $x \preceq a$ holds \cite{abramsky1994domain,scott1972continuous}. (The set of elements way-below $x \in P$ is represented by
$\twoheaddownarrow x$. Moreover, $\twoheaduparrow x$ denotes the set of elements \emph{way-above} $x$, i.e. the set of elements $y \in P$ such that $x \ll y$.)
We interpret an element way-below another as containing \emph{essential information} about the latter. We do so since, if a computational process approaches the latter, it must gather the information in the former. (As an example, if $A_1 \subseteq \twoheaddownarrow x$ is a directed set fulfilling $\sqcup A_1 =x$ and $A_2 \subseteq P$ with $\sqcup A_2=x$ is directed, then there is, for each $a \in A_1$, some $b \in A_2$ such that $a \preceq b$.
We think of $A_1$ as a canonical computation of $x$. (We will clarify what we mean in Proposition \ref{compu sets}.) 

Now that we have exposed what essential information is, we can define computability in the uniform approach.

\subsubsection{Uniform computation}

In this section, we introduce the uniform approach to computation in domain theory. To achieve this, in analogy to effective weak bases, we define \emph{effective bases}.
A subset $B \subseteq P$ is a \emph{basis} if, for any $x \in P$, there exists a directed set $B_x \subseteq \twoheaddownarrow x \cap B$ such that $\sqcup B_x = x$ \cite{abramsky1994domain}. The connection between bases and the Scott topology is stronger than that of weak bases. In fact, if $B$ is a basis, then $(\twoheaduparrow b)_{b \in B}$ constitutes a topological basis for the Scott topology $\sigma(P)$ \cite{abramsky1994domain}.
We say a dcpo is \emph{continuous} or a \emph{domain} if bases exist. As in the case of weak bases, we are interested in the cases where countable bases exist, or  \emph{$\omega$-continuous} dcpos. (Here, as in the non-uniform case, we can introduce computability via
Turing machines, as we will see later on.) An example of an $\omega$-continuous dcpo is the Cantor domain $\Sigma^\infty$, where $\Sigma^*$ is a (countable) basis.

As in the case of weak bases, we call a basis $B \subseteq P$ \emph{effective} if there is a finite map $\alpha$ enumerating $B$, $B=(b_n)_{n\geq0}$, a \emph{bottom} element $\perp \in B$ (i.e. an element $\perp \in B$ fulfilling $\perp \preceq x$ for all $x \in P$), and
\begin{equation*}
\big\{\langle n,m \rangle \big| b_n \ll b_m \big\}
\end{equation*}
is recursively enumerable \cite{edalat1999domain,stoltenberg2008computability}. (We assume $b_0 = \perp$ in the following, which we can do w.l.o.g..) Note that $\Sigma^*$ is an effective basis for $\Sigma^\infty$.

Now that we have defined the basic ingredients, we can define computability in the uniform scenario. If $P$ is a dcpo with an effective basis $B=(b_n)_{n\geq0}$, we say $x \in P$ is \emph{computable} provided 
\begin{equation*}
\big\{n \in \mathbb{N} \big|b_n \ll x \big\}
\end{equation*}
is recursively enumerable \cite{edalat1999domain}. More specifically, we may say that $x$ is $B$-computable to specify the basis. (We deal with the dependence of element computability on the different components in this order-theoretic model in the following section.) 


The uniform approach also allows us to introduce computable functions. Before doing so, we need some order-theoretic/topological notions. The first one is a direct extension of the notion of (real-valued) monotones from Section \ref{opt preorders}. In particular, a map $f:P \rightarrow Q$ between dcpos $P,Q$ is
a \emph{monotone} if $x \preceq_P y$ implies $f(x) \preceq_Q f(y)$. Furthermore, $f$ is
\emph{continuous} if it is a monotone and, for any directed set $A \subseteq P$, we have $f(\sqcup A) = \sqcup f(A)$ \cite{abramsky1994domain,scott1970outline}.
Note this definition coincides with the usual topological one when we take Scott topology on $P$ \cite{kelley2017general}.
We can now introduce computable functions.

We say a function $f:P \to Q$ between two dcpos $P$ and $Q$ with respective effective bases $B=(b_n)_{n\geq0}$ and $B'=(b'_n)_{n\geq0}$ is \emph{computable} if it is continuous and the set
 \begin{equation*}
\big\{\langle n,m \rangle \big| b'_n \ll f(b_m) \big\}
 \end{equation*}
is recursively enumerable \cite{edalat1999domain}. As in the case of elements, we may specify the effective bases we consider and we may call $f:P \to Q$ \emph{$(B,B')$-computable} instead. (We deal with the dependence of function computability on the different components in this order-theoretic model in the following section.) 



From the exposition in this section, it is clear that there are several differences between the non-uniform approach
and the one in domain theory. The most prominent contrast being that we did not even attempt to define computable functions in the former. We devote the following section to the discussion of the dissimilarities between these two approaches, which are detailed in \cite{hack2022computation} more thoroughly.

\subsection{The difference between the uniform and non-uniform approaches}
\label{sub: difference comput}

The basic differences between the non-uniform approach that we introduced in Section \ref{order in compu} and the uniform one in Section \ref{uniform compu} are the following:

\begin{enumerate}[label=(\Roman*)]
\item The inclusion of $\ll$ and, hence, the substitution of weak bases by bases.
\item The definition of computable elements. The uniform approach imposes the recursive enumerability of $\alpha^{-1}(B \cap \twoheaddownarrow x)$, which is stronger than the natural extension of the non-uniform framework.
\item The definition of computable functions, which is absent in the non-uniform approach.
\item The inclusion of a bottom element $\perp$ in the uniform framework.
\item The definition of effectivity, which again is stronger in the uniform approach than the natural extension of the non-uniform approach.
\end{enumerate}

In the reminder of this section, we discuss the influence of the differences (I)-(V) in terms of four different key computational angles: computable elements, computable functions, model dependence of computability and complexity.

\subsubsection{Computable elements}

Regarding computable elements, there are two basic differences between the approaches we have considered: (a) the scope (i.e. the set of order structures where we get a computability notion), and (b) the definition of computable elements (instead of element computability being equal to the existence of \emph{some} computational path approaching it, it is tied to a specific path).
We discuss these differences in this section.

Concerning the scope of these approaches, it should be noted that the non-uniform computability covers a strictly smaller set of partial orders, as we state in the following proposition.


\begin{prop}
\label{weak basis compu and no basis compu}
If $P$ is a dcpo with an effective basis, then it has an effective weak basis. However, there are dcpos with effective weak bases that have no basis.
\end{prop}

Note that Proposition \ref{weak basis compu and no basis compu} clarifies why Definition \ref{def:eff weak basis} cannot be simplified. In particular, although it should be further clarified, the existence of effective bases such that $\{\langle n,m \rangle| b_n \preceq b_m\}$ is not recursively enumerable seems likely. If that were the case, requiring a stronger form of Definition \ref{def:eff weak basis} would result in some order structures with uniform computability and without non-uniform computability.

Regarding the definition of computable elements, we note in Proposition \ref{compu sets} that the difference is only apparent. That is, in the uniform approach, the existence of the specific computational path to some element used in its definition
is equivalent to that of an arbitrary computational path leading to it. We state this in the following proposition.


\begin{prop}
\label{compu sets}
If $P$ is a dcpo with an effective basis $B=(b_n)_{n\geq0}$ and $x \in P$, then the following are equivalent:
\begin{enumerate}[label=(\roman*)]
\item There exists a directed set $B_x \subseteq B$ such that $\sqcup B_x=x$ and $\{n\in\mathbb{N}|b_n \in B_x\}$ is recursively enumerable.
\item $x$ is $B$-computable.
\end{enumerate}
\end{prop}

In the following section, we address the second key difference between the approaches, namely, the definition of computable functions.

\subsubsection{Computable functions}

Regarding computable functions, the difference between both approaches lied in whether we defined them or not. The reason for this distinction comes from the intuition we aim to capture bu such a definition. Essentially, a computable function is meant to map
computable inputs to computable outputs \cite{ko2012complexity,braverman2005complexity}. Indeed, this holds for the definition that we gave in the uniform case, as we state in the following proposition (which is part of \cite{edalat1999domain} and we proved more directly in \cite{hack2022computation}).

\begin{prop}[\cite{edalat1999domain}]
\label{why << needed}
Let $P$ and $Q$ be dcpos
with effective bases $B$ and $B'$, respectively. If $f:P \to Q$ is a $(B,B')$-computable function and $x \in P$ is computable, then $f(x)$ is computable.
\end{prop}

The intuition behind the definition of a computable function $f:P \to Q$ is the concatenation of finite instructions: We require an effective procedure related to $f$ such that, when combined with an effective procedure for $x \in P$, we obtain a procedure for $f(x)$. This is precisely how computable functions are defined in the uniform approach.
However, this does not seem to be the possible, in general, in the non-uniform approach. In particular, the potential disconnection between the paths leading to $x$ in $P$ and those leading to $f(x)$ in $Q$ may prevent us from constructing such a path. (The details can be found in \cite{hack2022computation}.)

\subsubsection{Model dependence of computability}

The dependence of the computability properties on the \emph{free parameters} of the model (i.e. the finite map and the basis/weak basis) also seems to change form the uniform to the non-uniform approach. Regarding the first free parameter, and assuming the Church-Turing thesis, computability in the non-uniform framework is independent of the chosen finite map, as the following proposition states.

\begin{prop}
\label{indep finite map}
If $P$ is a dcpo with an effective weak basis, then the set of computable elements is independent of the chosen finite map.
\end{prop}

The situation concerning weak bases seems to be different. In particular, as in the case of computable functions, the potential disconnection between two weak bases of the same dcpo seems to point out towards the dependence of computability on the chosen weak basis. We informally argue this is the case in \cite{hack2022computation}.

In contrast with its more general counterpart, computability in the uniform approach also seems to be independent of the chosen basis. We state this more specifically in the following proposition.

\begin{prop}
\label{model indep bases}
If $P$ and $Q$ are dcpos with effective bases, then the following statements hold:
\begin{enumerate}[label=(\roman*)]
\item The set of computable elements in $P$ is independent of the chosen finite map. Moreover, the set of computable functions $f:P \to Q$ is independent of the chosen finite maps.
\item If $B=(b_n)_{n\geq0}$ and $B'=(b'_n)_{n\geq0}$ are effective bases for $P$ such that $\{\langle n,m\rangle| b'_n \ll b_m\} $ is recursively enumerable, then $x \in P$ is $B'$-computable provided it is $B$-computable.
\item Let $B=(b_n)_{n\geq0}$ and $B'=(b'_n)_{n\geq0}$ be effective bases for $P$ and both $C=(c_n)$ and $C'=(c'_n)$ be effective bases for $Q$. If both $\{ \langle n,m \rangle | b_n \ll b'_m\}$ and $\{ \langle n,m \rangle | c'_n \ll c_m\}$ are recursively enumerable, then any $(B,C)$-computable function $f:P \to Q$ is $(B',C')$-computable.
\end{enumerate}
\end{prop}

Note that the extra hypotheses concerning recursive enumerability in Proposition \ref{model indep bases} $(ii)$ and $(iii)$ are likely to be true. This is the case since they involve the finiteness of some order relation connecting elements of two sets where (a) the elements in each set can be enumerated by finite means and (b) the order relation is finite when we restrict it to pairs of elements within either of the two sets. Hence, computability seems to be independent of the chosen basis in the uniform approach. 

\subsubsection{Complexity}

In general, the fundamental difference between computability and complexity is the existence of some precision measure through which we can distinguish how many resources we need to spend in order to obtain approximations to some precision. In this section, we start considering complexity and the existence of appropriate precision measure for the uniform approach and conclude by addressing the non-uniform case. 

For the uniform approach, we have borrowed in \cite{hack2022computation} the maps \emph{inducing the Scott topology everywhere} \cite{martin2000foundation,martin2008technique,waszkiewicz2001distance,waszkiewicz2003quantitative} as precision measures. More specifically, taking $[0,\infty)^{op}$ the dcpo composed by the non-negative real numbers equipped with the \emph{reversed order} $\leq_{op}$ (i.e. $x \leq_{op} y$ if and only if $x \geq y$ for all $x,y \in [0,\infty)$), a monotone map $\mu:P \rightarrow [0,\infty)^{op}$ is said to \emph{induce the Scott topology everywhere in $P$} if, for all $O \in \sigma(P)$ and  $x \in O \subseteq P$, there exists some $\varepsilon >0$ such that $x \in \mu_{\varepsilon}(x) \subseteq O$, where
\begin{equation*}
    \mu_\varepsilon(x) \coloneqq \{ y \in P| y \preceq x \text{ and } \mu(y)< \varepsilon\}.
\end{equation*}
Therefore, we consider $\mu$ to be a measure of information or precision. Moreover, maps inducing the Scott topology everywhere always exist in the uniform approach (see \cite[Proposition 10]{hack2022computation}). This led us to the following definitions of complexity \cite{hack2022computation}, which seems to be absent in the literature on domain theory.

\begin{defi}[Element complexity]
\label{def:element complex}
Take $P$ a dcpo with an effective basis $(b_n)_{n\geq0}$ and $\mu$ a continuous map inducing the Scott topology everywhere in $P$. We say the time (space) complexity of an element $x \in P$ is bounded by $t:\mathbb{N} \rightarrow \mathbb{N}$ if there exists a computable function $\phi:\mathbb{N} \to \mathbb{N}$ which computes $\phi(n)$ in $t(n)$ steps (using $t(n)$ cells)  for all $n \in \mathbb{N}$
and we have, for all $n \geq 0$,
\begin{equation*}
    \mu(b_{\phi(n)})-\mu(x) < 2^{-n} \text{ and } \sqcup (b_{\phi(n)})_{n\geq 0}=x.
\end{equation*}
\end{defi}

\begin{defi}[Function complexity]
\label{def:function complex}
Take $f:D \rightarrow D'$ a computable function between dcpos $D,D'$ with effectively given bases $(b_n)_{n\geq0}$ and $(b'_n)_{n\geq0}$, respectively, and $\mu$ a map inducing the Scott topology everywhere in $D'$. We say the time (space) complexity of $f$ is bounded by $t:\mathbb{N} \rightarrow \mathbb{N}$ if there exists a computable function $\phi:\mathbb{N} \to \mathbb{N}$ 
which computes $\phi(n)$ in $t(n)$ steps (using $t(n)$ cells) for all $n \in \mathbb{N}$
and we have, for all $n,m \geq 0$,
\begin{equation*}
    \mu(b'_{\phi(n)})-\mu(f(b_{\pi_1(n)})) < 2^{-\pi_2(n)} \text{ and } \sqcup (b'_{\phi(n)})_{n \in \{\langle m,p \rangle| p \in \mathbb{N}\}} = f(b_m).
\end{equation*}
\end{defi}

\begin{defi}[Polynomial-time computable elements and functions]
A computable element $x \in P$ or a computable function $f:D \rightarrow D'$  is said to be polynomial-time computable if its time complexity is bounded by a polynomial $p:\mathbb{N} \rightarrow \mathbb{N}$.
\end{defi}

Despite the fact we have introduced a satisfactory notion of complexity for the uniform approach, its existence does not directly extend to the non-uniform framework. In fact, the precision measures (i.e. maps inducing the Scott topology everywhere) on which we based the previous definitions do not even exist in some non-uniform instances, as the following proposition states.

\begin{prop}
\label{prop no weak complex}
There are dcpos where effective weak bases exist and continuous maps that induce the Scott topology everywhere do not.
\end{prop}

(A graphical representation of the dcpo that supports the last proposition can be found in \cite[Figure 3]{hack2022computation}.) As a result of the last proposition, there does not seem to be, in general, a proper precision measure in the non-uniform case. Hence, the uniform approach seems to be superior also in this regard.

In the following section, we summarize the reasons behind the dissimilarities between both the uniform and non-uniform approaches.
Given their computational importance, this summary provides specific support for several assumptions in the stronger framework.

\subsection{Justification of the stronger framework}

The inclusion of the differences we listed in Section \ref{sub: difference comput} in the uniform approach can be justified by the following reasons: 

\begin{enumerate}[label=(\Roman*)]
\item $\ll$ is essential in order for $(a)$ a proper definition of
computable functions to exist,
$(b)$ the independence of computability regarding the choices in our model, and $(c)$ the existence of precision measures (i.e. complexity theory).
\item The uniform and non-uniform definition of computable elements are equivalent in the uniform approach. 
\item A proper definition of computable functions does not seem to exist in the more general framework.
\item The bottom element $\perp$ is included to allow certain functions that are based on comparisons to be computable. In particular, when the comparisons are not fulfilled, we have a default correct element we can useful output. (This is also why $\perp=b_0$ is fixed.) As we have detailed in \cite{hack2022computation}, the inclusion of $\perp$ is key to show the equivalence between the  definitions of computable elements in the uniform picture, to define computable functions and to diminish the dependence of computability on the chosen parameters. 
\item The strong version of effectivity in the uniform approach allows the definitions of computable element to coincide. (This would fail if we naturally extend the \emph{non-uniform} definition.)
\end{enumerate}

In this section, we have discussed the order-theoretic details that are essential to translate computation from Turing machines to uncountable spaces via order structures. In the following section, we focus specifically on the countability restriction we impose on these structures (via either bases or weak bases) for them to be connected to Turing machines. In particular, we relate these restrictions with the ones that are usually deployed in the study of real-valued representations of partial orders, namely, order density and multi-utilities (see Sections \ref{sec:mu charact} and \ref{sec: order props}).

\section{Denumerability constraints:\\ From Turing machines to order structures} 
\label{sec: count const turing}

In the previous section, we were concerned with the sort of order structure that is required in order to introduce computability in uncountable spaces. Throughout that discussion, we proposed a modification of the countability requirements that allows us to derive the notion of computability on uncountable spaces from Turing machines. In particular, we considered the substitution of countable \emph{bases} by their \emph{weak} counterpart. These sort of restrictions have only been considered, to our best knowledge, in the realm of computability that belongs to domain theory. Hence, it is the purpose of this section to bring them closer to the more usual countability restrictions that are considered in order theory, like order denseness properties (Section \ref{sec: order props}) and multi-utilities (Section \ref{sec:mu charact}). This improves, thus, on the relation between the order-theoretic approaches to computability and the study of learning system. The original connection we presented between them was in terms of uncertainty, which is the fundamental notion that underlies both areas. Here, furthermore, we show connections between the fundamental properties that allow the introduction of optimization principles in the former and the introduction of a notion of uncountable computability derived from Turing machines in the latter.

\subsection{Order density and weak bases}

As a first connection, we consider the relation between order density and weak bases. In this regard, in the following proposition, the first thing we notice is that, although the existence of a countable weak basis does imply that of a countable set of elements that somewhat separates a dcpo $P$, it does not necessarily imply that $P$ is Debreu separable.

\begin{prop}
\label{weak basis implications}
The following statements hold:
\begin{enumerate}[label=(\roman*)]
    \item If $P$ is a dcpo with a countable weak basis $B \subseteq P$, then there exists a countable Debreu upper dense subset $D\subseteq P$. Furthermore, if $x \prec y$, then there exists some $b \in B$ such that $b \preceq y$ and either $x \prec b$ or $x \bowtie b$ holds.
    \item There exist dcpos with countable weak bases and no countable Debreu dense subset.
\end{enumerate}
\end{prop}

As we have just stated, the existence of countable weak bases does not imply that of countable Debreu dense subsets. However, we have not addressed the converse yet. As a matter of fact, provided an arbitrary (not necessarily countable) weak basis exists and $P$ is a Debreu separable dcpo, we can construct a dcpo $Q$ that contains several elements of $P$ and has a countable weak bases. Nonetheless, even strengthening the order density restriction to Debreu upper separability is insufficient for a countable weak basis to exist. We state both these results precisely in the following theorem. In order to do so, however, we require a definition:  We say $x \in P$ has a \emph{non-trivial} directed set if there exists a directed set $A \subseteq P$ such that $\sqcup A=x$ and $x \not \in A$ \cite{hack2022relation}. Accordingly, given a weak basis $B \subseteq P$, we call the set of elements which have non-trivial directed sets $A \subseteq B$ the \emph{non-trivial} elements of $B$ and denote it by $\mathcal{N}_B$. In the same vein, we may call elements \emph{trivial}.

\begin{theo}
\label{thm 1}
The following statements hold:
\begin{enumerate}[label=(\roman*)]
    \item If $P$ is a dcpo with a countable Debreu dense set $D\subseteq P$ and a weak basis $B\subseteq P$, then there exists a dcpo $Q$ with a countable weak basis such that $D \cup \mathcal{N}_B \subseteq Q$.
    \item There exist Debreu upper separable dcpos without countable weak bases.
    \end{enumerate}
\end{theo}

Given that, as last theorem states, Debreu upper separability is insufficient for the existence of countable weak bases, we consider stronger order density properties. We obtain, in particular, two situations where this is the case, which we state in the following proposition. Before stating it, however, we need a definition. If $x,y \in P$, then we say $y$ is an \emph{immediate successor} of $x$ if both $x \prec y$ and $(x, y) \coloneqq \{z \in P|x \prec z \prec y\} = \emptyset$ hold \cite{bridges2013representations}.

\begin{prop}
The following statements hold:
\begin{enumerate}[label=(\roman*)]
\item If $P$ is an order separable dcpo, then $P \setminus min(P)$ is a dcpo with a countable weak basis. Furthermore, if $P$ also has a countable Debreu upper dense subset, then $P$ has a countable weak basis.
\item If $P$ is a dcpo with a countable Debreu dense subset whose elements have a finite number of immediate successors, then $P \setminus min(P)$ is a dcpo with a countable weak basis. Furthermore, if $P$ also has a countable Debreu upper dense subset, then $P$ has a countable weak basis.
\end{enumerate}
\end{prop}

As exposed in this section, the existence of countable weak bases has certain implications regarding order density and, moreover, if the order density properties are strong enough, then countable weak bases exist. However, these properties are sometimes equivalent. This is the case, for instance, for the Cantor domain $\Sigma^\infty$. (In particular, the countable set $\Sigma^* \subseteq \Sigma^\infty$ is Debreu dense, Debreu upper dense and a weak basis.) In the following section, we consider a class of dcpos which includes $\Sigma^\infty$ and achieves the equivalence. 

\subsection{Order density and bases}

The property we exploit in this section and that is shared by the Cantor domain $\Sigma^\infty$ is \emph{conditional connectedness}. We introduced this property in Section \ref{other repres}, when considering the real-valued characterizations that are usually taken into account in the axiomatic approaches to thermodynamics. Under this condition, Debreu upper separability coincides with the existence of countable weak bases and, moreover, also with the existence of countable bases.

\begin{theo}
\label{thm 2}
If $P$ is a conditionally connected dcpo, then the following statements are equivalent:
\begin{enumerate}[label=(\roman*)]
\item $P$ is Debreu upper separable.
\item $P$ is $\omega$-continuous.
\item $P$ has a countable weak basis.
\item $P$ is Debreu separable and $K(P)$ is countable.
\end{enumerate}
\end{theo}

\begin{proof}
Note that we showed the equivalence between $(i)$, $(ii)$ and $(iv)$ in \cite[Theorem 2]{hack2022relation}. To conclude the improvement on this result that we present here, it suffices to show that $(ii)$ and $(iii)$ are equivalent. In particular, given that the converse is straightforward, we ought to show that $(iii)$ implies $(ii)$, as we do in the following paragraphs.

Take some $x \in P$. Since there exists a countable weak basis $B \subseteq P$, there is some directed set $B_x \subseteq B$ such that $\sqcup B_x =x$.  If $x \notin B_x$, then we have that $b \prec x$ for all $b \in B_x$. Hence, by \cite[Proposition 11 ]{hack2022relation}, $b \ll x$ for all $b \in B_x$ and we have finished. To conclude, assume that, for all directed sets $B_x \subseteq B$ such that $\sqcup B_x=x$, we have that $x \in B_x$, that is, that $x$ is trivial in $B$. As we show in the following lemma, this implies that, for all directed sets $A_x \subseteq P$ such that $\sqcup A_x=x$, we have that $x \in A_x$, that is, that $x$ has no non-trivial directed sets in $P$.

\begin{lemma}
Let $B$ be a weak basis of a conditionally connected dcpo $P$ and $x \in P$. If $x$ is trivial in $B$, then $x$ is trivial in $P$.
\end{lemma}

\begin{proof}
Instead of directly showing the statement, we prove its contrapositive. More specifically, we prove that, if there exists some directed set $A_x \subseteq P\setminus\{x\}$ such that $\sqcup A_x=x$, then there exists some directed set $B_x \subseteq B\setminus\{x\}$ such that $\sqcup B_x=x$. Take, in particular,
\begin{equation*}
    B_x \coloneqq \bigcup_{a \in A_x} B_a,
\end{equation*}
where, for all $a \in A_x$, $B_a \subseteq B$ is a directed set such that $\sqcup B_a =a$. Note that $B_x$ is directed by conditional connectedness of $P$, given that, if $b \in B_a$ and $b' \in B_{a'}$ for a pair $a,a' \in A_x$, then we have that $b,b' \preceq x$ and, hence, $\neg(b \bowtie b')$ holds. Lastly, note that $\sqcup B_x = x$. By transitivity, it is clear that $b \preceq x$ for all $b \in B_x$. Assume, then, that there exists some $z \in P$ such that $b \preceq z$ for all $b \in B_x$. In particular, by definition of supremum of $B_a$, we have that $a \preceq z$ for all $a \in A_x$. Hence, by definition of supremum of $A_x$, we have that $x \preceq z$. As a result, $\sqcup B_x = x$. This concludes the proof.
\end{proof}

Thus, if $x \in B_x$ for all directed sets $B_x \subseteq B$ such that $\sqcup B_x=x$, we can apply again \cite[Proposition 11]{hack2022relation} and conclude that $x \ll x$. Hence, if $B \subseteq P$ is a weak basis, then it is also a basis and we have finished.
\end{proof}

The last theorem can be strengthened if we consider a stronger order denseness property, namely upper separability. We state the stronger result in Corollary \ref{upper separable}.  Notice, as in Theorem \ref{thm 2}, Corollary \ref{upper separable} improves upon \cite[Corollary 2]{hack2022relation} by adding a clause (which we showed in the last theorem) regarding countable weak bases.

\begin{coro}
\label{upper separable}
If $P$ is a conditionally connected dcpo, then the following statements are equivalent:
\begin{enumerate}[label=(\roman*)]
    \item $P$ is upper separable
    \item $P$ is $\omega$-continuous and either $K(P)=\{\perp\}$ or $K(P)=\emptyset$ holds.
    \item $P$ has a countable weak basis and either $K(P)=\{\perp\}$ or $K(P)=\emptyset$ holds.
\end{enumerate}
\end{coro}

In the following section, we extend the results in this one by considering the relation between countability restrictions for bases and multi-utilities. 

\subsection{Multi-utilities and bases}

As a first result connecting countability for multi-utilities and both bases and weak bases, we have the following proposition.

\begin{prop}
\label{LSC m-u existence}
If $P$ is a dcpo, then the following statements hold:
\begin{enumerate}[label=(\roman*)]
\item $P$ has a lower semicontinuous multi-utility.
\item If $P$ has a (basis) weak basis $B \subseteq P$, then there it has a (lower semicontinuous) multi-utility $(u_b)_{b\in B}$.
\item There exist dcpos with a lower semicontinuous multi-utility $(u_i)_{i \in I}$ such that any weak basis $B$ has larger cardinality than $I$. This holds, in particular, for a countable $I$.
\end{enumerate}
\end{prop}

Note that Proposition \ref{LSC m-u existence} states in a more precise way \cite[Proposition 18]{hack2022relation}.

Although the existence of (lower semicontinuous) countable multi-utilities does not assure the presence of any of the countability restrictions that we considered in the previous section (as Proposition \ref{LSC m-u existence} and \cite[Proposition 9]{hack2022classification} show), we can strengthen the restrictions on the multi-utility to partially recover the equivalences in Theorem \ref{thm 2}. We state these equivalences in the following theorem, which is a slight improvement on our result in \cite[Theorem 4]{hack2022relation}.

\begin{theo}
\label{thm 4}
If $P$ is a dcpo with a finite lower semicontinuous strict monotone multi-utility, then the following hold:
\begin{enumerate}[label=(\roman*)]
    \item $P$ is $\omega$-continuous.
    \item $P$ has a countable weak basis.
    \item $P$ is Debreu separable and $K(P)$ is countable.
    \item $K(P)$ is countable.
\end{enumerate}
\end{theo}

\begin{proof}
Since we already proved the equivalence between $(i)$, $(iii)$ and $(iv)$ in \cite[Theorem 4]{hack2022relation}, we only ought to show that $(i)$ and $(ii)$ are equivalent.

Since $(i)$ implies $(ii)$ by definition, we only ought to see the converse is true. In order to do so, given the equivalence between $(i)$ and $(iv)$, it is sufficient to note that, if there exists a countable weak basis $B \subseteq P$, then $K(P)$ is also countable. This follows from the fact $K(P) \subseteq B$. To see this, take some $x \in K(P)$. By definition, there exists some directed set $D_x \subseteq B$ such that $\sqcup D_x=x$. However, since $x \in K(P)$, there exists some $d \in D_x$ such that $x \preceq d$. Given that $d \preceq x$ holds by definition of supremum and that $\preceq$ is antisymmetric, we have $x=d \in D_x \subseteq B$. Hence, $K(P) \subseteq B$ and, in particular, it is countable.
\end{proof}

While Theorem \ref{thm 4} shows that most of the equivalences in Theorem \ref{thm 2} remain true if we require the existence of a \emph{lower semicontinuous finite strict monotone multi-utility} instead of \emph{conditional connectedness}, this is not the case for some of them. In particular, in Theorem \ref{thm 4}, we eliminated the first clause in Theorem \ref{thm 2}. We state the relation of that clause with the others in the following theorem.

\begin{theo}
\label{new thm compu}
Let $P$ be a dcpo with a finite lower semicontinuous strict monotone multi-utility. If $P$ is $\omega$-continuous, then it is Debreu upper separable. However, there exist Debreu upper separable dcpos that have a finite lower semicontinuous strict monotone multi-utility and are not $\omega$-continuous.
\end{theo}

As the final remark of this section, let us note that, in general, the existence of a finite strict monotone multi-utility that is lower semicontinuous in $\sigma(P)$ does not imply neither Debreu separability nor the existence of a countable Debreu upper dense subset. We include this counterexample as a complement of Theorem \ref{thm 4}.

\begin{prop}
\label{finite RP no deb sep}
There exist partial orders with finite strict monotone multi-utilities that are lower semicontinuous in the Scott topology such that every subset that is either Debreu dense or Debreu upper dense is uncountable.
\end{prop}

\begin{proof}
We can use the partial order we defined in \cite[Proposition 7]{hack2022geometrical}. That is, take $X \coloneqq \mathbb R \setminus \{0\}$ equipped with $\preceq$, where 
\begin{equation*}
    x \preceq y \iff u_1(x) \leq u_1(y) \text{ and } u_2(x) \leq u_2(y),
\end{equation*}
with $u_1(x) \coloneqq x$ for all $x \in X$, and $u_2(x) \coloneqq 1/x$ if $x >0$ and $u_2(x) \coloneqq - 1/|x|$ if $x <0$. The statement follows from the proof of \cite[Proposition 7]{hack2022geometrical}. We simply ought to note that $\{u_1,u_2\}$ is a finite strict monotone multi-utility composed of functions that are lower semicontinuous in $\sigma(P)$, given that they are monotone and, for every directed set $D \subseteq X$, we have that $\sqcup D \in D$.
\end{proof}

It should be noted that the last proposition shows precisely why the technique used in the proof of Theorem \ref{debreu thm} (see \cite[Theorem 1.4.8]{bridges2013representations}) to show that preorders with a utility function are Debreu separable cannot be extended to non-total preorders with a finite number of strict monotones. Moreover, it clarifies why the equivalence between $(iii)$ and $(iv)$ in Theorem \ref{thm 4} is not trivial. In particular, it shows that, given a dcpo $P$, the countability of $K(P)$ is necessary in order for $P$ to be Debreu separable. Lastly, since $K((X,\preceq))=X$, Proposition \ref{finite RP no deb sep} improves on \cite[Proposition 20 $(i)$]{hack2022relation}.

In the following section, we conclude our study of the countability restrictions in the order-theoretic approaches to computation by relating order density with both continuity in the Scott topology and completeness.

\subsection{Order density and Scott-continuity}

As a final relation between the usual order denseness properties and the order-theoretic approach to computability, we include implications of order denseness for both completeness and continuity in the Scott topology. Before stating this precisely in the following theorem, we recall a few definitions from \cite{abramsky1994domain}. Take $P$ a dcpo. We call a sequence $(x_n)_{n \geq 0} \subseteq P$ \emph{increasing} if $x_n \preceq x_{n+1}$ for all $n \geq 0$. Moreover, if $Q$ is also a dcpo, we say a monotone $f:P \to Q$ \emph{preserves suprema of increasing sequences} if, given any increasing sequence $(x_n)_{n\geq0} \subseteq P$, we have $\sqcup (f(x_n))_{n\geq0} = f(\sqcup (x_n)_{n\geq0})$. Lastly, we use the term \emph{sequential continuity} (applied to the Scott topology $\sigma(P)$) in the usual topological sense \cite{kelley2017general}. 

\begin{theo}
\label{thm 3}
If $P$ is a Debreu separable dcpo, then the following statements hold:
\begin{enumerate}[label=(\roman*)]
\item The following statements are equivalent:
\begin{enumerate}[label=(\alph*)]
\item Every directed set has a supremum.
\item Every increasing sequence has a supremum.
\end{enumerate}
\item If $Q$ is a dcpo and $f:P \rightarrow Q$ is a map, then the following are equivalent:
\begin{enumerate}[label=(\alph*)]
\item $f$ is sequentially continuous.
\item $f$ is monotone and preserves suprema of increasing sequences.
\item $f$ is continuous.
\end{enumerate}
\end{enumerate}
\end{theo}

\section{Summary}

In this chapter, we have introduced an order-theoretic computability framework more general than the one in domain theory. Moreover, we have derived the fundamentals of our approach intuitively. As a consequence of gaining generality, however, we have encountered several undesired consequences, like the absence of a notion of both computable functions and complexity. Moreover, we have noticed an increase in the dependence of computability on the specific parameters we choose in our computational model. Furthermore, we have connected the usual countability restrictions in these frameworks (the ones that allow the notion of computability to be inherited from Turing machines) to the usual countability restrictions in order theory, which are in terms of order denseness and multi-utilities.

In the following chapter, we make some concluding remarks regarding our work here.

\newpage
\thispagestyle{empty}
\mbox{}
\newpage



\chapter{Conclusion}

In this final chapter, we conclude our work here with three closing sections. We (i) provide a brief answer to our research questions in Section \ref{answer RQ}, (ii) give a detailed account of our contributions in Section \ref{contributions} and (iii) connect our findings with the literature in Section \ref{discussion}.   

\section{Answer to our research questions}
\label{answer RQ}

In this section, we give a brief answer to the research questions we considered in the introduction. (For a more detailed account of what our contribution to the different topics was, one can read Section \ref{contributions}.) Our answers to these questions are the following:

\begin{enumerate}[label=\textbf{(Q\arabic*)}]
\item \textbf{Under what conditions do the fluctuation theorems hold for learning systems?}
\end{enumerate}

\textbf{Answer to (Q1):} In order to answer this question, we have derived the thermodynamic framework from the theory of Markov chains. As a result from this detachment concerning physics, we have reduced the conditions that were required for Jarzynski's equality to hold in the context of learning systems and, moreover, we have shown Crooks' fluctuation theorem requires an extra condition given that the definition of some of the involved quantities vary in our context with respect to physics. Hence, we have improved on the applicability of these results to learning systems and, in particular, on the previous approach to them from the learning systems community \cite{grau2018non}. 

\begin{enumerate}[label=\textbf{(Q\arabic*)}]\addtocounter{enumi}{1}
\item \textbf{Can these theorems be observed in the adaptation of biological learning systems? In particular, in human sensorimotor adaptation?}
\end{enumerate}

\textbf{Answer to (Q2):} We have provided the first piece of experimental evidence supporting the validity of the fluctuation theorems in the adaptation of biological learning systems. In particular, we have tested them in the context of a human sensorimotor task. It should be noted that, to our best knowledge, this constitutes the first piece of evidence supporting fluctuation theorem outside the realm of physics.

\begin{enumerate}[label=\textbf{(Q\arabic*)}]\addtocounter{enumi}{2}
\item \textbf{Given the set of possible transitions of a system, how well can we encode it via measurement devices (i.e. real-valued functions)?}
\end{enumerate}

\textbf{Answer to (Q3):} We have improved on the study of real-valued characterizations of preorders. More specifically, motivated by majorization and the maximum entropy principle, we have contributed to the study of complexity, optimization and their interplay in preordered spaces. In this regard, the main novelties in our work are (i) the introduction of new classes of real-valued functions of interest (injective monotones and injective monotone multi-utilities), (ii) finding new relations among the different sorts of real-valued characterizations, and (iii) the establishment of the (so far) most general classification of preorders in terms of real-valued functions. Moreover, we have related this with the study of the most suitable notions of dimension for order structures. In fact, we have contributed to it with the definition of a new notion of dimension (the Debreu dimension) and with the establishment of relations between the different definitions that have been considered. Aside from gaining insight into learning systems and the maximum entropy principle, we have also found some applications in physics to our general approach to \textbf{(Q3)}, in particular in both molecular diffusion and resource theory.

\begin{enumerate}[label=\textbf{(Q\arabic*)}]\addtocounter{enumi}{3}
\item \textbf{Can we clarify what order-theoretic structure is required in order to introduce computability on uncountable sets?}
\end{enumerate}

\textbf{Answer to (Q4):} We have introduced a more general framework that follows from an intuitive picture and allows us to define computable elements. However, we have found three criteria to argue in favor of the more specific approach which is followed in domain theory: computable functions, the dependence of computability on the choices one has in each framework, and complexity. More specifically, we support the domain-theoretic approach because (i) computable functions are defined , (ii) computability is less dependent on  the model's degrees of freedom, and (iii) complexity theory is defined.  

\begin{enumerate}[label=\textbf{(Q\arabic*)}]\addtocounter{enumi}{4}
\item \textbf{In particular, what sort of countability restrictions on these structures are required to inherit computation from Turing machines?}
\end{enumerate}

\textbf{Answer to (Q5):} We have shown that the countability restrictions that are used in order to derive computability from Turing machines in both our more general approach and in domain theory are closely related to the different countability restrictions that we have considered as part of our approach to \textbf{(Q3)}. In particular, we have established close relations between the countability constraints in these approaches of computability and both countable order density properties and multi-utilities. This includes their equivalence for several spaces of computational interest.

Given that we have roughly outlined the answers to our research questions in this section, we proceed to state the main contributions of this work in the following one.

\section{Contributions}
\label{contributions}

The main contributions of this work are the following:
\begin{enumerate}[label=(C\arabic*)]
    \item We derive Jarzynski's equality and Crooks' fluctuation theorem in the context of general Markov chains. We emphasize that there is no need for hypotheses regarding the reversed process or detailed balance to derive the former (which was assumed in previous approaches to it in the context of learning systems). Furthermore, we introduce a new notion of reversed process that is more natural in the context of Markov chains and, as a result, we point out certain differences in Crooks' fluctuation theorem. Hence, we clarify the applicability of these results in the context of learning systems.
    \item We provide experimental evidence supporting Jarzynski's equality and Crooks' fluctuation theorem in a human sensorimotor task. In fact, we perform the first test of these results in a non-physical context. Hence, we contribute to the illustration of the potential usefulness of thermodynamics in the study of learning systems.  
    \item Motivated by the study of optimization in learning system, we introduce a new sort of real-valued representation of preorders, injective monotones, that extends some of the appealing properties of Jaynes' maximum entropy principle, and connect it to previously known notions of complexity and optimization in preordered spaces. Moreover, we define a new notion of order density, Debreu upper separability, which proves to be relevant in both the study of complexity and optimization in preorders and the introduction of computability in uncountable sets via order-theoretic tools.
    \item We obtain the most complete classification of preordered spaces in terms of real-valued functions so far. In particular, we gain insight concerning both the internal relations and the interplay between complexity and optimization in preordered spaces.
    \item Following our interest in the complexity of order structures, we introduce a new notion of dimension for partial orders, the Debreu dimension, that follows from a modification of the classical notion by Dushnik and Miller and has the advantage of being better motivated from a geometrical viewpoint. Moreover, we study how the building blocks of the Debreu dimension, the Debreu linear extensions, can be constructed
    and relate our new notion of dimension to the previously known ones.
    \item As an application of our abstract approach to preorders and their representation by real-valued functions, we improve on the relation between molecular diffusion and entropy. In particular, we show that, whenever trying to reproduce the properties of the second law by following an order-theoretic approach to molecular diffusion, a (countably) infinite number of functions is needed. We contribute hence to the study of the complexity of thermodynamic proposals like that by Ruch and others (the so-called \emph{principle of increasing mixing character}).
    Moreover, we prove an analogous result in the context of systems coupled to a heat bath.
    \item We consider computation as a process of uncertainty reduction and introduce an order-theoretic framework that formalizes this and, moreover, allows the definition of computable elements in uncountable spaces. When comparing to previous approaches, the advantage of our proposal lies in its generality. However, by pointing out the limitations of our framework, we provide specific reasons supporting the stronger restrictions on the order structure that are usually required in the literature.
    \item As part of our attempt to understand the order structures that are actually needed in order to define computability, we bridge two areas that, to our best knowledge, have always been studied independently. More specifically, we refer, on the one hand, to the study of real-valued representations of preorders and, on the other hand, to the order-theoretic approaches to computability. In particular, we  link the usual countability restrictions in order theory (order density and multi-utilities) with those in the order-theoretic approaches to computability on uncountable sets. These restrictions are crucial since, in these approaches to computability, Turing machines are assumed to be a baseline for computability and, because of this, the order structures ought to be closely related to Turing machines by countability constraints. Hence, we gain insight into the sort of order structures that can be used in order for an uncountable space to inherit computability from Turing machines.
\end{enumerate}

In the following section, we connect our contributions to each other and relate them to the literature.

\section{Discussion}
\label{discussion}

In this section, we discuss several points connected to our work that we believe might be of interest to put it in context. Following the general structure, we divide the discussion into the three parts that we have considered along this work.

\subsection{Uncertainty and learning systems}

\subsubsection{Derivation and applicability of the fluctuation theorems}

We have shown both Jarzynski's equality and Crooks' fluctuation theorem in the framework of general Markov chains. In fact, we have adopted an unconventional approach where, instead of assuming the definition of the variables of interest in physics, we constructed them starting from a Markov chain, that is, we introduced them from well-defined notions for Markov chains. Hence, we have made these results available for the wide range of areas based on Markov chains,
like learning systems.

To our best knowledge, there is only one source that
used
these results in a non-physical context, namely \cite{grau2018non}. However, the assumptions there were too strong. In particular, it was assumed that the transitions matrices in Theorem \ref{Jarzynski's equality} ought to satisfy detailed balance with respect to their equilibrium distributions. As we showed in \cite{hack2022jarzyski}, this is not necessarily the case. Actually, the equality can be shown without incorporating any hypothesis related to the time reverse of the Markov chain in question. Nonetheless, the adoption of such premises in \cite{grau2018non} is understandable, given that it is under those stronger restrictions that the result is usually shown when adopting the Markov chain scenario in non-equilibrium thermodynamics \cite{crooks1998nonequilibrium,crooks1999excursions,crooks1999excursions,crooks1999entropy}. Notice, albeit following a different method, the same assumptions that we have established for the validity of Theorem \ref{Jarzynski's equality} can be found in \cite{jarzynski1997equilibrium}.

The situation regarding Crooks' fluctuation theorem is different. In this case, assumptions regarding time reversibility of the Markov chain in question are, actually, mandatory. The difference between our approach and the one in non-equilibrium thermodynamics lies in the definition of the involved quantities. More specifically, in our approach, the involved quantities are unique (up to a constant \cite{hack2022jarzyski}). On the contrary, in the usual non-equilibrium thermodynamics setup \cite{crooks1998nonequilibrium,crooks2000path,crooks1999excursions}, there is some degree of ambiguity regarding these quantities and the Markov chains. Hence, following our approach (where general Markov chains constitute the starting point), we incorporate changes in some of the definitions related to Crooks' theorem. As a result, some small differences arise in its statement. More information concerning this matter can be found in the extensive discussion we included in \cite[Section 5]{hack2022jarzyski}, where we expose the degree of arbitrariness of the definitions in several thermodynamic approaches to Crooks' theorem. 

As a final remark, it should be noted that the thermodynamics literature suffers from some ambiguity in its terminology. In particular, \emph{detailed balance} (which has a well-established mathematical definition \cite{levin2017markov}) is confused with \emph{microscopic reversibility} (see \cite{hack2022jarzyski} for its definition). Moreover, it is also mistaken for the assumption that some specific Boltzmann distributions \eqref{boltzmann} are stationary for the considered transition matrices. We have discussed this issue in \cite[Section 5]{hack2022jarzyski}. Hence, when consulting the literature, the reader should be careful in this regard.

\subsubsection{Importance of the fluctuation theorems}

Aside from the physical value of the fluctuation theorems as relations between equilibrium quantities and non-equilibrium realizations of physical processes, their importance stems from the fact they constitute testable predictions concerning whether a system follows $d$-majorization for some specific $d \in \mathbb P_\Omega$ or not (see Section \ref{sec: intro pre} for more details). The fact that these predictions ought to be probabilistic follows from the non-deterministic dynamics of the system. 
To illustrate this, it is useful to compare our situation with that of a
dynamical system,
which consists, roughly, of a set $X$ and a map that associates, to each pair consisting of an element $x \in X$ and some time $t \geq 0$, a unique point $y \in X$ where the system will be if it starts at $x$ and evolves freely until time $t$.
Predictions are simple in such a deterministic system since, given an initial state of the system, one can uniquely determine its state after certain amount of time from the theory and, hence, test it.
However, such an approach is not possible in $d$-majorization, since there is not a single but several possible evolution paths for the system.
Nonetheless, given the regularity exhibited by the possible dynamical processes in $d$-majorization, we are allowed to 
obtain another type of predictions, which is exemplified by both Crooks' fluctuation theorem and Jarzynski's equality.


\subsubsection{Fluctuation theorems in the study of learning systems}
We have contributed to the development of the parallelism between two, a priori, quite separated disciplines, namely thermodynamics and the study of learning systems, by (i) proving Jarzynski's and Crooks' fluctuation theorem for general Markov chains, and (ii) providing the first piece of experimental evidence supporting the validity of two fluctuation theorems for learning systems and, in particular, for human sensorimotor adaptation. Thus, we have further developed the \emph{bounded rationality} approach to learning systems that considers uncertainty or information processing as the key constraint that prevents adaptation to the environment (i.e. learning) from being optimal. In fact, our work belongs to a recent trend that is concerned with gaining
new insights into learning systems by using thermodynamic tools \cite{goldt2017stochastic,perunov2016statistical,england2015dissipative,still2012thermodynamics,ortega2013thermodynamics}.

\subsection{Real-valued representations of preorders}
\label{discu: real repre}

\subsubsection{Classification of preorders by real-valued monotones}

Several scientific disciplines rely on real-valued monotones in order to describe the
transitions
that their objects of study may experience. In fact, in many cases, these transitions are faithfully represented by an order structure, namely a preorder. The abandonment of the original order structure in favor of a set of real-valued functions is a common practice in some of them, given that we possess an intuitive comprehension of the real line and its powers. Hence, by studying the translation from order structures to functions in general, we hope to contribute to the development of several areas that make use of this framework and, moreover, to facilitate the translation of results between them. The most important contributions concerning these topics that we present here are the following:

\begin{enumerate}[label=(\roman*)]
\item We have studied what we call injective monotones, a special type of real-valued valued monotones on preordered spaces,
and have shown that preordered spaces that possess such a monotone constitute a distinct class. Thus, we have enriched the classification of preordered spaces via real-valued monotones. Moreover, this new class contributes to the study of optimization on general preordered spaces and, hence, to the establishment of principles similar to maximum entropy on general preorders. Importantly, these principles are considered for general preordered spaces, that is, we avoid the structural constraints required for convexity, which play a key role in the maximum entropy principle.
\item We have extended and  emphasized the usefulness of the characterization of the different classes of real-valued monotones via families of increasing sets.
\item We have profited from the characterization in (ii) and have introduced several preordered spaces that have allowed us to refine the classification of preordered spaces through real-valued monotones. In fact, we have achieved the most complete classification so far. Moreover, this advancement has helped us to improve our understanding regarding the interplay between the existence of optimization principles (i.e. single real-valued monotones) and the complexity of a preordered space (i.e. the sort and number of functions that conform a multi-utility).
\end{enumerate}

Several questions regarding the classification remain open. For example, the relation between finite strict monotone multi-utilities and finite injective monotone multi-utilities remains unclear.


\subsubsection{Majorization} 

In most of our work, majorization has been a key example through which we have motivated the objects we have studied. Along the way, we have established several results for majorization which, given its extensive use, are also worth mentioning:
\begin{enumerate}[label=(\roman*)]
\item We have clarified the relation between Shannon entropy (i.e. the maximum entropy principle) and majorization (see Section \ref{opt preorders} and \cite{hack2022representing}) in the context of machine learning and decision theory. We believe that majorization carries a fundamental notion of uncertainty, which we motivated as the preferences of a casino owner, and that it can be considered as an extension of the \emph{principle of insufficient reason} \cite{bernoulli1713ars}. However, this intuitive picture is often abandoned in favor of a much more specific one where Shannon entropy plays a leading role, namely (Jaynes') \emph{principle of maximum entropy} \cite{jaynes1957information}. This principle introduces an artificial distinction between equally unbiased distributions, that is, distributions that are deemed as equivalent according to the principle of insufficient reason. This results in a prominent role of the Boltzmann distribution \eqref{boltzmann}. Importantly, this does not only happen in cumbersome applications of the principle of maximum entropy, but in its usual applications \cite{jaynes1957information}, where the constraint is linear (cf. Figure \ref{fig:maxent}). In fact, under linear constraints, Shannon entropy can even favor the Boltzmann distribution over an uncountable set of equally unbiased distributions (cf. Proposition \ref{uncount max ent}).

\item We have shown the existence of upper semicontinuous injective monotones for $(\mathbb P_\Omega, \preceq_U)$ with its usual topology. Hence, we have established the existence of functions that retain the properties (up to order equivalence) of the maximum entropy principle and, moreover,  extend them to all compact sets.

\item  We have established the order density properties of majorization (see \cite[Lemma 5]{hack2022representing}).
\item We have improved on the characterization of majorization in terms of real-valued monotones. In particular, we have shown that, at least, $|\Omega|-1$ monotones are required to conform a multi-utility for majorization on $\mathbb P_\Omega$, that is, that its geometrical dimension is $|\Omega|-1$ \cite{hack2022geometrical}. Moreover, we have shown that, if $|\Omega| \geq 3$, then, at least, a countably infinite set of monotones are required to conform a strict monotone multi-utility for majorization on $\mathbb P_\Omega$ \cite{hack2022disorder}. 
\end{enumerate}

\subsubsection{The notion of order density for non-total preorders}

We have introduced a new notion of order density on non-total preorders, which we have called \emph{Debreu upper separability}. The relevance of this notion lies in its generality. In fact, it holds, for example, for any preorder on a countable set, unlike the definitions that were previously considered in the literature. More importantly, it plays a key role in: $(i)$ the classification of preorders by real-valued functions, $(ii)$ the study of order dimension, and $(iii)$ the use of order structures in computation. The most significant instances where this definition has proven to be useful are the following: 
\begin{enumerate}[label=(\roman*)]
\item In the study of order dimension and the classification of preordered spaces in terms of real-valued monotones. In particular, we have shown that (a) whenever this property holds, then countable multi-utilities exist and, consequently, injective optimization principles \cite{hack2022representing}, and (b) the basic components of the Debreu dimension, Debreu separable linear extensions, can be obtained as the limit of a sequential process of elimination of incomparable pairs that begins with the original partial order $(X,\preceq)$ \cite{hack2022geometrical}.
\item In the connection between the density requirements that are common in order theory \cite{bridges2013representations} and those that appear in domain theory \cite{abramsky1994domain}.
Although they exist independently in general, $\omega$-continuity (which allows computability to be derived from Turing machines in domain theory) is equivalent to Debreu upper separability for conditionally connected dcpos (cf. Theorem \ref{thm 2}). It should be noted that the \emph{conditionally connected} dcpos include one of the most prominent examples of domain theory, namely the Cantor domain.
Furthermore, we have shown that, if we consider the notion that was previously used instead of Debreu upper separability (upper separability), a similar equivalence holds for a much more restricted set of partial orders (cf. Corollary \ref{upper separable}). Lastly, we have shown that, if the complexity is small (more specifically, if finite lower semicontinuous strict monotone multi-utilities exist), then $\omega$-continuity implies Debreu upper separability (cf. Theorem \ref{new thm compu}).
\end{enumerate}

\subsubsection{The notion of dimension for partial orders}

Motivated by the replacement of the uncertainty preorder by real-valued functions (in particular, by the maximum entropy principle), we
have pondered on the question of what the most suitable notion of dimension for partial orders is. In an attempt to close the gap between the two definitions that can be found in the literature, namely what we have called the \emph{Dushnik-Miller} \cite{dushnik1941partially} and the \emph{geometrical} (which can somewhat be traced back to \cite{ok2002utility,evren2011multi}) dimension, we introduced a new notion, namely the \emph{Debreu} dimension. Following the intention of gaining insight into the notion of dimension itself in the context of partial orders, we have tried to relate the different definitions we have encountered. Our main contributions in this regard are the following:
\begin{enumerate}[label=(\roman*)]
\item We have provided mild hypotheses under which the building blocks of the Debreu dimension (i.e. Debreu separable linear extensions) exist. More specifically, we have shown they can be constructed as the limit of a sequence of partial orders where each extends the initial one (cf. Theorem \ref{count mu has D-ext}).
\item By extending the technique in (i), we have recovered, through a different method, the fact that having a finite geometrical dimension implies the Dushnik-Miller dimension is also finite (cf. Proposition \ref{before teo II} and Corollary \ref{teo II}).
\item We have related the geometrical and Debreu dimensions. In particular, we have shown the Debreu dimension is countable if and only if the same holds for the geometrical dimension (cf. Corollary \ref{why deb dim}), proving, hence, the new definition satisfies our aim of being closer to geometry than the Dushnik-Miller dimension. The relation between the Debreu and geometrical dimensions is, nonetheless, more involved. As a matter of fact, although the contrary holds, there exist partial orders with finite geometrical dimension and countably infinite Debreu dimension (cf. Theorem \ref{finite geo infinite deb}). Hence, although our new notion is closer to the geometrical dimension, the gap between these two can still be infinite. However, in some well-behaved cases, we can actually recover the equivalence (cf. Theorem \ref{equi finite geo and deb}).
\end{enumerate}

\subsubsection{Linear extensions of partial orders}

The usual approach to the dimension of partial orders relies on the Szpilrajn extension theorem (see \cite[Theorem 1]{szpilrajn1930extension} or \cite[Theorem 2.3.2]{harzheim2006ordered}), which can be stated as follows:
\begin{theo}[Szpilrajn {\cite{szpilrajn1930extension}}]
\label{szpilrajn thm}
If $(X,\preceq)$ is a partial order, then it has a linear extension $(X,\preceq')$.
\end{theo}

(Recall that by \emph{linear extension} we mean a total partial order that contains the relations given by another partial order, the one it \emph{extends}.)

The standard proof of Theorem \ref{szpilrajn thm} consists of two steps. In the first step, it is shown that (not necessarily linear) extensions of $(X,\preceq)$ exist. In the second step, the set
\begin{equation*}
    U \coloneqq \Big\{(X,\preceq') \Big| \preceq' \text{ is a partial order that extends } \preceq \Big\}
\end{equation*}
is considered. Taking the set inclusion as order structure, it is easy to see that $(U,\subseteq)$ satisfies the conditions in Zorn's lemma \cite{zorn1935remark}, that is, that each chain $(C_i)_{i \in I} \subseteq U$ has an upper bound, namely $\cup_{i \in I} C_i$. Hence, by Zorn's lemma, $(U,\subseteq)$ has a supremum, which ought to be linear.

As the previous paragraph details, the existence of linear extensions for arbitrary partial orders usually relies on a non-constructive method (given that Zorn's lemma follows from the axiom of choice). This is different in our work, where we avoid using the axiom of choice by constructively generating linear extensions as the result of a sequential process of extending partial orders to relate pairs of elements that where incomparable previously. In general, this procedure is clearly not feasible, since it follows from the fact the partial orders that interest us here fulfill some sort of countability restriction (like Debreu upper separability or the existence of a countable multi-utility). This distinction is particularly important in Theorem \ref{count mu has D-ext} and \cite[Proposition 5]{hack2022geometrical}, given that their generalizations where no countability restriction is assumed, Theorem \ref{szpilrajn thm} and \cite[Theorem 2.3.3]{harzheim2006ordered}, respectively, rely on the axiom of choice. Proposition \ref{before teo II} and some other results we included in \cite{hack2022geometrical} 
also exemplify this improvement.
The single other constructive approach to linear extensions that we are aware of was proposed in \cite{ok2002utility}. The issue with this approach from our perspective is that, as we discussed in \cite{hack2022geometrical}, it fails, in general, to produce Debreu separable linear extensions.

\subsubsection{Generalizations of a fundamental representation theorem by Debreu}

One of the fundamental connections between order structures and the real line is given by Theorem \ref{debreu thm}, where the existence of a utility functions is equated to that of a countable density property for total preorders. In fact, this result is key in the establishment of the (real-valued) optimization principles that underlie several ares (like learning systems). In this work, we have extended this fundamental theorem in the following ways:
\begin{enumerate}[label=(\roman*)]
\item We have defined Debreu upper separability, a suitable extension of Debreu separability for preorders where incomparabilities arise, and injective monotones, a type of monotone, that comes closer to a utility function for non-total preorders than the already known notion of strict monotone.
Regarding extensions of Theorem \ref{debreu thm}, we have shown that Debreu upper separable preorders have injective monotones (in fact, countable multi-utilities). However, we have also proven that the converse is false, that is, that there exist preorders with injective monotones that are not Debreu separable (like majorization, see \cite{hack2022representing}). In fact, an analogous relation holds between the weaker assumption that a preorder has a countable multi-utility and the existence of injective monotones (see \cite[Propositions 5 and 8]{hack2022representing}) .
\item As the previous point suggests, we can think of the equivalence in Theorem \ref{debreu thm} as being between countable multi-utilities and Debreu separability (for total preordered spaces). In this regard, we have extended Theorem \ref{debreu thm} to non-total preorders by showing that countable multi-utilities exist if and only if countable Debreu separable realizers do (cf. Corollary \ref{why deb dim}). Nonetheless, we have also proven that this equivalence fails for the finite case, that is, that there exist preorders with finite multi-utilities and only countably infinite Debreu separable realizers (cf. Theorem \ref{finite geo infinite deb}). Lastly,
we have recovered the equivalence in some well-behaved cases.
\item A fundamental fact behind Theorem \ref{debreu thm} is that, for total preorders $(X,\preceq)$, both conditions imply that the set of jumps is countable, that is, the countability of the set of pairs $x,y \in X$ such that $x \prec y$ and there is no $z \in X$ fulfilling $x \prec z \prec y$. We have shown that this fact also holds in the more general framework of conditionally connected Debreu upper separable preorders (see \cite[Proposition 13]{hack2022relation}). 
\end{enumerate}


\subsubsection{The role of order structures an their real-valued representations in physics}

Aside from the uncertainty preorder in the information-processing approach to decision-making,
several branches of physics, including thermodynamics \cite{ruch1976principle}, quantum physics \cite{muller2016generalization} and relativity \cite{bombelli1987space}, rely on order structures as their fundamental transition model. In fact, even the stronger order structure in domain theory has found applications in these areas \cite{coecke2010partial,martin2012spacetime}. Once the fundamental transitions are fixed, questions regarding how well they can be described by measurement devices arise naturally. It is here where we have tried to illustrate possible applications of the more general framework. In particular, in the study of the complexity of transition systems. An example of this is our contribution to the work of Ruch and others \cite{ruch1976principle,mead1977mixing} on thermodynamic transitions. More specifically, we have shown that, when trying to emulate the role of the second law of thermodynamics for either the principle of increasing mixing character or the principle of decreasing mixing distance, a countably infinite number of functions is needed (see Theorems \ref{dim majo} and \ref{d-majo dim}). Hence, as an application of our more abstract picture, we have encountered a considerable disadvantage of this proposal. Given their structural equivalence, the same considerations hold for another application in physics, namely the resource theory of entanglement, which we addressed briefly in Section \ref{subsec: resource th}.
Several of the applications in physics present issues to be solved. An example of this would be the geometrical dimension of entanglement catalysis, which remains to be determined.

\subsection{Uncertainty and computation}

The reader should note that we touched on several of the discussion points concerning uncertainty and computation in Section \ref{discu: real repre}. Hence, we only make a few remarks here.

\subsubsection{Formalizing computability on uncountable sets}

Computability on a set as common as the real numbers has not been formalized in a satisfactory manner yet \cite{weihrauch2012computability}. As a consequence, several proposals have been put forward, like the one based on order-theoretic tools that belongs to domain theory. Far from addressing this question in general, we have tried to clarify and simplify previous order-theoretic approaches with the aim of contributing to the debate by providing a clearer picture of one of the most prominent approaches. As we have argued in Section \ref{sub: difference comput}, it turns out that one can find several justifications for the stronger restrictions imposed by domain theory, which do not seem to naturally follow from a naive intuitive approach like the one we attempted in Sections \ref{intuition compu} and \ref{bisection ex}.

\subsubsection{Other models of computability on uncountable sets}

Aside from domain theory, several other approaches to computability on uncountable sets can be found (see, for example, \cite[Section 2.2]{stoltenberg2008computability}). Among them, Type-2 Theory of Effectivity (TTE) is arguably the most widely studied \cite{weihrauch2012computable,kreitz1985theory,weihrauch2000computability}. To form an impression of how other approaches may look like, it is worth noting its basic distinction with respect to domain theory. In particular, in the latter, computability is inherited from Turing machines while, in the former, it relies on a variation of the Turing machine: the Type-2 machine \cite{weihrauch1995simple}. As in the case of the order-theoretic approaches, TTE is based on a partial surjective map $\delta: \Sigma^\omega \rightarrow X$ that translates computability from $\Sigma^\omega$ (i.e. the set of infinite strings over some countable alphabet $\Sigma$) to an uncountable set of interest $X$, with computability on $\Sigma^\omega$ being defined through a Type-2 machine. Despite their apparent differences, domain theory and TTE are in fact closely related (see, for example, \cite{blanck2008reducibility}). To conclude, let us briefly describe what a Type-2 machine is. In its simplest form, we can think of a Type-2 machine as a Turing machine consisting of two tapes: (i) an input tape in which one can read symbols starting with the first and go read the following symbols without being able to return to the previous, and (ii) an output tape in which one can write symbols starting in the first blank space and go write symbols in the following blank positions without being able to return to the previous.

\subsubsection{Computability and order density}

One of the appealing properties of domain theory is that it derives its computability properties from the well-established model on the natural numbers (i.e. Turing machines). In order to do so, it takes advantage of some order-theoretic structure that possesses a countability restriction which links it to Turing machines. As a result, we obtain some close relations between these restrictions and the order denseness properties that we discussed when considering the representation of preorders by real-valued functions, which is another manifestation of the interdisciplinary aim of this work. (We have discussed the specific connections more in detail in Section \ref{discu: real repre}.)

\subsection{Other related topics}

\subsubsection{What is computation?}

The classical notion of computation is the one concerned with the finite description of objects, as pioneered by Turing. Thus, this approach is mainly focused on computability.
However, there is an ongoing discussion about what a \emph{computation} actually is and which systems perform them. For example, some authors \cite{mitchell2011ubiquity,denning2007computing} consider the adaptation of natural systems to the conditions imposed by their environments a \emph{(natural) computation}.
We can think of adaptation in a natural system as a series of transitions which occur between the different possible configurations of the system and are determined by both the environment and the system's nature.
Alternatively, in the case of biological systems, we can consider adaptation as the set of internal changes that aim to improve the profit made by the system from the environment. That is, we can conceptualize adaptation as the actions taken to bring the system to better states, where \emph{better} is imposed by the environment. Hence, we can think of the system's evolution as a series of transitions between states, where we assume the system always tends towards better states, although it may be unable, in general, to immediately translate to the optimal one.

Given the plethora of phenomena it may potentially cover, attempting a definition of computation that includes the two different notions in the last paragraph proves to be challenging. However, we can use computation as a link between the order structures we have studied in the context of real-valued representations and those used in computability theory.
In particular, the first one is related to the non-classical notion of computation and the second to the the classical one. Hence, instead of defining computation, we admit the existence of the two different notions in the above paragraph and devote the use the following section to connect the different applications of order theory we have considered by distinguishing \emph{abstract} from \emph{concrete} computation.

\subsubsection{Order theory in abstract and concrete computation}

The question regarding what
computation is or, more specifically, what phenomena ought to be classified as computation,
belongs
to a long debate that reaches the contemporary field called \emph{philosophy of mind} \cite{campbell2021does,piccinini2010computation,copeland1996computation}. In these studies, instead of a definition of computation, the subject is approached through the fundamental dichotomy that involves (i) an abstract computational mechanism $(C,\preceq_C)$, like a Turing machine, and (ii) a physical system $(P,\preceq_P)$ where such a mechanism can be implemented. The usual approach relies on some mapping $\phi: C \to P$ from the computational states in $C$ to the physical ones in $P$, where $\phi$ is required to relate $\preceq_C$, the computational transitions in $C$, to those imposed by $\preceq_P$, the dynamics of $P$.

Here, we consider a computational mechanism as an input-output system. That is, we think of computations as a binary relation  $\preceq_C$ on some set $C$, where $c \preceq_C c'$ represents that a computation with input $c\in C$ has output $c'\in C$. The binary bit inverter logic gate or NOT gate, for example, can be represented by $(C_0,\preceq_{C_0})$, where $C_0 \coloneqq \{0,1\}$ and the relations $0 \preceq_{C_0} 1$ and $1 \preceq_{C_0} 0$ hold. In order to implement such a computational mechanism, we intent to use a physical system, which we can also consider as a binary relation $(P,\preceq_P)$, where $p \preceq_P p'$
represents that, if we initialize the physical system in state $p\in P$, it will evolve until state $p'\in P$.
For such a physical system to be computationally reliable, the existence of certain time $T_P>0$  is fundamental. In particular, we ought to have that, whatever initial state $p_0 \in P$ we consider, the final state $p_T \in P$ with $T\geq T_P$ is well-defined, that is, the evolution reaches some final state after time $T_P$. Hence, we say $p \preceq_P p'$ for a pair $p,p' \in P$ if, when initiated at $p$, the evolution of the physical system $P$ reaches the state $p'$ before time $T_P$ and remains there for all $T' \geq T_P$. As an example,
we can consider a gas in a box $(P_0,\preceq_{P_0})$. In particular, we can take $P_0 \coloneqq \{h,f\}$, where $h$ represents the state where the gas only occupies the left half of the box (prepared using an extra wall) and $f$ represents the state where the gas occupies the full box. Given the tendency of any gas to 
occupy the full volume at its disposal, the transitions in this example are given by $h \preceq_{P_0} f$ and $f \preceq_{P_0} f$.

 We think of
$(C,\preceq_C)$, for example, as the transitions from, say, a polynomial (the input) to one of its roots (the output) that can be achieved using
the bisection method. Hence, we think of the structure of $(C,\preceq_C)$ as part of our fundamental study of the order-theoretic approach to computability. Moreover, we can consider the characterization of preorders by real-valued functions as an attempt to understanding the transitions in several physical systems, given that it is actually with that goal in mind that they are studied in several areas of physics.
Since knowing its fundamental transitions is indispensable in order to use a physical system to implement a computation, we can consider that part of our work as a contribution to concrete computation. In fact, aside from the fundamental nature of the question, a great deal of interest in the non-classical notion of computation comes precisely from the desire to take advantage from the adaptation of biological systems in order to implement computational mechanisms more effectively.

Now that we have clarified what $(C,\preceq_C)$ and $(P,\preceq_P)$ are, we can address $\phi$. Assume we intend to use some $(P,\preceq_P)$ to implement $(C,\preceq_C)$. In order to do so, we need to identify the states in $C$ with those in $P$ through a bijection $\phi:C \to P$. Note it is fundamental for this map to be fixed, since the intention is that, to implement $\preceq_C$ on some $c \in C$, we prepare the physical system in the state $\phi(c)$, wait $T_P$ units of time and measure the state of the system $p_{T_P} \in P$. Lastly, we take $\phi^{-1}(p_{T_P})$ as the result of running the computation with $c$ as input. In this scenario, the automation of both the preparation and measurement process are fundamental. Moreover, it is not enough for $\phi$ to be a bijection.
$\phi$ should relate the transitions in $C$ with the dynamics of $P$ closely, with the physical system being able to implement the computations in $C$ and not allowing transitions outside of $(C,\preceq_C)$.
In short, in order for $(P,\preceq_P)$ to be able to implement $(C,\preceq_C)$, we ought to have
an \emph{order isomorphism} between them.


It should be noted that, although the digital computer has established itself as the predominant physical device where computations are carried out, several other systems have been proposed in the area known as \emph{unconventional computing} (like biological \cite{adamatzky2018towards,benenson2012biomolecular} or quantum \cite{de2007fundamentals} systems, to name a couple). A simple illustration of this would be the use of a thermodynamic system, like a gas in a box, as a device that allows us to sample from
Boltzmann distributions (taking the potential $f$ as the input and the Boltzmann distribution associated to $f$ as the output). In fact, this idea of extracting randomness from a physical process is the basis of the \emph{hardware random number generators}. This may be of interest, in particular, as an extension of the concept of finite instructions that occupied Turing. In fact, in the classical picture of computability, it is not possible to sample from a distribution. In particular, the usual way in which this is done is through \emph{pseudo-random number generators}. Extending the definition to include thermodynamic systems, however, we can reliably sample from these distributions.

\subsubsection{Pancomputationalism} 

As a last remark, it should be noted that, although the picture we presented in the previous paragraph regarding the implementation of computational mechanisms seems appropriate, it has been somewhat criticized in the past. In fact, several points of view have been raised regarding the nature of the mapping between $(C,\preceq_C)$ and $(P,\preceq_P)$, that is, concerning the class of computations that can be associated to the dynamics of certain physical system. Among them, a radical perspective on this matter is \emph{pancomputationalism}, which has even suggested that the notion of computation is trivial, that is, that any physical system with sufficiently complex dynamics implements any computational mechanism \cite{putnam1988representation,hemmo2022multiple}. This is also known as the \emph{multiple computation theorem} \cite{hemmo2019physics,hemmo2022multiple}. Although a philosophical debate regarding pancomputationalism lies outside the scope of this work, let us simply point out that, in general, a clear shortcoming of this theory is the absence of both preparation and measurement, which are key in any study of computation motivated by practice, given that they allow the automation of computation (i.e. we can reliably initialize a physical process that implements our abstract computation and reliably measure its outcome).

\section{Conclusion}

This work belongs to the study of uncertainty, which is key in several scientific disciplines. In particular, we have taken a coarse-grained approach to uncertainty reduction, equating it to decision-making. This has led us to a somewhat unified treatment where we addresses through the lens of order theory several phenomena like learning, optimization, thermodynamics and computation.

As a starting point, we considered the importance of uncertainty in the contemporary study of learning systems via bounded rationality with information-processing constraints. In this regard, we developed the uncertainty framework that deals with learning systems further. More specifically, we went beyond the usual study of optimization using uncertainty tools and improved on the relation between non-optimal learning and optimization by clarifying the applicability of fluctuation theorems on learning systems, which we then tested in the context of human sensorimotor learning. This part of our work is motivated by
the long term goal of understanding how biological systems adapt efficiently in order to implement similar mechanisms in artificial intelligence.

The following issue we addressed was that of uncertainty measurement. In particular, we noticed that the intuitive uncertainty picture given by the order structure known as majorization is usually abandoned in favor of some function like Shannon entropy, as one can see in both thermodynamics and the study of learning systems. This led us to question how transitions from an order structure to real-valued functions take place. Our original interest was focused on the existence of optimization principles that somewhat generalize the properties of the maximum entropy principle. To this end, we introduced injective monotones and related them to other properties of order structures. After that, we considered the substitution of order structures by real-valued functions more in general. In particular, we developed interest in measuring the complexity of order structures and its relation to optimization on them. This led us to establishing the classification of preorders in terms of real-valued monotones and to the definition of the Debreu dimension as a bridge between the different notions of complexity (or \emph{order dimension}) we encountered in the literature. This abstract approach to uncertainty in terms of order monotones proved to be useful in concrete scenarios where uncertainty is key. In particular, it enabled us to gain insight into both the maximum entropy principle and the complexity of the approach to molecular diffusion known as the \emph{principle of increasing mixing character}.
Furthermore, our approach is also applicable to the uncertainty-based quantum resource theories, as we briefly pointed out. This part of our work was concerned with the general issue of how the set of transitions of a system can be captured by functions, which is a fundamental research question in several areas of knowledge, like thermodynamics, general relativity, quantum physics, utility theory and others. More specifically, this part of our work is also inscribed in the study of appropriate notions of information, where the substitution of an (intuitive) order-theoretic model by functions is commonplace.  

Once we concluded our research regarding uncertainty measurement and its abstractions, we returned to the fundamental idea of uncertainty as an order structure. In this case, however, we were interested in the uncertainty-based approaches to computation, given that we can think of any computation process as a sequence of steps that gather information about the solution to a problem. That is, we considered computation as a decision-making process where some sort of \emph{finiteness} restrictions are also taken into account. As it turns out, like in the case of learning systems, previous studies with such a perspective on computation relied on an order structure to represent uncertainty and, moreover, to introduce a notion of computability, that is, to distinguish between elements having a finite description from those that do not. In fact, such frameworks have been important in the study of the (still open) problem concerning what a proper notion of computability on the real numbers is. We concluded our work by asking 
what properties such an order-theoretic approach to computation should have. In particular, we introduced a more general framework from an intuitive picture and we highlighted several appealing properties that are lost compared to the stronger framework. A key feature in these approaches comes from the inheritance of the computability properties from Turing machines, which is achieved via countability restrictions on the order structure. The last part of our work was concerned with these restrictions and, more precisely, with their relation to the more usual countability restrictions one encounters in the literature on
order structures more in general. In the big picture, the third part of our work is contained in the attempts to studying information and computation together as two attributes of the same entity.

Our work was partly motivated by the two notions of computation that we have encountered: (i) the one related to adaptation in natural system and (ii) the one concerned with the finite description of objects. Although we did not find a proper definition that encompasses both, we have highlighted the fact that order theory offers a framework in which both of them can be dealt with.
A unified approach to them remains, hence, a fundamental challenge for the future.

\newpage
\thispagestyle{empty}
\mbox{}
\newpage

\cleardoublepage
\phantomsection
\addcontentsline{toc}{chapter}{Bibliography}

\bibliographystyle{plain}
\bibliography{main.bib}

\newpage
\thispagestyle{empty}
\mbox{}
\newpage

\begin{appendices}

\newpage
\thispagestyle{empty}
\mbox{}
\newpage

\chapter{Manuscripts published at the time of submission}

\section{Representing preorders with injective monotones}

{\setlength{\parindent}{0cm}
Pedro Hack, Daniel A. Braun and Sebastian Gottwald. Representing preorders with injective monotones. Theory and Decision 93, pages 663–690, 2022.\newline

\textbf{DOI:} 10.1007/s11238-021-09861-w.\newline

\textbf{Status:} Published.\newline

\textbf{Journal:} Theory and Decision.\newline

\textbf{License:} Creative Commons Attribution 4.0 International License

(http://creativecommons.org/licenses/by/4.0/).\newline

The final publication is available at \href{https://link.springer.com/article/10.1007/s11238-021-09861-w}{https://link.springer.com/article/10.1007/s11238-021-09861-w}.
}

\includepdf[pages=-]{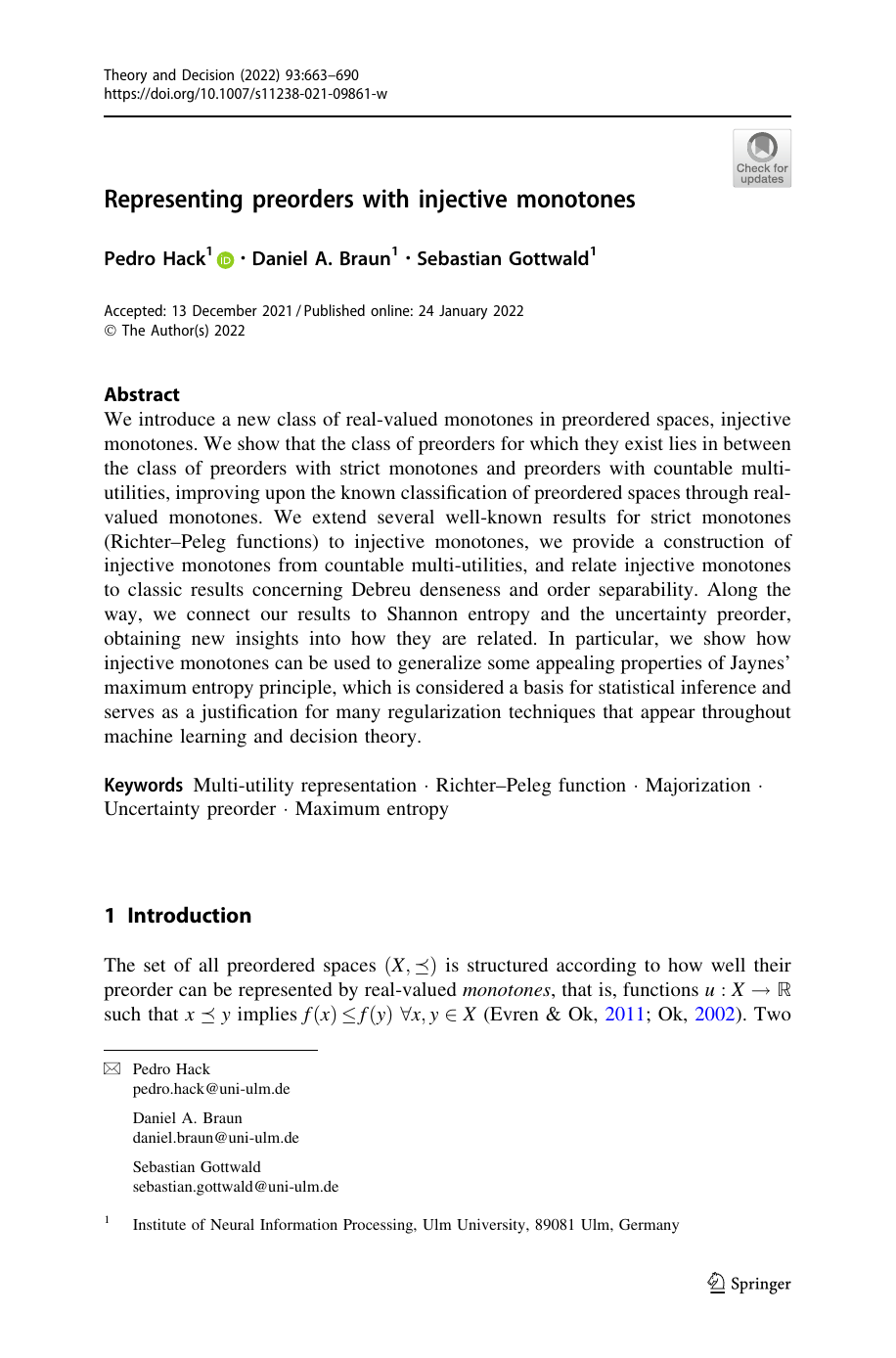}

\section{The classification of preordered spaces in terms of monotones:\\ Complexity and optimization}

{\setlength{\parindent}{0cm}
Pedro Hack, Daniel A. Braun and Sebastian Gottwald.
The classification of preordered spaces in terms of monotones: complexity and optimization. Theory and Decision 94, pages 693–720, 2023.\newline

\textbf{DOI:} 10.1007/s11238-022-09904-w.\newline

\textbf{Status:} Published.\newline

\textbf{Journal:} Theory and Decision.\newline

\textbf{License:} Creative Commons Attribution 4.0 International License

(http://creativecommons.org/licenses/by/4.0/).\newline

The final publication is available at \href{https://link.springer.com/article/10.1007/s11238-022-09904-w}{https://link.springer.com/article/10.1007/s11238-022-09904-w}.
}

\includepdf[pages=-]{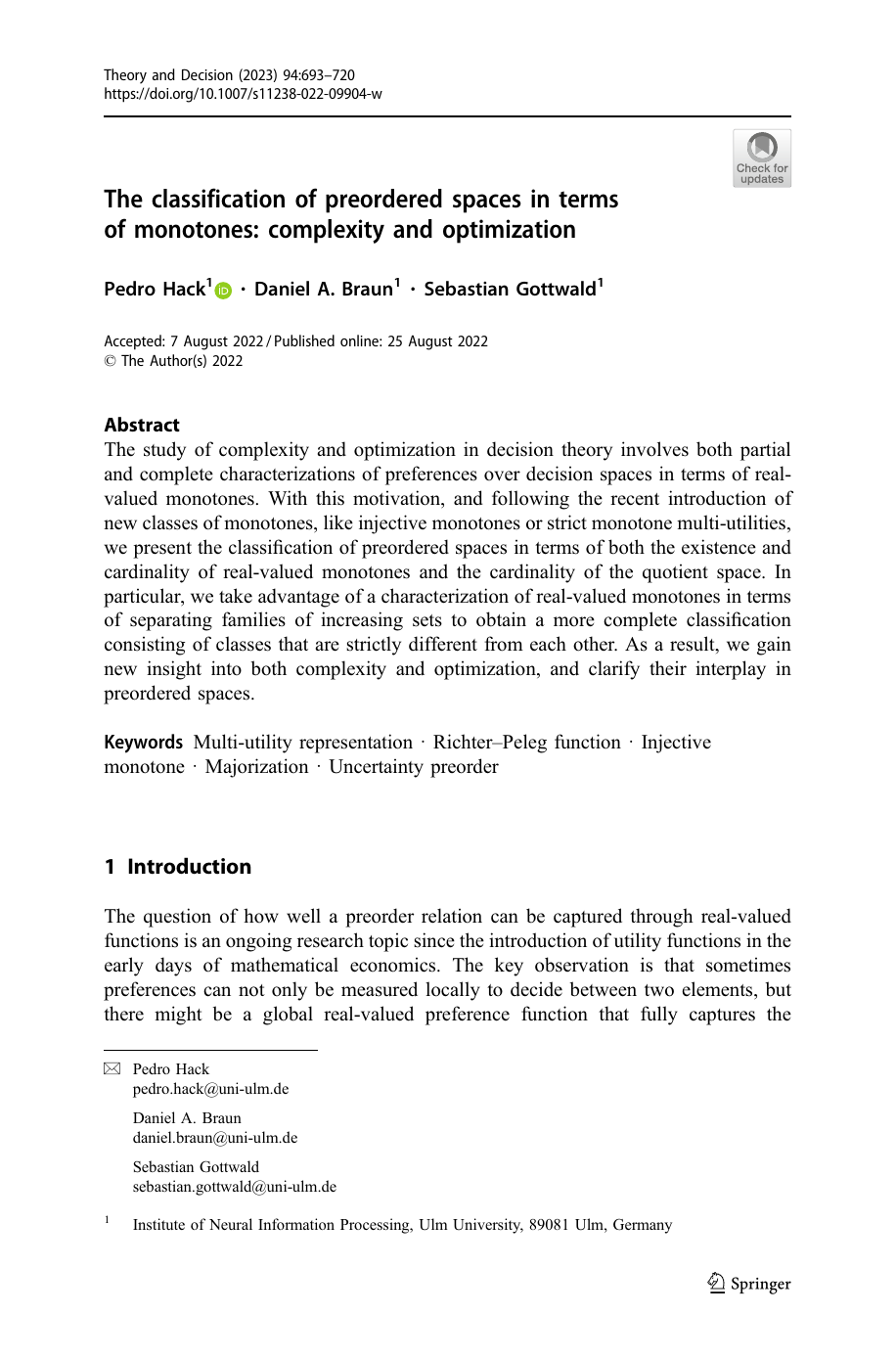}

\section{Thermodynamic fluctuation theorems govern human sensorimotor learning}

{\setlength{\parindent}{0cm}
Pedro Hack, Cecilia Lindig-Leon, Sebastian Gottwald and Daniel A. Braun.
Thermodynamic fluctuation theorems govern human sensorimotor learning.
Scientific Reports, 13(1):869, 2023.\newline

\textbf{DOI:} 10.1038/s41598-023-27736-8\newline

\textbf{Status:} Published.\newline

\textbf{Journal:} Scientific reports.\newline

\textbf{License:} Creative Commons Attribution 4.0 International License

(http://creativecommons.org/licenses/by/4.0/).\newline

The final publication is available at \href{https://www.nature.com/articles/s41598-023-27736-8}{https://www.nature.com/articles/s41598-023-27736-8}.
}

\includepdf[pages=-]{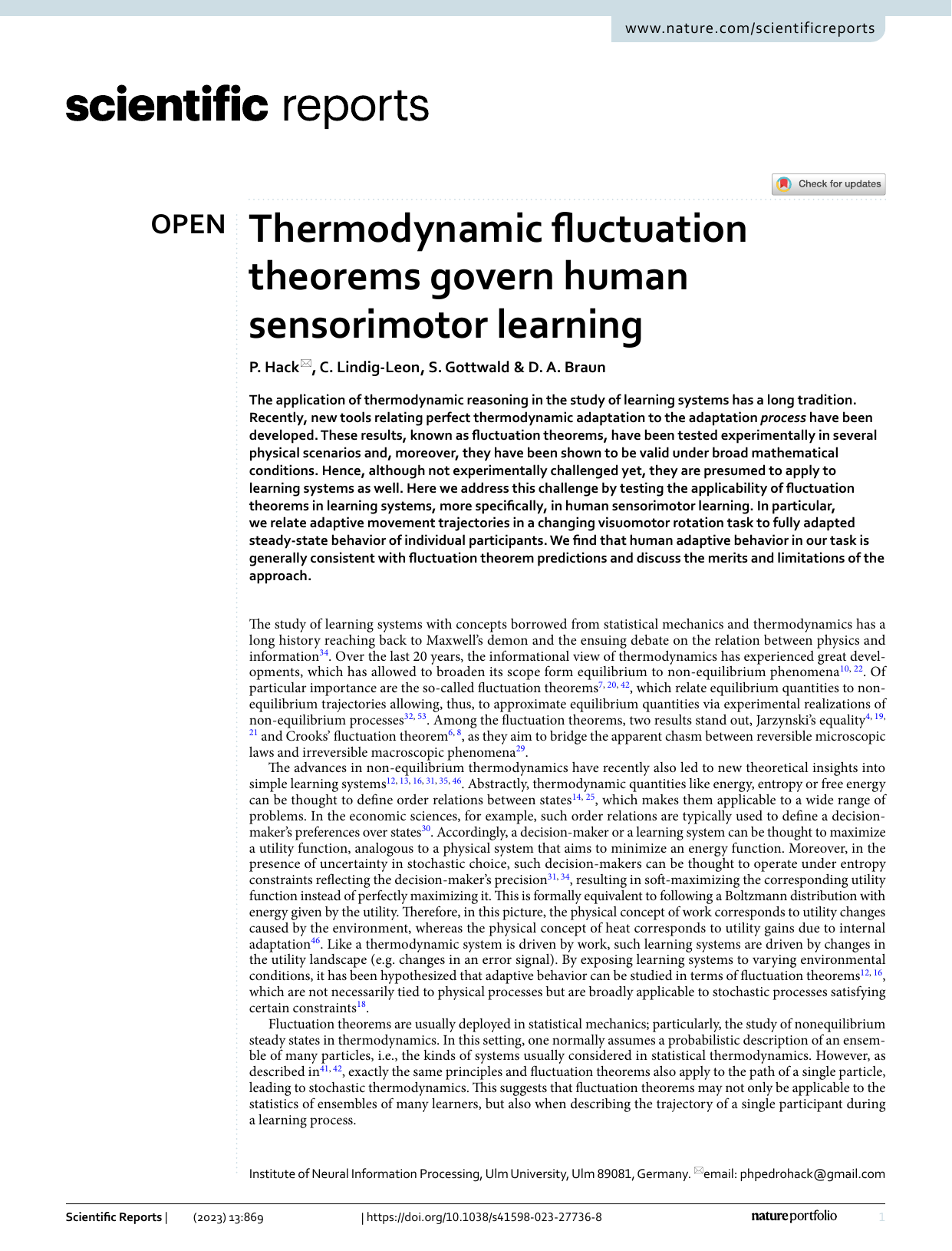}

\section{Jarzyski's equality and Crooks' fluctuation theorem for general Markov chains with application to decision-making systems}

{\setlength{\parindent}{0cm}
Pedro Hack, Sebastian Gottwald and Daniel A. Braun. Jarzyski’s equality
and Crooks’ fluctuation theorem for general Markov chains with application to
decision-making systems. Entropy, 24(12), 1731, 2022.\newline

\textbf{DOI:} 10.3390/e24121731.\newline

\textbf{Status:} Published.\newline

\textbf{Journal:} Entropy.\newline

\textbf{License:} Creative Commons Attribution 4.0 International License

(http://creativecommons.org/licenses/by/4.0/).\newline

The final publication is available at \href{https://www.mdpi.com/1099-4300/24/12/1731}{https://www.mdpi.com/1099-4300/24/12/1731}.
}

\includepdf[pages=-]{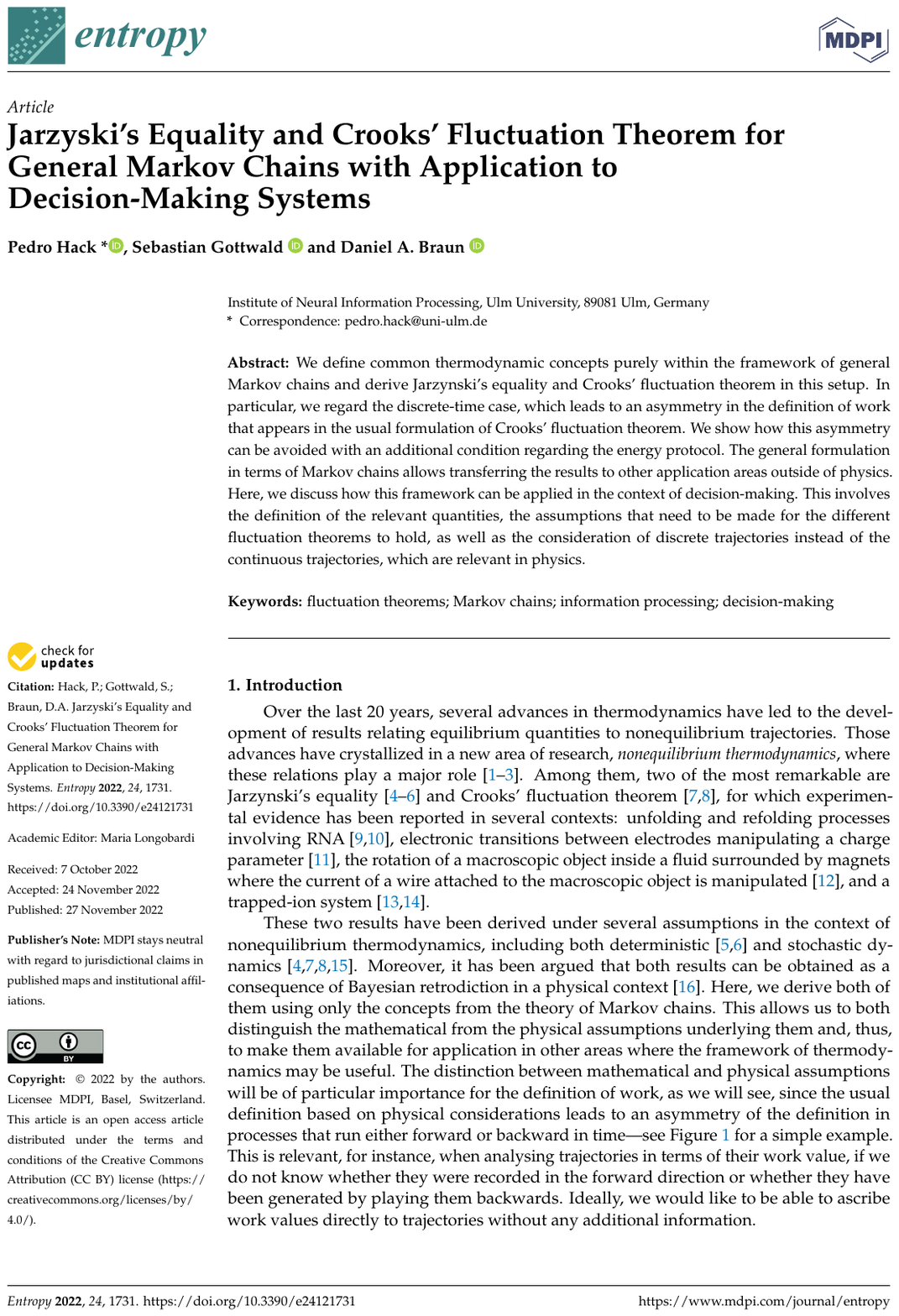}

\newpage
\thispagestyle{empty}
\mbox{}
\newpage

\chapter{Manuscripts under review at the time of submission}

\section{On a geometrical notion of dimension for
partially ordered sets}

{\setlength{\parindent}{0cm}
Pedro Hack, Daniel A. Braun, and Sebastian Gottwald. On a geometrical
notion of dimension for partially ordered sets. arXiv preprint arXiv:2203.16272,
2022.\newline

\textbf{DOI:} 10.48550/arXiv.2203.16272.\newline

\textbf{License:} arXiv.org perpetual, non-exclusive license

(https://arxiv.org/licenses/nonexclusive-distrib/1.0/license.html).
}

\includepdf[pages=-]{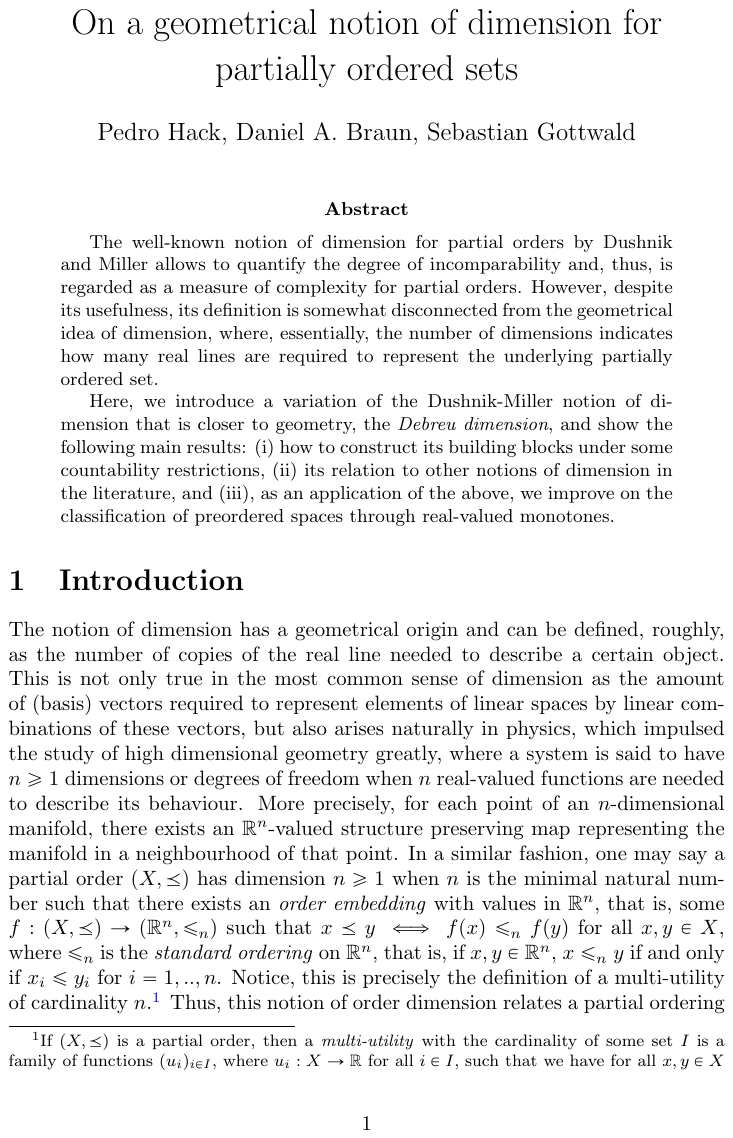}

\section{Majorization requires infinitely many second laws}

{\setlength{\parindent}{0cm}
Pedro Hack, Daniel A. Braun and Sebastian Gottwald. Majorization requires
infinitely many second laws. arXiv preprint arXiv:2207.11059, 2022.\newline

\textbf{DOI:} 10.48550/arXiv.2207.11059.\newline

\textbf{License:} arXiv.org perpetual, non-exclusive license

(https://arxiv.org/licenses/nonexclusive-distrib/1.0/license.html).
}

\includepdf[pages=-]{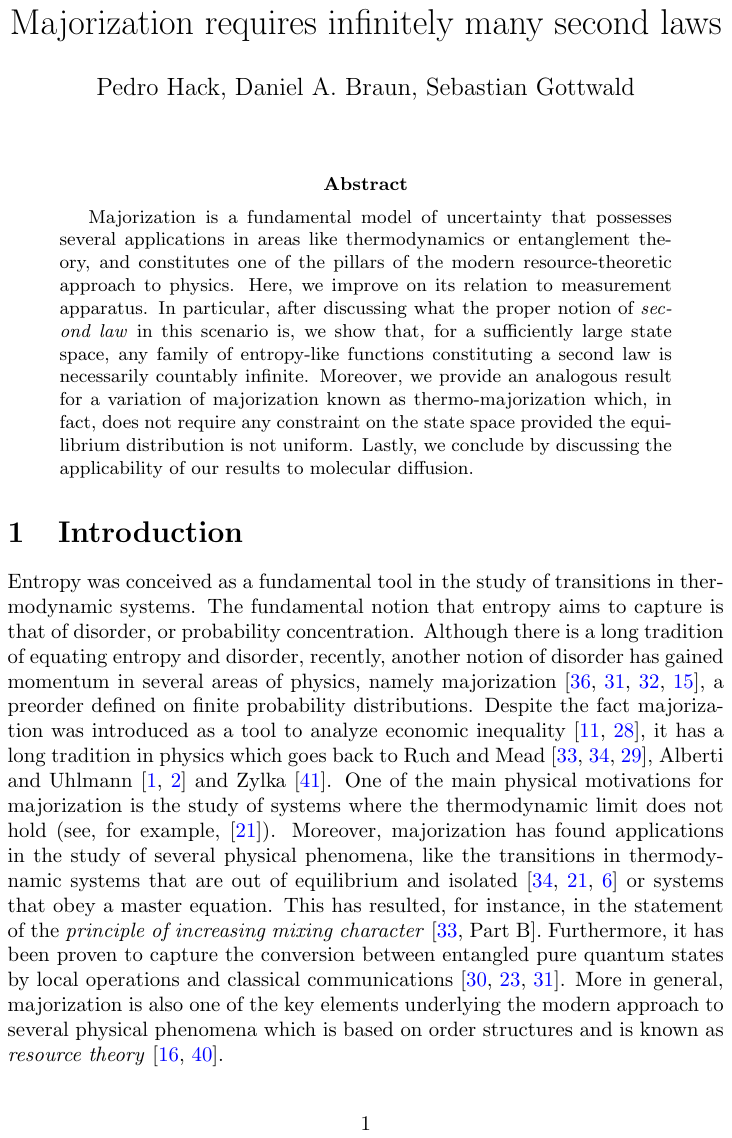}

\section{Computation as uncertainty reduction:\\
A simplified order-theoretic framework
}

{\setlength{\parindent}{0cm}
Pedro Hack, Daniel A. Braun and Sebastian Gottwald. Computation as
uncertainty reduction: a simplified order-theoretic framework. arXiv preprint
arXiv:2206.13885, 2022.\newline

\textbf{DOI:} 10.48550/arXiv.2206.13885.\newline

\textbf{License:} arXiv.org perpetual, non-exclusive license

(https://arxiv.org/licenses/nonexclusive-distrib/1.0/license.html).
}

\includepdf[pages=-]{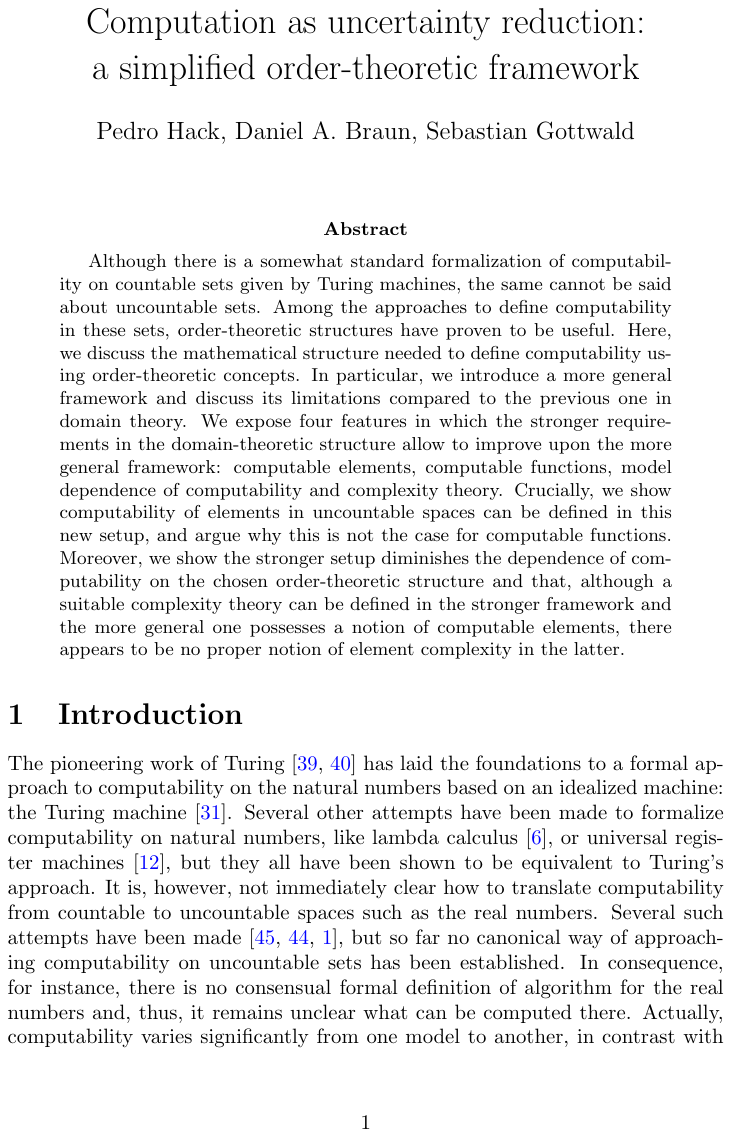}

\section{On the relation of order theory and computation in terms of denumerability}

{\setlength{\parindent}{0cm}
Pedro Hack, Daniel A. Braun and Sebastian Gottwald. On the relation of
order theory and computation in terms of denumerability. arXiv preprint
arXiv:2206.14484, 2022.\newline

\textbf{DOI:} 10.48550/arXiv.2206.14484.\newline

\textbf{License:} arXiv.org perpetual, non-exclusive license

(https://arxiv.org/licenses/nonexclusive-distrib/1.0/license.html).
}

\includepdf[pages=-]{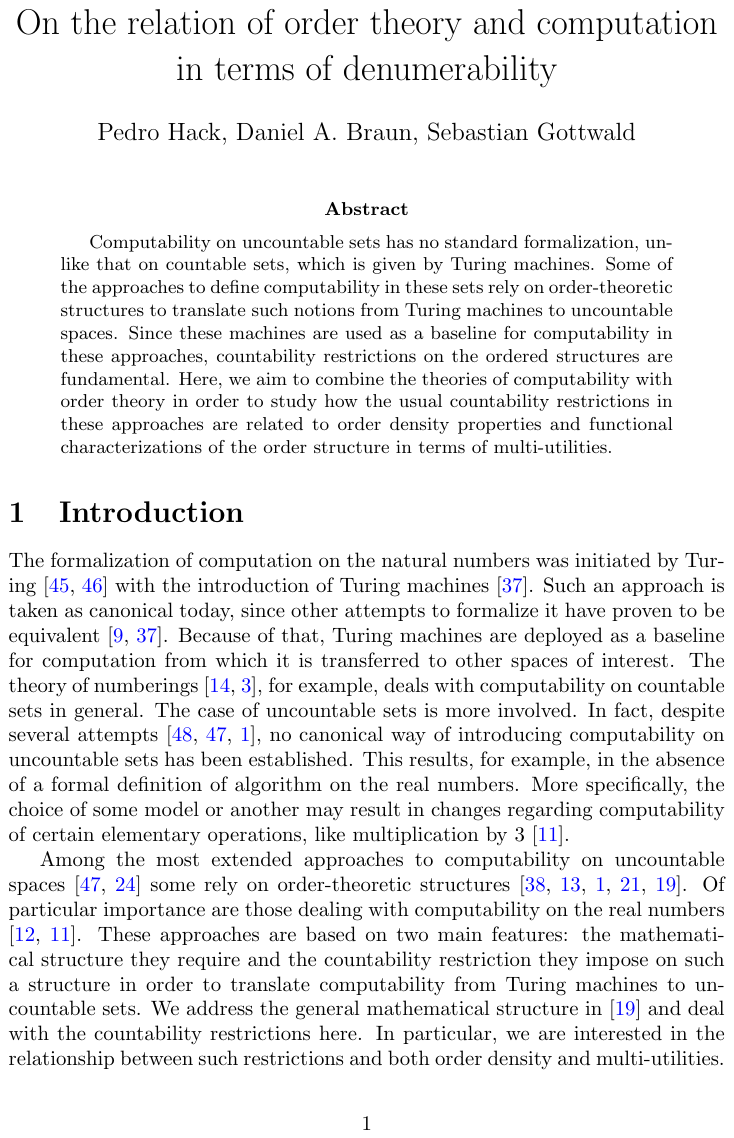}

\end{appendices}

\end{document}